\documentclass[letterpaper, 10 pt, conference]{ieeeconf}  % Comment this line out
\IEEEoverridecommandlockouts                              % This command is only
\usepackage{url,amsmath,amsfonts,amssymb, bm,bbm,xcolor, amsthm,algorithm,algorithmic,graphicx,subfigure,booktabs}
\usepackage{mathrsfs,mathtools}
\usepackage[noadjust]{cite}

\newtheorem{theorem}{Theorem}

\newtheorem{definition}{Definition}

\newtheorem{lemma}{Lemma}

\newtheorem{claim}{Claim}[section]

\newtheorem{assumption}{Assumption}

\theoremstyle{remark}
\newtheorem{remark}{Remark}

\newcommand{\rand}{{\textsc{rand}}}
\newcommand{\R}{{\mathbb R}}

\newcommand{\E}{{\mathbb E}}
\newcommand{\X}{{\mathcal X}}
\newcommand{\U}{{\mathcal U}}
\newcommand{\Y}{{\mathcal Y}}
\renewcommand{\P}{{\mathcal P}}
\newcommand{\Py}{{\mathbb P}}
\newcommand{\C}{\mathcal{C}}
\newcommand{\G}{\mathcal{G}}

\renewcommand{\H}{\mathcal{H}}
\newcommand{\vct}[1]{{#1}}

\newcommand{\rv}[1]{{{#1}}}

\newcommand{\gm}[1]{\mathscr{#1}}

\usepackage{subfigure}
\allowdisplaybreaks
\title{\LARGE \bf
Common Information Belief based Dynamic Programs for Stochastic Zero-sum Games with Competing Teams
}

\author{Dhruva Kartik, Ashutosh Nayyar and Urbashi Mitra% <-this % stops a space
\thanks{Dhruva Kartik, Ashutosh Nayyar and Urbashi Mitra are with the Department of Electrical and Computer Engineering, University of Southern California, Los Angeles, CA, USA. Email:
        {\tt\small \{mokhasun,ashutosn,ubli\}@usc.edu}. This research was supported by Grant ONR N00014-15-1-2550,
NSF CCF-1817200,
NSF ECCS 1750041,
NSF CCF-2008927,
ARO W911NF1910269,
Cisco Foundation 1980393,
ONR 503400-78050,
DOE DE-SC0021417,
Swedish Research Council 2018-04359
 and Okawa Foundation.} %
}

\begin{document}

\maketitle
\thispagestyle{empty}
\pagestyle{empty}

%%%%%%%%%%%%%%%%%%%%%%%%%%%%%%%%%%%%%%%%%%%%%%%%%%%%%%%%%%%%%%%%%%%%%%%%%%%%%%%%
\begin{abstract}
Decentralized team problems where players have asymmetric information about the state of the underlying stochastic system have been actively studied, but \emph{games} between such teams are less understood. We consider a general model of zero-sum stochastic games between two competing teams. This model subsumes many previously considered team and zero-sum game models. For this general model, we provide bounds on the upper (min-max) and lower (max-min) values of the game. Furthermore, if the upper and lower values of the game are identical (i.e., if the game has a \emph{value}), our bounds coincide with the value of the game.  Our bounds are obtained using two  dynamic programs based on a sufficient statistic known as the common information belief (CIB). We also identify certain information structures in which only the minimizing team controls the evolution of the CIB. In these cases,  we show that one of our CIB based dynamic programs can be used to find the min-max strategy (in addition to the min-max value). We propose an approximate dynamic programming approach for computing the values (and the strategy when applicable) and illustrate our results with the help of an example.
\end{abstract}

%%%%%%%%%%%%%%%%%%%%%%%%%%%%%%%%%%%%%%%%%%%%%%%%%%%%%%%%%%%%%%%%%%%%%%%%%%%%%%%%
\section{Introduction}\label{sec:intro}

In decentralized team problems, players collaboratively control a stochastic system to minimize a common cost. The information used by these players to select their control actions may be different. For instance, some of the players may have more information about the system state than others \cite{xie2020optimally}; or each player may have some private observations that are shared with other players with some delay \cite{nayyar2010optimal}. Such multi-agent team problems with an information asymmetry arise in a multitude of domains like autonomous vehicles and drones, power grids, transportation networks, military and rescue operations, wildlife conservation \cite{fang2019artificial} etc. Over the past few years, several methods have been developed to address decentralized team problems \cite{nayyar2013decentralized,oliehoek2016concise,rashid2018qmix,foerster2019bayesian,xie2020optimally}. However, \emph{games} between such teams are less understood. Many of the aforementioned systems are susceptible to adversarial attacks. Therefore, the strategies used by the team of players for controlling these systems must be designed in such a way that the damage inflicted by the adversary is minimized. {Such adversarial interactions can be modeled as zero-sum games between competing teams, and our main goal in this paper is develop a framework that can be used to analyze and solve them.}

The aforementioned works  \cite{nayyar2013decentralized,oliehoek2016concise,rashid2018qmix,foerster2019bayesian,xie2020optimally} on cooperative team problems solve them by first constructing an auxiliary single-agent Markov Decision Process (MDP). The auxiliary state (state of the auxiliary MDP) is the \emph{common information belief }(CIB). CIB is the belief on the system state and all the players' \emph{private} information conditioned on the \emph{common} (or public) information. Auxiliary actions (actions in the auxiliary MDP) are \emph{mappings} from agents' private information to their actions \cite{nayyar2013decentralized}. The optimal values of the team problem and this auxiliary MDP are identical. Further, an optimal strategy for the team problem can be obtained using any optimal solution of the auxiliary MDP with a simple transformation. The optimal value and strategies of this auxiliary MDP (and thus the team problem) can be characterized by dynamic programs (a.k.a. Bellman equations or recursive formulas). A key consequence of this characterization is that the CIB is a sufficient statistic for optimal control in team problems. We investigate whether a similar approach can be used to characterize values and strategies in zero-sum games between teams. This extension is not straightforward. In general games (i.e., not necessarily zero-sum), it may not be possible to obtain such dynamic programs (DPs) and/or sufficient statistics \cite{nayyar2017information,tang2021dynamic}. However, we show that for \emph{zero-sum} games between teams, the values can be characterized by CIB based DPs. Further, we show that for some specialized models, the CIB based DPs can be used to characterize a min-max strategy as well. A key implication of our result is that this CIB based approach can be used to solve several team problems considered before \cite{nayyar2013decentralized,xie2020optimally,foerster2019bayesian} even in the presence of certain types of adversaries.

{A phenomenon of particular interest and importance in team problems is \emph{signaling}. Players in a team can agree upon their control strategies \emph{ex ante}. Based on these agreed upon strategies, a player can often make inferences about the system state or the other players' private information (which are otherwise inaccessible to the player). This implicit form of communication between players is referred to as {signaling} and can be vital for effective coordination. While signaling is beneficial in cooperative teams, it can be detrimental in the presence of an adversary. This is because the adversary can exploit it to infer sensitive private information and inflict severe damage upon the system. A concrete example that illustrates this trade-off between \emph{signaling} and \emph{secrecy} is discussed is Section \ref{specialcases}. Our framework can be used optimize this trade-off in several stochastic games between teams.}

\paragraph{Related Work on Games}
 Zero-sum games between two individual players with asymmetric information have been extensively studied. In \cite{mertens2015repeated,rosenberg2004stochastic,renault2012value,gensbittel2014existence,li2016recursive,kartik2020upper}, stochastic zero-sum games with varying degrees of generality were considered and dynamic programming characterizations of the value of the game were provided. Various properties of the value functions (such as continuity) were also established and for some specialized information structures, these works also characterize a min-max strategy. Linear programs for computing the values and strategies in certain games were proposed in \cite{zheng2013decomposition,li2014lp}; and methods based on heuristic search value iteration (HSVI) \cite{smith2004heuristic} to compute the value of some games were proposed in \cite{horak2017heuristic,horak2019solving}. Zero-sum \emph{extensive form} games in which a team of players competes against an adversary have been studied in \cite{von1997team,farina2018ex,zhang2020computing}. Structured Nash equilibria in general games (i.e. not necessarily zero-sum) were studied in \cite{nayyar2014common,ouyang2017dynamic,vasal2019systematic} under some assumptions on the system dynamics and players' information structure. 
{A combination of reinforcement learning and search was used in \cite{rebel} to solve two-player zero-sum games. While this approach has very strong empirical performance, a better analytical understanding of it is needed.}
{Our work is closely related to \cite{kartik2019stochastic,kartik2020upper} and builds on their results. Our novel contributions in this paper over past works are summarized below.}

\paragraph{Contributions} (i) In this paper, we study a \emph{general} class of stochastic zero-sum games between two competing teams of players.  Our general model captures a variety of information structures including many previously considered stochastic team \cite{oliehoek2016concise,foerster2019bayesian,xie2020optimally} and zero-sum game models \cite{kartik2020upper,horak2017heuristic,horak2019solving} as well as game models that have \emph{not} been studied before. Our results provide a unified framework for analyzing a wide class of game models that satisfy some minimal assumptions. (ii) For our general model, {we adapt the techniques in \cite{kartik2020upper} to} provide bounds on the upper (min-max) and lower (max-min) values of the game. These bounds provide us with fundamental limits on the performance achievable by either team. Furthermore, if the upper and lower values of the game are identical (i.e., if the game has a \emph{value}), our bounds coincide with the value of the game.  Our bounds are obtained using two  dynamic programs (DPs) based on a sufficient statistic known as the common information belief (CIB). (iii) We also identify a subclass of game models in which only one of the teams (say the minimizing team) controls the evolution of the CIB. In these cases, we show that one of our CIB based dynamic programs can be used to find the min-max value as well as a min-max \emph{strategy}\footnote{{Note that this characterization of a min-max strategy is \emph{not} present in \cite{kartik2020upper}. A similar result for a \emph{very specific model with limited applicability} exists in \cite{kartik2019stochastic}. Our result is substantially more general than that in \cite{kartik2019stochastic}.}}. (iv) Our result reveals that the structure of the CIB based min-max strategy is similar to the structure of team optimal strategies. Such structural results have been successfully used in prior works \cite{foerster2019bayesian,rebel} to design efficient strategies for significantly challenging team problems. (v) Lastly, we discuss an approximate dynamic programming approach along with key structural properties for computing the values (and the strategy when applicable) and illustrate our results with the help of an example. 
% Our work is most closely related to \cite{kartik2020upper} which considers stochastic zero-sum games between two individual players. Our results on zero-sum games between teams are more general than those in \cite{kartik2020upper} and to the best of our knowledge, the first of their kind.

% Our results show that a wide class of game models can be analyzed in a unified way without requiring a case-by-case approach.

\paragraph{Notation}\label{notation}
Random variables are denoted by upper case letters, their realizations by the corresponding lower case letters. In general, subscripts are used as time index while superscripts are used to index decision-making agents. For time indices $t_1\leq t_2$, $\rv{X}_{t_1:t_2}$ is the short hand notation for the variables $(\rv{X}_{t_1},\rv{X}_{t_1+1},...,\rv{X}_{t_2})$. Similarly, $\rv{X}^{1:2}$ is the short hand notation for the collection of variables $(\rv{X}^1,\rv{X}^2)$.
%The indicator function of set $E$ is denoted by $\mathds{1}_{E}(\cdot)$, that is, $\mathds{1}_{E}(x) = 1$ if $x \in E$, and $0$ otherwise.
Operators $\Py(\cdot)$ and $\E[\cdot]$ denote the probability of an event, and the expectation of a random variable respectively.
For random variables/vectors $X$ and $Y$, $\Py(\cdot | \rv{Y}=y)$, $\E[\rv{X}| \rv{Y}=y]$ and $\Py(\rv{X} = x \mid \rv{Y} = y)$ are denoted by $\Py(\cdot | y)$, $\E[\rv{X}|y]$ and $\Py(x \mid y)$, respectively. 
For a strategy $g$, we use $\Py^g(\cdot)$ (resp. $\E^g[\cdot]$) to indicate that the probability (resp. expectation) depends on the choice of $g$. For any finite set $\mathcal{A}$, $\Delta\mathcal{A}$ denotes the probability simplex over the set $\mathcal{A}$.
For any two sets $\mathcal{A}$ and $\mathcal{B}$, $\mathcal{F}(\mathcal{A},\mathcal{B})$ denotes the set of all functions from $\mathcal{A}$ to $\mathcal{B}$. We define \textsc{rand} to be mechanism that given (i) a finite set $\mathcal{A}$, (ii) a distribution $d$ over $\mathcal{A}$ and a random variable $K$ uniformly distributed over the interval $(0,1]$, produces a random variable $X \in \mathcal{A}$ with distribution $d$, i.e.,
\begin{align}
X = \rand(\mathcal{A},d,K) \sim d.\label{randmech}
\end{align}

\section{Problem Formulation}
\label{probform}
Consider a dynamic system with two teams. Team 1 has $N_1$ players and Team 2 has $N_2$ players. The system operates in discrete time over a horizon\footnote{With a sufficiently large planning horizon, infinite horizon problems with discounted cost can be solved approximately as finite-horizon problems.} $T$. Let $\rv{X}_t \in \X_t$ be the state of the system at time $t$, and let $\rv{U}_t^{i,j} \in \U_t^{i,j}$ be the action of Player $j$, $j \in \{1,\dots,N_i\}$, in Team $i$, $i \in \{1,2\}$, at time $t$. Let 
\begin{align*}
    &U_t^1 \doteq \left(U_t^{1,1},\dots, U_t^{1,N_1}\right); & U_t^2 \doteq \left(U_t^{2,1},\dots, U_t^{2,N_2}\right),
\end{align*}
and $\mathcal{U}_t^i$ be the set of all possible realizations of $U_t^i$. {We will refer to $U_t^i$ as Team $i$'s action at time $t$.} The state of the system evolves in a controlled Markovian manner as
\begin{align}
\label{statevol}\rv{X}_{t+1} = f_t(\rv{X}_t, \rv{U}_t^1,\rv{U}_t^2,\rv{W}_t^s),
\end{align}
where $\rv{W}_t^s$ is the system noise. There  is an observation process $\rv{Y}_t^{i,j} \in \Y_t^{i,j}$ associated with each Player $j$ in Team $i$ and is given as
\begin{align}
\label{obseq}\rv{Y}_t^{i,j} = h_t^{i,j}(\rv{X}_t, \rv{U}_{t-1}^1,\rv{U}_{t-1}^2,\rv{W}_t^{i,j}),
\end{align}
where $\rv{W}_t^{i,j}$ is the observation noise. Let us define
\begin{align*}
    &Y_t^1 \doteq \left(Y_t^{1,1},\dots, Y_t^{1,N_1}\right);&Y_t^2 \doteq \left(Y_t^{2,1},\dots, Y_t^{2,N_2}\right).
\end{align*} We assume that the sets $\X_t$, $\U_t^{i,j}$ and $\Y_t^{i,j}$ are finite for all $i,j$ and $t$. Further,  the random variables $\rv{X}_1, \rv{W}_t^s, \rv{W}_t^{i,j}$ (referred to as \emph{the primitive random variables}) can take finitely many values and are mutually independent.
\begin{remark}
An alternative approach commonly used for characterizing system dynamics and observation models is to specify the transition and observation probabilities. We emphasize that this alternative characterization is equivalent to ours in equations \eqref{statevol} and \eqref{obseq} \cite{kumar2015stochastic}.
\end{remark}

\paragraph{Information Structure}\label{infostruct}
{At time $t$, Player $j$ in Team $i$ has access to \emph{a subset of all observations and actions generated so far}. Let $I^{i,j}_t$ denote the collection of variables (i.e. observations and actions) available to Player $j$ in team $i$ at time $t$.  Then  $ \rv{I}_t^{i,j} \subseteq \cup_{i,j}\{\rv{Y}^{i,j}_{1:t}, \rv{U}^{i,j}_{1:t-1}\} $. The set of all possible realizations of $\rv{I}_t^{i,j}$ is denoted by $\mathcal{I}^{i,j}_t$. Examples of such information structures include $I_t^{i,j} = \{Y_{1:t}^{i,j},U_{1:t-1}^{i,j}\}$ which corresponds to the information structure in Dec-POMDPs \cite{oliehoek2016concise} and $I_t^{i,j} = \{Y_{1:t}^{i,j},Y^{1:2}_{1:t-d},U_{1:t-1}^{1:2}\}$ wherein each player's actions are seen by all the players and their observations become public after a delay of $d$ time steps.}

Information $\rv{I}_t^{i,j}$ can be decomposed into \emph{common} and \emph{private} information, i.e. $\rv{I}_t^{i,j} = \rv{C}_t \cup \rv{P}_t^{i,j}$; common information $\rv{C}_t$ is the set of variables known to \emph{all} players at time $t$. The private information $P_t^{i,j}$ for Player $j$ in Team $i$ is defined as $\rv{I}_t^{i,j}\setminus C_t$. Let 
\begin{align*}
    &P_t^1 \doteq \left(P_t^{1,1},\dots, P_t^{1,N_1}\right); &P_t^2 \doteq \left(P_t^{2,1},\dots, P_t^{2,N_2}\right).
\end{align*}
We will refer to $P_t^i$ as Team $i$'s private information. Let $\mathcal{C}_t$ be the set of all {possible} realizations of common information at time $t$,  $\mathcal{P}_t^{i,j}$ be the set of all {possible} realizations of private information for Player $j$ in Team $i$ at time $t$ {and $\mathcal{P}_t^{i}$ be the set of all possible realizations of $P^i_t$}.  We make the following assumption on the evolution of common and private information. This is similar to Assumption 1 of  \cite{nayyar2014common,kartik2020upper}. 

\begin{assumption}\label{infevolve}
The evolution of common and private information available to the players is as follows:
(i) The common information $\rv{C}_t$ is non-decreasing with time, i.e. $\rv{C}_t \subseteq \rv{C}_{t+1}$. Let $\rv{Z}_{t+1} \doteq \rv{C}_{t+1}\setminus \rv{C}_t$ be the increment in common information. Thus, $\rv{C}_{t+1} = \{\rv{C}_t,\rv{Z}_{t+1}\}$. Furthermore,
\begin{align}
\label{commonevol}\rv{Z}_{t+1} = \zeta_{t+1}(\rv{P}_t^{1:2},\rv{U}_t^{1:2},\rv{Y}_{t+1}^{1:2}),
\end{align}
where $\zeta_{t+1}$ is a fixed transformation.
(ii) The private information evolves as
\begin{align}
\label{privevol}\rv{P}^{i}_{t+1} = \xi_{t+1}^{i}(\rv{P}_t^{1:2},\rv{U}_t^{1:2},\rv{Y}_{t+1}^{1:2}),
\end{align}
where $\xi_{t+1}^{i}$ is a fixed transformation {and $i=1,2$}.

\end{assumption}
As noted in \cite{nayyar2013decentralized,kartik2020upper}, a number of information structures satisfy the above assumption. Our analysis applies to any information structure that satisfies Assumption \ref{infevolve} including, among others, Dec-POMDPs and the delayed sharing information structure discussed above.

\paragraph{Strategies and Values}
Players can use any information available to them to select their actions and we allow behavioral strategies for all players. Thus, at time $t$, Player $j$ in Team $i$ chooses a distribution $\delta\rv{U}_t^{i,j}$ over its action space using a \emph{control law} $g_t^{i,j}:  \mathcal{I}_t^{i,j} \rightarrow \Delta\mathcal{U}_t^{i,j}$, i.e., $\delta\rv{U}_t^{i,j} = g_t^{i,j}(\rv{I}_t^{i,j}) = g_t^{i,j}(\rv{C}_t,\rv{P}_t^{i,j}).$
% \begin{equation}
% \delta\rv{U}_t^{i,j} = g_t^{i,j}(\rv{I}_t^{i,j}) = g_t^{i,j}(\rv{C}_t,\rv{P}_t^{i,j}). \label{eq:stg}
% \end{equation}
The distrubtion $\delta U^{i,j}_t$ is then used to randomly generate the control action $U^{i,j}_t$ as follows. We assume that player $j$ of Team $i$ has access to i.i.d. random variables $\rv{K}^{i,j}_{1:T}$ that are uniformly distributed over the interval $(0,1]$. These uniformly distributed variables are independent of each other and of the primitive random variables. The action $U^{i,j}_t$ is generated using $K^{i,j}_t$ and the randomization mechanism described in \eqref{randmech}, i.e.,
\begin{align}
    \label{rand1}\rv{U}_t^{i,j} = \rand(\mathcal{U}_t^{i,j},\delta U_t^{i,j},\rv{K}_t^{i,j}).
\end{align}
%\begin{equation}
%\rv{U}_t^i = \rand(g_t^i(\rv{I}_t^i),\rv{V}_t^i).
%\end{equation}

%\begin{remark}
%We will at times refer to $\delta \rv{U}^i_t$ as player $i$'s \emph{behavioral action} at time $t$.
%\end{remark}

The collection of control laws used by the players in Team $i$ at time $t$ is denoted by $g_t^i \doteq (g_t^{i,1},\dots,g_t^{i,N_i})$ and is referred to as the control law of Team $i$ at time $t$. Let the set of all possible control laws for Team $i$ at time $t$ be denoted by $\mathcal{G}_t^i$. The collection of control laws $\vct{g}^i \doteq (g_1^i,\dots,g_T^i)$ is referred to as the \emph{control strategy} of Team $i$, and the pair of control strategies $(\vct{g}^1,\vct{g}^2)$ is referred to as a \emph{strategy profile}. Let the set of all possible control strategies for Team $i$ be $\G^i$.

The total expected cost associated with a strategy profile $(\vct{g}^1,\vct{g}^2)$ is
\begin{align}
J(\vct{g}^1,\vct{g}^2)\doteq\E^{\left(\vct{g}^1,\vct{g}^2\right)}\left[\sum_{t=1}^T c_t(\rv{X}_t,\rv{U}_t^{1},\rv{U}_t^2)\right], \label{eq:totalcost}
\end{align}
where $c_t:\X_t \times \U_t^1 \times \U_t^2 \rightarrow \R$ is the cost function at time $t$.
%\begin{align}
%J_m(\delta\vct{g}^1,\delta\vct{g}^2)= \sum_{\hat{\vct{g}}^1}\sum_{\hat{\vct{g}}^2}\delta\vct{g}^1(\hat{\vct{g}}^1)J(\hat{\vct{g}}^1,\hat{\vct{g}}^2)\delta\vct{g}^2(\hat{\vct{g}}^2).
%\end{align}
Team 1 wants to minimize the total expected cost, while Team 2 wants to maximize it. We refer to this zero-sum game between Team 1 and Team 2 as Game $\gm{G}$. 
\begin{definition}\label{valdef}
The upper and lower values of the game $\gm{G}$ are respectively defined as
\begin{align}
\label{uppval}S^u(\gm{G}) &\doteq \min_{g^1 \in \mathcal{G}^1}\max_{g^2 \in \mathcal{G}^2} J(g^1,g^2),\\
\label{lowval}S^l(\gm{G}) &\doteq \max_{g^2 \in \mathcal{G}^2}\min_{g^1 \in \mathcal{G}^1} J(g^1, g^2).
\end{align}
If the upper and lower values are the same,  they are referred to as  the value of the game and denoted by $S(\gm{G})$. The minimizing strategy in \eqref{uppval} is referred to as Team 1's optimal strategy and the maximizing strategy in \eqref{lowval} is referred to as Team 2's optimal strategy\footnote{The strategy spaces $\G^1$ and $\G^2$ are compact and the cost $J(\cdot)$ is continuous in $g^1,g^2$. Hence, the existence of optimal strategies can be established using Berge's maximum theorem \cite{guide2006infinite}.}.
\end{definition}

A key objective of this work is to characterize the upper and lower values $S^u(\gm{G})$ and $S^l(\gm{G})$ of Game $\gm{G}$. 
To this end, we will define an \emph{expanded} virtual game $\gm{G}_e$. This virtual game will be used to obtain bounds on the upper and lower values of the original game $\gm{G}$. These bounds happen to be tight when the upper and lower values of game $\gm{G}$ are equal. For a sub-class of information structures, we will show that the expanded virtual game $\gm{G}_e$ can be used to obtain optimal strategies for one of the teams.

 \begin{remark}
 An alternative way of randomization is to use \emph{mixed strategies} wherein a player randomly chooses a deterministic strategy at the beginning of the game and uses it for selecting its actions. According to Kuhn's theorem, mixed and behavioral strategies are equivalent when players have perfect recall \cite{maschler2013game}. 
 \end{remark}
 {\begin{remark}[Independent and Shared Randomness]
 In most situations, the source of randomization is either privately known to the player (as in \eqref{rand1}) or publicly known to all the players in both teams. In this paper, we focus on independent randomization as in \eqref{rand1}.
 In some situations, a shared source of randomness may be available to all players in Team $i$ but not to any any of the players in the opposing team. Such shared randomness can help players in a team coordinate better. We believe that our approach can be extended to this case as well with some modifications.
 \end{remark}}

%\paragraph{Team Nash Equilibria}
We note that if the upper and lower values of game $\gm{G}$ are the same, then any pair of optimal strategies $(g^{1*},g^{2*})$ forms a \emph{Team Nash Equilibrium}\footnote{When players in a team randomize independently, Team Nash equilirbia may not exist in general \cite{anantharam2007common}.}, i.e., for every $\vct{g}^1 \in \G^1$ and $\vct{g}^2 \in \G^2$,
\begin{align*}
J(\vct{g}^{1*},\vct{g}^{2}) \leq J(\vct{g}^{1*},\vct{g}^{2*}) \leq J(\vct{g}^{1},\vct{g}^{2*}).
\end{align*}
In this case, $J(g^{1*},g^{2*})$ is the value of the game, i.e. $J(\vct{g}^{1*},\vct{g}^{2*}) = S^l(\gm{G}) = S^u(\gm{G}) = S(\gm{G}).$  Conversely, if a Team Nash Equilibrium exists, then the upper and lower values are the same \cite{osborne1994course}.

\section{Expanded Virtual Game $\gm{G}_e$}\label{expandedgame}
The expanded virtual game $\gm{G}_e$ is constructed using the methodology in \cite{kartik2020upper}. This game involves the same underlying system model as in game $\gm{G}$. The key distinction between games $\gm{G}$ and $\gm{G}_e$ lies in the manner in which the actions used to control the system are chosen. In game $\gm{G}_e$, 
all the players in each team of game $\gm{G}$ are replaced by a virtual player. Thus, game $\gm{G}_e$ has two virtual players, one for each team, and they operate as follows. 

\paragraph{Prescriptions} Consider virtual player $i$ associated with Team $i$, $i=1,2$. At each time $t$ and for each $j = 1,\dots,N_i$, virtual player $i$ selects a function $\Gamma^{i,j}_t$ that maps private information $P^{i,j}_t$ to a distribution $\delta \rv{U}_t^{i,j}$ over the space $\mathcal{U}_t^{i,j}$. Thus, $\delta U_t^{i,j} = \Gamma_t^{i,j}(P_t^{i,j})$. The set of all such mappings is denoted by $\mathcal{B}_t^{i,j} \doteq \mathcal{F}(\P_t^{i,j},\Delta \U_t^{i,j})$. We refer to the tuple $ \Gamma_t^i \doteq (\Gamma_t^{i,1},\dots,\Gamma_t^{i,N_i})$ of such mappings as virtual player $i$'s \emph{prescription} at time $t$. {The set of all possible prescriptions for virtual player $i$ at time $t$ is denoted by $\mathcal{B}_t^i \doteq \mathcal{B}_t^{i,1}\times\dots\times\mathcal{B}_t^{i,N_i}$.}
Once virtual player $i$ selects its prescription, the action $U^{i,j}_t$ is randomly generated according to the distribution $\Gamma_t^{i,j}(\rv{P}_t^{i,j})$. More precisely,
\begin{align}
\rv{U}_t^{i,j} &= \rand(\U_t^{i,j},\Gamma_t^{i,j}(\rv{P}_t^{i,j}),\rv{K}_t^{i,j}),
\end{align}
where the random variable $K^{i,j}_t$ and the mechanism $\rand$ are the same as in equation \eqref{rand1}.
% Virtual player $i$'s prescription induces a mapping between team $i$'s private information $\P_t^i$ and $\Delta \U_t^i$ as defined below. We will refer to this mapping as the \emph{behavioral form} of virtual player $i$'s prescription and denote it with $\Gamma_t^i$. 

%Note that $\mathcal{B}_t^i \subseteq \mathcal{F}(\P_t^i,\Delta \U_t^i)$ and is usually a strict subset because of the independent randomization in \eqref{randmech}.

%For a prescription $\gamma_t^i$, the probability of selecting an action $u_t^i$ when the private information is $p_t^i$ is denoted by $\gamma_t^i(p_t^i ; u_t^i)$.

%% Note that equations (\ref{virdyn1}-\ref{virdyn4}) describes a well-defined dynamic system.

\paragraph{Strategies} The virtual players in game $\gm{G}_e$ have access to the common information $\rv{C}_t$ and all the past prescriptions of both players, i.e., $\Gamma_{1:t-1}^{1:2}$. Virtual player $i$ selects its prescription at time $t$  using a control law $\tilde{\chi}_t^i$, i.e., $\Gamma_t^i = \tilde{\chi}_t^i(\rv{C}_t,\Gamma_{1:t-1}^{1:2}).$
Let $\tilde{\mathcal{H}}_t^i$ be the set of all such control laws at time $t$ and
$\tilde{\mathcal{H}}^i  \doteq \tilde{\mathcal{H}}_1^i \times \dots \times \tilde{\mathcal{H}}_T^i$ be the set of all control strategies for virtual player $i$. 
The total cost for a strategy profile $(\tilde{\vct{\chi}}^1,\tilde{\vct{\chi}}^2)$ is
\begin{align}
{\mathcal{J}}(\tilde{\vct{\chi}}^1,\tilde{\vct{\chi}}^2)=\E^{(\tilde{\vct{\chi}}^1,\tilde{\vct{\chi}}^2)}\left[\sum_{t=1}^T c_t(\rv{X}_t,\rv{U}_t^{1},\rv{U}_t^2)\right]. \label{eq:virtualJ}
\end{align}
The upper and lower values in $\gm{G}_e$ are defined as
\begin{align*}
S^u(\gm{G}_e) &\doteq \min_{\tilde{\chi}^1 \in \tilde{\mathcal{H}}^1}\max_{\tilde{\chi}^2 \in \tilde{\mathcal{H}}^2} {\mathcal{J}}(\tilde{\chi}^1,\tilde{\chi}^2)\\
S^l(\gm{G}_e) &\doteq \max_{\tilde{\chi}^2 \in \tilde{\mathcal{H}}^2}\min_{\tilde{\chi}^1 \in \tilde{\mathcal{H}}^1} {\mathcal{J}}(\tilde{\chi}^1, \tilde{\chi}^2).
\end{align*}

The following theorem establishes the relationship between the upper and lower values of the expanded game $\gm{G}_e$ and the original game $\gm{G}$. This result is analogous to Theorem 1 from \cite{kartik2020upper}.

\begin{theorem}[Proof in App. \ref{virtgameproof}]\label{origvirt}
The lower and upper values of the two games described above satisfy the following: $S^l(\gm{G})  \leq S^l(\gm{G}_e) \leq S^u(\gm{G}_e) \leq  S^u(\gm{G}).$
% \begin{align}
% \label{valuerel}S^l(\gm{G})  \leq S^l(\gm{G}_e) \leq S^u(\gm{G}_e) \leq  S^u(\gm{G}).
% \end{align}
Further, all these inequalities become equalities when a Team Nash equilibrium exists in Game $\gm{G}$.
\end{theorem}

\subsection{The Dynamic Programming Characterization}\label{dpsec}
We describe a methodology for finding the upper and lower values of the expanded game $\gm{G}_e$ in this subsection. {The results (and their proofs) in this subsection are similar to those in Section 4.2 of \cite{kartik2020upper}. However, the prescription spaces $\mathcal{B}_t^i$ in this paper are different (and more general) from those in \cite{kartik2020upper}, and thus our results in this paper are more general.} Our dynamic program is based on a sufficient statistic for virtual players in game $\gm{G}_e$ called the common information belief (CIB).
\begin{definition}
At time $t$, the common information belief (CIB), denoted by $\Pi_t$, is defined as the virtual players' belief on the  state and private information based on their information in game $\gm{G}_e$. Thus, for each $x_t \in \X_t, p_t^1 \in \P_t^1$ and $p_t^2 \in \P_t^2,$ we have
\begin{align*}
&\Pi_t(\vct{x}_t,\vct{p}_t^{1:2}) \doteq \Py\left[\rv{X}_t = \vct{x}_t,\rv{P}_t^{1:2} = \vct{p}_t^{1:2} \mid \rv{C}_t,\Gamma_{1:t-1}^{1:2}\right].
\end{align*}
The belief $\Pi_t$ takes values in the set $\mathcal{S}_t \doteq \Delta(\mathcal{X}_t \times \mathcal{P}^1_t \times \mathcal{P}^2_t)$.
\end{definition}
% The virtual players can compute the CIB using the information available to them, i.e. common information $C_t$ and prescription history $\Gamma_{1:t-1}^{1:2}$, without using their strategies. 
The following lemma describes an update rule that can be used to compute the CIB.

\begin{lemma}[Proof in App. \ref{infstateproof}]\label{infstate}
For any strategy profile $(\tilde{\vct{\chi}}^{1}, \tilde{\vct{\chi}}^{2})$ in Game $\gm{G}_e$, the common information based belief $\Pi_t$ evolves almost surely as
\begin{align}
\Pi_{t+1} = F_t(\Pi_t, \Gamma_t^{1:2},\vct{Z}_{t+1}), \label{eq:pieq}
\end{align}
where $F_t$ is a fixed transformation that does not depend on the virtual players' strategies. Further, the total expected cost can be expressed as
\begin{align}
{\mathcal{J}}(\tilde{\vct{\chi}}^1,\tilde{\vct{\chi}}^2)=\E^{(\tilde{\vct{\chi}}^1,\tilde{\vct{\chi}}^2)}\left[\sum_{t=1}^T \tilde{c}_t(\Pi_t,\Gamma_t^1,\Gamma_t^2)\right],
\end{align}
where the function $\tilde{c}_t$ is as defined in equation (\ref{tildec}) in Appendix \ref{infstateproof}.
\end{lemma}

%\begin{remark}\label{fremark}
%{Because \eqref{eq:pieq} is an almost sure equality, the transformation $F_t$ in Lemma \ref{infstate} is not necessarily unique. In Appendix \ref{infstateproof}, we identify a class of transformations such that for any transformation $F_t$ in this class, Lemma \ref{infstate} holds. We denote this class by $\mathscr{B}$.}
%\end{remark}

\paragraph{Values in Game $\gm{G}_e$} We now describe two dynamic programs, one for each virtual player in $\gm{G}_e$.
The minimizing virtual player (virtual player 1) in game $\gm{G}_e$ solves the following dynamic program. Define $V^u_{T+1}(\pi_{T+1}) = 0$ for every  $\pi_{T+1}$. In a backward inductive manner, at each time $t \leq T$ and for each possible  common information belief $\pi_t$ and prescriptions $\gamma^1_t, \gamma^2_t$, define the upper cost-to-go function $w_t^u$ and the upper value function $V_t^u$ as
\begin{align}
\label{uppercost}&w^u_t(\pi_t,\gamma_t^1,\gamma_t^2) \\
&\doteq \tilde{c}_t(\pi_t,\gamma_t^1,\gamma_t^2) + \E[V^u_{t+1}(F_t(\pi_{t},\gamma_t^{1:2},\rv{Z}_{t+1}))\mid \pi_t,\gamma_t^{1:2}],\nonumber\\
\label{minequa}&V_t^u(\pi_t) \doteq \min_{{\gamma}_t^1} \max_{\gamma_t^2}w_t^u(\pi_t,\gamma_t^1,\gamma_t^2).
\end{align}
The maximizing virtual player (virtual player 2) solves an analogous max-min dynamic program with a lower cost-to-go function $w_t^l$ and lower value function $V_t^l$ (See App. \ref{equiexistlemmaproof} for details).

\begin{lemma}[Proof in App. \ref{equiexistlemmaproof}]\label{equiexistlemma}
%Consider the min-max dynamic program in (\ref{minequa}). 
For each $t$, there exists a measurable mapping $\Xi^1_t: \mathcal{S}_t \rightarrow \mathcal{B}_t^1$ such that $V_t^u(\pi_t) = \max_{\gamma_t^2}w_t^u(\pi_t,\Xi_t^1(\pi_t),\gamma_t^2)$.
% \begin{align}
% V_t^u(\pi_t) &= \max_{\gamma_t^2}w_t^u(\pi_t,\Xi_t^1(\pi_t),\gamma_t^2).
% \end{align}
Similarly, there exists a measurable mapping $\Xi^2_t: \mathcal{S}_t \rightarrow \mathcal{B}_t^2$ such that $V_t^l(\pi_t) = \min_{\gamma_t^1}w_t^l(\pi_t,\gamma_t^1,\Xi_t^2(\pi_t))$.
% \begin{align}
% V_t^l(\pi_t) &= \min_{\gamma_t^1}w_t^l(\pi_t,\gamma_t^1,\Xi_t^2(\pi_t)).
% \end{align}

\end{lemma}

\begin{theorem}[Proof in App. \ref{dpproof}]\label{dp}
The upper and lower values of the expanded virtual game $\gm{G}_e$ are given by $S^u(\gm{G}_e) = \E[V_1^u(\Pi_1)]$ and $S^l(\gm{G}_e) = \E[V_1^l(\Pi_1)]$.
% \begin{align}
% &S^u(\gm{G}_e) = \E[V_1^u(\Pi_1)]; &S^l(\gm{G}_e) = \E[V_1^l(\Pi_1)].
% \end{align}
\end{theorem}

Theorem \ref{dp} gives us a dynamic programming characterization of the upper and lower values of the expanded game. As mentioned in Theorem \ref{origvirt}, the upper and lower values of the expanded game provide bounds on the corresponding values of the original game. If the original game has a Team Nash equilibrium, then the dynamic programs described above characterize the value of the game. 
%Note that this applies to any dynamic game of  the form in Section \ref{sec:probform} where the common information is non-decreasing in time and the private information has a ``state-like'' update equation (see Assumption \ref{infevolve}). As noted before, a variety of information structures satisfy this assumption \cite{nayyar2013decentralized}, \cite{nayyar2014common}.  

\paragraph{Optimal Strategies in Game $\gm{G}_e$} The mappings $\Xi^1$ and $\Xi^2$ obtained from the dynamic programs described above (see Lemma \ref{equiexistlemma}) can be used to construct optimal strategies for both virtual players in game $\gm{G}_e$ in the following manner.
{\begin{definition}\label{stratdef} Define strategies $\tilde{\chi}^{1*}$ and $\tilde{\chi}^{2*}$ for virtual players 1 and 2 respectively as follows: for each instance of common information $c_t$ and prescription history $\gamma_{1:t-1}^{1:2}$, let
\begin{align*}
&\tilde{\chi}^{1*}_t(c_t,\gamma_{1:t-1}^{1:2}) \doteq \Xi_t^1(\pi_t); &\tilde{\chi}^{2*}_t(c_t,\gamma_{1:t-1}^{1:2}) \doteq \Xi_2^1(\pi_t),
\end{align*}
where $\Xi_t^1$ and $\Xi_t^2$ are the mappings defined in Lemma \ref{equiexistlemma} and $\pi_t$ (which is a function of $c_t,\gamma_{1:t-1}^{1:2}$) is obtained in a forward inductive manner using the update rule $F_t$ defined in Lemma \ref{infstate}.
\end{definition}}

\begin{theorem}[Proof in App. \ref{dpproof}]\label{stratthm1}
The strategies $\tilde{\chi}^{1*}$ and $\tilde{\chi}^{2*}$ as defined in Definition \ref{stratdef} are, respectively, min-max and max-min strategies in the expanded virtual game $\gm{G}_e$.
\end{theorem}

\section{Only Virtual Player 1 controls the CIB}\label{oneplayercontrol}
In this section, we consider a special class of instances of Game $\gm{G}$ and show that the dynamic program in \eqref{minequa} can be used to obtain a min-max \emph{strategy} for Team 1, the minimizing team in game $\gm{G}$. The key property of the information structures considered in this section is that the common information belief $\Pi_t$ is controlled\footnote{Note that the players in Team 2 might still be able to control the state dynamics through their actions.} only by virtual player 1 in the corresponding expanded game $\gm{G}_e$. This is formally stated in the following assumption.

\begin{assumption}\label{onesidecontrolassum}
For any strategy profile $(\tilde{\vct{\chi}}^{1}, \tilde{\vct{\chi}}^{2})$ in Game $\gm{G}_e$, the CIB $\Pi_t$ evolves almost surely as
\begin{align}
\Pi_{t+1} = F_t(\Pi_t, \Gamma_t^{1},\vct{Z}_{t+1}),  \label{eq:pieq}
\end{align}
where $F_t$ is a fixed transformation that does not depend on the virtual players' strategies.
%\begin{align}
%\pi_{t+1}(x_{t+1},p_{t+1}^{1:2}) &= \frac{{P}^j_t(\pi_t,\gamma_t^{1:2},z_{t+1},x_{t+1},p_{t+1}^{1:2})}{\mathscr{P}_t^m(\pi_t,\gamma_t^{1:2},z_{t+1})}\\
%&= F_t(\pi_t, \gamma_t^{1:2},z_{t+1},x_{t+1},p_{t+1}^{1:2}),
%\end{align}

\end{assumption}
We will now describe some instances of Game $\gm{G}$ that satisfy Assumption \ref{onesidecontrolassum}. {We note that
two-player zero-sum games that satisfy a property similar to Assumption \ref{onesidecontrolassum} were studied in \cite{gensbittel2014existence}.}

\subsection{{Game Models} Satisfying Assumption \ref{onesidecontrolassum}}\label{infexample}
\paragraph{All players in Team 2 have the same information} Consider an instance of game $\gm{G}$ in which every player $j$ in Team 2 has the following information structure $I_t^{2,j} = \left\{Y_{1:t}^{2}, U_{1:t-1}^{2} \right\}$.
% \begin{align}
% I_t^{2,j} = \left\{Y_{1:t}^{2}, U_{1:t-1}^{2} \right\}, ~~ j = 1,\dots,N_2.
% \end{align}
Further, Team 2's information is known to every player in Team 1. Thus, the common information $C_t = I_t^{2,j}$. Under this condition, players in Team 2 do not have any private information. Thus, their private information $P_t^2 = \varnothing$. Any information structure satisfying the above conditions satisfies Assumption \ref{onesidecontrolassum}, see Appendix \ref{infproof1} for a proof. Since Team 1's information structure is relatively unrestricted, the above model subsumes many previously considered team and game models. {Notable examples of such models include: (i) all purely cooperative team problems in \cite{nayyar2013decentralized,dibangoye2016optimally,xie2020optimally,foerster2019bayesian}, and (ii) two-player zero-sum game models where one agent is more informed than the other \cite{zheng2013decomposition,gensbittel2014existence,horak2017heuristic}.}
% \cite{gensbittel2014existence,zheng2013decomposition,horak2017heuristic,xie2020optimally,dibangoye2016optimally}

\paragraph{Team 2's observations become common information with one-step delay}Consider an an instance of game $\gm{G}$ where the current private information of Team 2 becomes common information in the very next time-step. More specifically, we have $C_{t+1} \supseteq \{Y^2_{1:t}, U^2_{1:t}\}$ and for each Player $j$ in Team 2, $P^{2,j}_t = Y^{2,j}_t$.
% \begin{align}
% &P_t^{2,j} = Y^{2,j}_t; &Z_{t+1} \supseteq \{Y_t^{2},U_t^2\}.
% \end{align}
Note that unlike in \cite{zheng2013decomposition,horak2017heuristic}, players in Team 2 have some private information in this model. Any information structure that satisfies the above conditions satisfies Assumption \ref{onesidecontrolassum}, see Appendix \ref{infproof2} for a proof.

\paragraph{Games with symmetric information} Consider the information structure where $ \rv{I}_t^{i,j} = \cup_{i,j}\{\rv{Y}^{i,j}_{1:t}, \rv{U}^{i,j}_{1:t-1}\} $ for every $i,j$. All the players in this game have the same information and thus, players do not have any private information. Note that this model subsumes perfect information games. It can be shown that this model satisfies Assumption \ref{onesidecontrolassum} using the same arguments in Appendix \ref{infproof1}. In this case, the CIB is not controlled by both virtual players and thus, we can use the dynamic program to obtain both min-max and max-min strategies.

{In addition to the models discussed
above, there are other instances of $\gm{G}$ that
satisfy Assumption 2. These are included in
Appendix \ref{additionalinf}.}

\subsection{Min-max Value and Strategy in Game $\gm{G}$}\label{onesidedp}
\paragraph{Dynamic Program} 
Since we are considering special cases of Game $\gm{G}$, we can use the analysis in Section \ref{expandedgame} to write the min-max dynamic program for virtual player 1. Because of Assumption \ref{onesidecontrolassum}, the belief update $F_t(\pi_t,\gamma_t^{1:2},z_{t+1})$ in \eqref{uppercost} is replaced by $F_t(\pi_t,\gamma_t^{1},z_{t+1})$. Using Theorems \ref{dp} and \ref{stratthm1}, we can can conclude that the upper value of the expanded game $S^u(\gm{G}_e) = \E [V_1^u(\Pi_1)]$ and that the strategy $\tilde{\chi}^{1*}$ obtained from the DP is a min-max strategy for virtual player 1 in Game $\gm{G}_e$. An approximate dynamic programming based approach for solving the dynamic programs is discussed in Appendix \ref{dpsolve}. This discussion includes certain structural properties of the value functions that make their computation significantly more tractable.

{\paragraph{Min-max Value and Strategy}
The following results provide a characterization of the min-max value $S^u(\gm{G})$ and a min-max strategy $g^{1*}$ in game $\gm{G}$ under Assumption \ref{onesidecontrolassum}. Note that unlike the inequality in Theorem \ref{origvirt}, the upper values of games $\gm{G}$ and $\gm{G}_e$ are always equal in this case.}
{\begin{theorem}[Proof in App. \ref{strategyproof}]\label{onesidevalue}
Under Assumption \ref{onesidecontrolassum}, we have $S^u(\gm{G}) = S^u(\gm{G}_e) =  \E[V_1^u(\Pi_1)].$
\end{theorem}}
% Consider the strategy $g^{1*}$ for Team 1 in game $\gm{G}$ wherein each player $j$ at time $t$ uses the following procedure to obtain its behavioral action $\delta U_t^{1,j}$. Player $j$ first computes the common information belief $\pi_t$ based on the common information $c_t$ using the following recursive relation
% \begin{align}
% &\pi_1(x_1,p_1^{1:2}) \\
% \nonumber&= \Py[X_1 = x_1, P_t^1 = p_t^1,P_t^2 = p_t^2 \mid C_1 = c_1] ~ \forall \; x_1,p_1^1,p_1^2\\
% &\pi_{\tau + 1} = F_\tau(\pi_\tau, \Xi_\tau^1(\pi_\tau),z_{\tau+1}), ~ 1 \leq \tau < t,
% \end{align}
% where $\Xi_t^1$ is obtained from the dynamic program described above and $F_t$ is the update rule in Assumption \ref{onesidecontrolassum}. Player $j$ then obtains a prescription $\gamma_t^1$ as
% \begin{align}
% \chi_t^{1*}(\vct{c}_t) = \gamma_t^1  \doteq \Xi_t^1(\pi_t),\label{presdef}
% \end{align}
% where $\chi^{1*}$ is used to denote this mapping between the common information $c_t$ and the prescription $\gamma_t^1$. This behavioral prescription is then translated into its native form to obtain the tuple $(\gamma_t^{1,1},\dots,\gamma_t^{1,N_1})$. Finally, Player $j$ selects its behavioral action as
% \begin{align}
%     g_t^{1,j*}(c_t,p_t^{1,j}) \doteq \gamma_t^{1,j}(p_t^{1,j}).
% \end{align}

\begin{algorithm}[tb]
  \caption{Strategy $g^{1,j*}$ for Player $j$ in Team 1}
  \label{alg:example}
\begin{algorithmic}
%   \STATE {\bfseries Input:} data $x_i$, size $m$
%   \REPEAT
    \STATE Input: $\Xi^1_t(\pi)$ obtained from DP for all $t$ and all $\pi$
  \FOR{$t=1$ {\bfseries to} $T$}
  \STATE Current information: $C_t,P_t^{1,j}$
  \COMMENT{where $C_{t} = \{C_{t-1},Z_t\}$}
  \STATE Update CIB $\Pi_{t} = F_{t-1}(\Pi_{t-1}, \Xi_{t-1}^1(\Pi_{t-1}),Z_{t})$ \COMMENT{If $t=1$, Initialize CIB $\Pi_t$ using $C_t$}
  \STATE Get prescription $\Gamma_t^1 = (\Gamma_t^{1,1},\dots,\Gamma_t^{1,N_1}) = \Xi_t^1(\Pi_t)$ 
  \STATE Get distribution $\delta U_t^{1,j} = \Gamma_t^{1,j}(P_t^{1,j}) $ and select action $U_t^{1,j} = \rand(\mathcal{U}_t^{1,j},\delta U_t^{1,j},\rv{K}_t^{1,j})$

  \ENDFOR
%   \UNTIL{$noChange$ is $true$}
\end{algorithmic}
\end{algorithm}
{{\begin{theorem}[Proof in App. \ref{strategyproof}]\label{strategy}
Under Assumption \ref{onesidecontrolassum}, the strategy $\vct{g}^{1*}$ defined in Algorithm \ref{alg:example} is a min-max strategy for Team 1 in the original game $\gm{G}$.
\end{theorem}}}

\section{A Special Case and Numerical Experiments}\label{specialcases}
Consider an instance of Game $\gm{G}$ in which Team 1 has two players and Team 2 has only one player.
%\begin{equation}
%\rv{X}_{t+1} = f_t(\rv{X}_t, \rv{U}_t^{1}, \rv{U}_t^{2}, \rv{W}_t^s).
%\end{equation}
%Note that the state evolu controlled by both players. 
%At each time $t$, player 1 knows the state exactly. Thus,
%\begin{equation}
%\rv{Y}_t^1 = \rv{X}_t,
%\end{equation}
%and the second player's observation process $\rv{Y}_t^2$ is as described earlier in Section \ref{sec:probform}, i.e.
%\begin{equation}
%\rv{Y}_t^2 = h_t^2(\rv{X}_t, \rv{U}_{t-1}^{1:2},\rv{W}_t^2).
%\end{equation}
At each time $t$, Player 1 in Team 1 observes the state perfectly, i.e. $Y_t^{1,1} = X_t$, but the player in Team 2 gets an imperfect observation $Y^2_t$ defined as in \eqref{obseq}.
Player 1 has complete information: at each time $t$, it knows the entire state, observation and action histories of all the players. The player in Team 2 has partial information: at each time $t$, it knows its observation history $Y^2_{1:t}$ and action histories of all the players. Player 2 in Team 1 has the same information as that of the player in Team 2.
%\begin{remark}
%In our original model in Section \ref{sec:sys}, the observation $\rv{Y}_t^2$ depends only on the state $\rv{X}_t$. However, we allow the observation to depend on the first player's previous action $\rv{U}_{t-1}^1$. To be consistent with the original model, we can considered an expanded state $\rv{X}^e_t \doteq \{\rv{X}_t, \rv{U}^1_{t-1}\}$ such that $\rv{Y}_t^2 = h_t^2(\rv{X}_t^e,\rv{W}_t^2).$ We later show that this state expansion does not affect the dynamic programming characterization of the value function.
%\end{remark}
Thus, the total information available to each player at $t$ is as follows:
\begin{align*}
&\rv{I}_t^{1,1} = \{\rv{X}_{1:t}, \rv{Y}_{1:t}^{2},\rv{U}_{1:t-1}^{1:2}\}; &I_t^{1,2} = \rv{I}_t^2 = \{\rv{Y}_{1:t}^{2},\rv{U}_{1:t-1}^{1:2}\}.
\end{align*}
Clearly, $\rv{I}_t^2 \subseteq \rv{I}_t^{1,1}$. The  common and private information for this game can be written as follows: $\rv{C}_t = \rv{I}_t^2$,  $\rv{P}_t^{1,1} = \{\rv{X}_{1:t}\}$ and $P_t^{1,2} = \rv{P}_t^2 = \varnothing$. The increment in common information at time $t$ is $\rv{Z}_t = \{\rv{Y}_t^2,\rv{U}_{t-1}^{1:2}\}$.
In the game described above, the private information in $P_t^{1,1}$ includes the entire state history. However, Player 1 in Team 1 can ignore the past states $X_{1:t-1}$ without loss of optimality.
\begin{lemma}[Proof in App. \ref{structlemmaproof}]\label{structlemma}
There exists a min-max strategy $g^{1*}$ such that the control law $g^{1,1*}_t$ at time $t$ uses only $\rv{X}_t$ and  $\rv{I}_t^2$ to select $\delta U_t^{1,1}$, i.e.,
$\delta U^{1,1}_t = g^{1,1*}_t(X_t,I^2_t).$
\end{lemma}
The lemma above implies that, for the purpose of characterizing the value of the game and a min-max strategy for Team 1, we can restrict player 1's information structure to be $I^{1,1}_t = \{X_t,I^2_t\}$.
Thus, the common and private information become:  $\rv{C}_t = \rv{I}_t^2$, $\rv{P}_t^{1,1} = \{\rv{X}_{t}\}$ and $\rv{P}_t^2 = P_t^{1.1} = \varnothing$. 
We refer to this game with reduced private information as Game $\gm{H}$. The corresponding  expanded virtual game is denoted by $\gm{H}_e$. A general methodology for reducing private information in decentralized team and game problems can be found in \cite{tavafoghi2018sufficient}.
% \begin{remark}
% When Player 1 in Team 1 does not observe the state $X_t$ fully, then Player 1's information can be reduced to $I_t^2$ and its belief on the state $X_t$. In that case, the CIB is essentially a belief on Player 1's belief on the state. Methods that employ this kind of \emph{second order} beliefs can be found in \cite{gensbittel2014existence,xie2020optimally}.
% \end{remark}
The information structure in $\gm{H}$ is a special case of the first information structure in Section \ref{infexample}, and thus satisfies Assumption \ref{onesidecontrolassum}. Therefore, using the dynamic program in Section \ref{onesidedp}, we can obtain the value function $V_1^u$ and the min-max strategy $g^{1*}$. 
% We have the following corollary due to Theorem \ref{strategy}.
% \begin{corollary}
% The upper value $S^u(\gm{H}) = \E[V_1^u(\Pi_1)]$ and $g^{1*}$ is a min-max strategy in game $\gm{H}$.
% \end{corollary}

\paragraph{Numerical Experiments}
We consider a particular example of game $\gm{H}$ described above. In this example, there are two entities ($l$ and $r$) that can potentially be attacked and at any given time, exactly one of the entities is vulnerable. Player 1 of Team 1 knows which of the two entities is vulnerable whereas all the other players do not have this information. Player 2 of Team 1 can choose to \emph{defend} one of the entities. The attacker in Team 2 can either launch a \emph{blanket attack} on both entities or launch a \emph{targeted attack} on one of the entities. When the attacker launches a blanket attack, the damage incurred by the system is minimal if Player 2 in Team 1 happens to be defending the vulnerable entity and the damage is substantial otherwise. When the attacker launches a targeted attack on the vulnerable entity, the damage is substantial irrespective of the defender's position. But if the attacker targets the invulnerable entity, the attacker becomes passive and cannot attack for some time. Thus, launching a targeted attack involves high risk for the attacker. The state of the attacker (active or passive) and all the players' actions are public information. The system state $X_t$ thus has two components, the hidden state ($l$ or $r$) and the state of the attacker ($a$ or $p$). For convenience, we will denote the states $(l,a)$ and $(r,a)$ with 0 and 1 respectively.

The only role of Player 1 in Team 1 in this game is to signal the hidden state using two available actions $\alpha$ and $\beta$. The main challenge is that both the defender and the attacker can see Player 1's actions. Player 1 needs to signal the hidden state to some extent so that its teammate's defense is effective under blanket attacks. However, if too much information is revealed, the attacker can exploit it to launch a targeted attack and cause significant damage. In this example, the key is to design a strategy that can balance between these two contrasting goals of \emph{signaling} and \emph{secrecy}. A precise description of this model is provided in Appendix \ref{numappend}.

% In game $\gm{H}$, player 1 in Team 1 knows the state whereas the player 2 in Team 1 does not. Hence, in order to coordinate effectively with player 2, player 1 needs to use an appropriate signaling mechanism so that player 2 can infer the hidden state $X_t$ using its information. However, the adversary in Team 2 also has the same information as player 2 in Team 1. If player 1 in Team 1 reveals too much information about the state $X_t$, then the adversary can potentially exploit it to inflict severe damage upon Team 1. In our example, the key is to design a strategy that can balance between these two contrasting goals of signaling and secrecy. The precise details about the model are provided in Appendix \ref{numappend}.

\begin{figure}[]
\centering
  \subfigure[][An estimate of the value function $V_1^u(\cdot)$.]{
    \includegraphics[width=0.9\columnwidth]{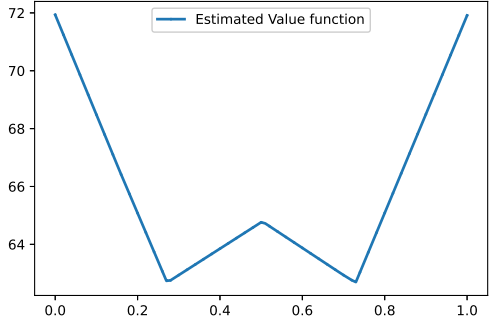}\label{valuefunc}}
  \subfigure[][Prescriptions at $t=1$ for Player 1 in Team 1.]{
    \includegraphics[width=0.9\columnwidth]{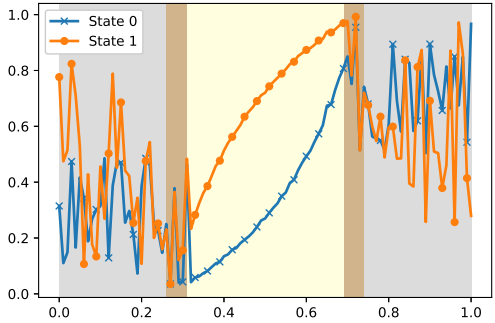}\label{strategy_fig}}
    \caption{In these plots, the $x$-axis represents $\pi_1(0)$ and we restrict our attention to those beliefs where $\pi_1(0)+\pi_1(1) = 1$, i.e. when the attacker is active. {In Figure \ref{strategy_fig}, the blue and red curves respectively depict the Bernoulli probabilities associated with the distributions $\gamma_1^{1,1}(0)$ and $\gamma_1^{1,1}(1)$, where $\gamma_1^{1,1}$ is Player 1's prescription in Team 1.}}
    \label{somefig}
    \vspace{-0.15in}
\end{figure}

%
% \begin{figure}[]
% \begin{center}
% \centerline{\includegraphics[width=0.5\columnwidth]{fig_15_est.eps}}
% \caption{An estimate of the value function $V_1^u(\cdot)$. In this plot, the $x$-axis represents $\pi_1(0)$ and we restrict our attention to those beliefs where $\pi_1(0)+\pi_1(1) = 1$, i.e. when the attacker is active.}
% \label{valuefunc}
% \end{center}
% \vskip -0.3in
% \end{figure}
%
%
% \begin{figure}[]
% \begin{center}
% \centerline{\includegraphics[width=0.5\columnwidth]{fig_strat.eps}}
% \caption{In this plot, the $x$-axis represents $\pi_1(0)$ and we restrict our attention to those beliefs where $\pi_1(0)+\pi_1(1) = 1$, i.e. when the attacker is active. The blue and red curves respectively depict the Bernoulli probabilities associated with the distributions $\gamma_1^{1,1}(0)$ and $\gamma_1^{1,1}(1)$, where $\gamma_1^{1,1}$ is Player 1's prescription in Team 1.}
% \label{strategy_fig}
% \end{center}
% \vskip -0.3in
% \end{figure}

% \begin{figure}[]
% \begin{center}
% \centerline{\includegraphics[width=0.5\columnwidth]{fig_strat.eps}}
% \caption{An estimate of the value function $V_1^u(\cdot)$. In this plot, the $x$-axis represents $\pi_1(0)$ and we restrict our attention to those beliefs where $\pi_1(0)+\pi_1(1) = 1$, i.e. when the attacker is active.}
% \label{valuefunc}
% \end{center}
% \vskip -0.3in
% \end{figure}
In order to solve this problem, we used the approximate DP approach discussed in Appendix \ref{dpsolve}. The value function $V_1^u(\cdot)$ thus obtained is shown in Figure \ref{valuefunc}. 
The tension between signaling and secrecy can be seen in the shape of the value function in Figure \ref{valuefunc}. When the CIB $\pi_1(0) = 0.5$, the value function is concave in its neighborhood and decreases as we move away from 0.5. This indicates that in these belief states, revealing the hidden state to some extent is preferable. However, as the belief goes further away from 0.5, the value function starts increasing at some point. This indicates that the adversary has too much information and is using it to inflict damage upon the system. Figure \ref{strategy_fig} depicts Player 1's prescriptions leading to non-trivial signaling patterns at various belief states. Notice that the distributions $\gamma_1^{1,1}(0)$ and $\gamma_1^{1,1}(1)$ for hidden states 0 and 1 are quite distinct when $\pi_1(0) = 0.5$ (indicating significant signaling) and are nearly identical when $\pi_1(0) = 0.72$ (indicating negligible signaling). A more detailed discussion on our experimental results can be found in Appendix \ref{numappend}.

\section{Conclusions}\label{conc}
We considered a general model of stochastic zero-sum games between two competing decentralized teams and provided bounds on their upper and lower values in the form of CIB based dynamic programs. When game has a value, our bounds coincide with the value. We identified several instances of this game model (including previously considered models) in which the CIB is controlled only by one of the teams (say the minimizing team). For such games, we also provide a characterization of the min-max strategy. Under this strategy, each player only uses the current CIB and its private information to select its actions. The sufficiency of the CIB and private information for optimality can potentially be exploited to design efficient strategies in various problems. Finally, we proposed a computational approach for approximately solving the CIB based DPs. There is significant scope for improvement in our computational approach. Tailored forward exploration heuristics for sampling the belief space and adding a policy network can improve the accuracy and tractability of our approach.

%%%%%%%%%%%%%%%%%%%%%%%%%%%%%%%%%%%%%%%%%%%%%%%%%%%%%%%%%%%%%%%%%%%%%%%%%%%%%%%%

\bibliographystyle{IEEEbib}
\bibliography{refs}

\begin{thebibliography}{10}

\bibitem{xie2020optimally}
Yuxuan Xie, Jilles Dibangoye, and Olivier Buffet,
\newblock ``Optimally solving two-agent decentralized pomdps under one-sided
  information sharing,''
\newblock in {\em International Conference on Machine Learning}. PMLR, 2020,
  pp. 10473--10482.

\bibitem{nayyar2010optimal}
Ashutosh Nayyar, Aditya Mahajan, and Demosthenis Teneketzis,
\newblock ``Optimal control strategies in delayed sharing information
  structures,''
\newblock {\em IEEE Transactions on Automatic Control}, vol. 56, no. 7, pp.
  1606--1620, 2010.

\bibitem{fang2019artificial}
Fei Fang, Milind Tambe, Bistra Dilkina, and Andrew~J Plumptre,
\newblock {\em Artificial intelligence and conservation},
\newblock Cambridge University Press, 2019.

\bibitem{nayyar2013decentralized}
Ashutosh Nayyar, Aditya Mahajan, and Demosthenis Teneketzis,
\newblock ``Decentralized stochastic control with partial history sharing: A
  common information approach,''
\newblock {\em IEEE Transactions on Automatic Control}, vol. 58, no. 7, pp.
  1644--1658, 2013.

\bibitem{oliehoek2016concise}
Frans~A Oliehoek and Christopher Amato,
\newblock {\em A concise introduction to decentralized POMDPs},
\newblock Springer, 2016.

\bibitem{rashid2018qmix}
Tabish Rashid, Mikayel Samvelyan, Christian Schroeder, Gregory Farquhar, Jakob
  Foerster, and Shimon Whiteson,
\newblock ``Qmix: Monotonic value function factorisation for deep multi-agent
  reinforcement learning,''
\newblock in {\em International Conference on Machine Learning}. PMLR, 2018,
  pp. 4295--4304.

\bibitem{foerster2019bayesian}
Jakob Foerster, Francis Song, Edward Hughes, Neil Burch, Iain Dunning, Shimon
  Whiteson, Matthew Botvinick, and Michael Bowling,
\newblock ``Bayesian action decoder for deep multi-agent reinforcement
  learning,''
\newblock in {\em International Conference on Machine Learning}. PMLR, 2019,
  pp. 1942--1951.

\bibitem{nayyar2017information}
Ashutosh Nayyar and Abhishek Gupta,
\newblock ``Information structures and values in zero-sum stochastic games,''
\newblock in {\em American Control Conference (ACC), 2017}. IEEE, 2017, pp.
  3658--3663.

\bibitem{tang2021dynamic}
Dengwang Tang, Hamidreza Tavafoghi, Vijay Subramanian, Ashutosh Nayyar, and
  Demosthenis Teneketzis,
\newblock ``Dynamic games among teams with delayed intra-team information
  sharing,''
\newblock {\em arXiv preprint arXiv:2102.11920}, 2021.

\bibitem{mertens2015repeated}
Jean-Fran{\c{c}}ois Mertens, Sylvain Sorin, and Shmuel Zamir,
\newblock {\em Repeated games}, vol.~55,
\newblock Cambridge University Press, 2015.

\bibitem{rosenberg2004stochastic}
Dinah Rosenberg, Eilon Solan, and Nicolas Vieille,
\newblock ``Stochastic games with a single controller and incomplete
  information,''
\newblock {\em SIAM journal on control and optimization}, vol. 43, no. 1, pp.
  86--110, 2004.

\bibitem{renault2012value}
J{\'e}r{\^o}me Renault,
\newblock ``The value of repeated games with an informed controller,''
\newblock {\em Mathematics of operations Research}, vol. 37, no. 1, pp.
  154--179, 2012.

\bibitem{gensbittel2014existence}
Fabien Gensbittel, Miquel Oliu-Barton, and Xavier Venel,
\newblock ``Existence of the uniform value in zero-sum repeated games with a
  more informed controller,''
\newblock {\em Journal of Dynamics and Games}, vol. 1, no. 3, pp. 411--445,
  2014.

\bibitem{li2016recursive}
Xiaoxi Li and Xavier Venel,
\newblock ``Recursive games: uniform value, tauberian theorem and the mertens
  conjecture,''
\newblock {\em International Journal of Game Theory}, vol. 45, no. 1-2, pp.
  155--189, 2016.

\bibitem{kartik2020upper}
Dhruva Kartik and Ashutosh Nayyar,
\newblock ``Upper and lower values in zero-sum stochastic games with asymmetric
  information,''
\newblock {\em Dynamic Games and Applications}, pp. 1--26, 2020.

\bibitem{zheng2013decomposition}
Jiefu Zheng and David~A Casta{\~n}{\'o}n,
\newblock ``Decomposition techniques for {Markov} zero-sum games with nested
  information,''
\newblock in {\em 52nd IEEE Conference on Decision and Control}. IEEE, 2013,
  pp. 574--581.

\bibitem{li2014lp}
Lichun Li and Jeff Shamma,
\newblock ``{LP} formulation of asymmetric zero-sum stochastic games,''
\newblock in {\em 53rd IEEE Conference on Decision and Control}. IEEE, 2014,
  pp. 1930--1935.

\bibitem{smith2004heuristic}
Trey Smith and Reid Simmons,
\newblock ``Heuristic search value iteration for pomdps,''
\newblock in {\em Proceedings of the 20th conference on Uncertainty in
  artificial intelligence}, 2004, pp. 520--527.

\bibitem{horak2017heuristic}
Karel Hor{\'a}k, Branislav Bo{\v{s}}ansk{\`y}, and Michal
  P{\v{e}}chou{\v{c}}ek,
\newblock ``Heuristic search value iteration for one-sided partially observable
  stochastic games,''
\newblock in {\em Proceedings of the AAAI Conference on Artificial
  Intelligence}, 2017, vol.~31.

\bibitem{horak2019solving}
Karel Hor{\'a}k and Branislav Bo{\v{s}}ansk{\`y},
\newblock ``Solving partially observable stochastic games with public
  observations,''
\newblock in {\em Proceedings of the AAAI Conference on Artificial
  Intelligence}, 2019, vol.~33, pp. 2029--2036.

\bibitem{von1997team}
Bernhard von Stengel and Daphne Koller,
\newblock ``Team-maxmin equilibria,''
\newblock {\em Games and Economic Behavior}, vol. 21, no. 1-2, pp. 309--321,
  1997.

\bibitem{farina2018ex}
Gabriele Farina, Andrea Celli, Nicola Gatti, and Tuomas Sandholm,
\newblock ``Ex ante coordination and collusion in zero-sum multi-player
  extensive-form games,''
\newblock in {\em Conference on Neural Information Processing Systems (NIPS)},
  2018.

\bibitem{zhang2020computing}
Youzhi Zhang and Bo~An,
\newblock ``Computing team-maxmin equilibria in zero-sum multiplayer
  extensive-form games,''
\newblock in {\em Proceedings of the AAAI Conference on Artificial
  Intelligence}, 2020, vol.~34, pp. 2318--2325.

\bibitem{nayyar2014common}
Ashutosh Nayyar, Abhishek Gupta, Cedric Langbort, and Tamer Ba{\c{s}}ar,
\newblock ``Common information based {Markov} perfect equilibria for stochastic
  games with asymmetric information: Finite games,''
\newblock {\em IEEE Transactions on Automatic Control}, vol. 59, no. 3, pp.
  555--570, 2014.

\bibitem{ouyang2017dynamic}
Yi~Ouyang, Hamidreza Tavafoghi, and Demosthenis Teneketzis,
\newblock ``Dynamic games with asymmetric information: Common information based
  perfect bayesian equilibria and sequential decomposition,''
\newblock {\em IEEE Transactions on Automatic Control}, vol. 62, no. 1, pp.
  222--237, 2017.

\bibitem{vasal2019systematic}
Deepanshu Vasal, Abhinav Sinha, and Achilleas Anastasopoulos,
\newblock ``A systematic process for evaluating structured perfect bayesian
  equilibria in dynamic games with asymmetric information,''
\newblock {\em IEEE Transactions on Automatic Control}, vol. 64, no. 1, pp.
  78--93, 2019.

\bibitem{rebel}
Noam Brown, Anton Bakhtin, Adam Lerer, and Qucheng Gong,
\newblock ``Combining deep reinforcement learning and search for
  imperfect-information games,''
\newblock in {\em Advances in Neural Information Processing Systems},
  H.~Larochelle, M.~Ranzato, R.~Hadsell, M.~F. Balcan, and H.~Lin, Eds. 2020,
  vol.~33, pp. 17057--17069, Curran Associates, Inc.

\bibitem{kartik2019stochastic}
Dhruva Kartik and Ashutosh Nayyar,
\newblock ``Stochastic zero-sum games with asymmetric information,''
\newblock in {\em 58th IEEE Conference on Decision and Control}. IEEE, 2019.

\bibitem{kumar2015stochastic}
Panganamala~Ramana Kumar and Pravin Varaiya,
\newblock {\em Stochastic systems: Estimation, identification, and adaptive
  control}, vol.~75,
\newblock SIAM, 2015.

\bibitem{guide2006infinite}
A~Hitchhiker’s Guide,
\newblock {\em Infinite dimensional analysis},
\newblock Springer, 2006.

\bibitem{maschler2013game}
Michael Maschler, Eilon Solan, and Shmuel Zamir,
\newblock {\em Game Theory},
\newblock Cambridge University Press, 2013.

\bibitem{anantharam2007common}
Venkat Anantharam and Vivek Borkar,
\newblock ``Common randomness and distributed control: A counterexample,''
\newblock {\em Systems \& control letters}, vol. 56, no. 7-8, pp. 568--572,
  2007.

\bibitem{osborne1994course}
Martin~J Osborne and Ariel Rubinstein,
\newblock {\em A course in game theory},
\newblock MIT press, 1994.

\bibitem{dibangoye2016optimally}
Jilles~Steeve Dibangoye, Christopher Amato, Olivier Buffet, and Fran{\c{c}}ois
  Charpillet,
\newblock ``Optimally solving dec-pomdps as continuous-state mdps,''
\newblock {\em Journal of Artificial Intelligence Research}, vol. 55, pp.
  443--497, 2016.

\bibitem{tavafoghi2018sufficient}
Hamidreza Tavafoghi, Yi~Ouyang, and Demosthenis Teneketzis,
\newblock ``A sufficient information approach to decentralized decision
  making,''
\newblock in {\em 2018 IEEE Conference on Decision and Control (CDC)}. IEEE,
  2018, pp. 5069--5076.

\bibitem{hernandez2012discrete}
On{\'e}simo Hern{\'a}ndez-Lerma and Jean~B Lasserre,
\newblock {\em Discrete-time Markov control processes: basic optimality
  criteria}, vol.~30,
\newblock Springer Science \& Business Media, 2012.

\bibitem{bertsekas1996neuro}
Dimitri~P Bertsekas and John~N Tsitsiklis,
\newblock {\em Neuro-dynamic programming}, vol.~5,
\newblock Athena Scientific Belmont, MA, 1996.

\bibitem{lin2020gradient}
Tianyi Lin, Chi Jin, and Michael Jordan,
\newblock ``On gradient descent ascent for nonconvex-concave minimax
  problems,''
\newblock in {\em International Conference on Machine Learning}. PMLR, 2020,
  pp. 6083--6093.

\bibitem{heusel2017gans}
Martin Heusel, Hubert Ramsauer, Thomas Unterthiner, Bernhard Nessler, and Sepp
  Hochreiter,
\newblock ``Gans trained by a two time-scale update rule converge to a local
  nash equilibrium,''
\newblock {\em arXiv preprint arXiv:1706.08500}, 2017.

\end{thebibliography}
\newpage
\onecolumn

\appendices
\section{Virtual Game $\gm{G}_v$ and Proof of Theorem \ref{origvirt}}\label{virtgameproof}

\subsection{Virtual Game $\gm{G}_v$}
In order to establish the connection between the original game $\gm{G}$ and the expanded game $\gm{G}_e$, we construct an intermediate virtual game $\gm{G}_v$. Virtual game $\gm{G}_v$ is very similar to the expanded game $\gm{G}_e$. The only difference between games $\gm{G}_v$ and $\gm{G}_e$ is that the virtual players in Game $\gm{G}_v$ have access only to the common information $C_t$ at time $t$. On the other hand, the virtual players in Game $\gm{G}_e$ have access to the common information $C_t$ \emph{as well as} the prescription history $\Gamma_{1:t-1}^{1:2}$ at time $t$. We formally define virtual game $\gm{G}_v$ below.

\paragraph{Model and Prescriptions} Consider virtual player $i$ associated with Team $i$, $i=1,2$. At each time $t$, virtual player $i$ selects a prescription $\Gamma^{i}_t = (\Gamma_t^{i,1},\dots,\Gamma_t^{i,N_i}) \in \mathcal{B}_t^i$, where $\mathcal{B}_t^i$ is the set of prescriptions defined in Section \ref{expandedgame}. 
Once virtual player $i$ selects its prescription, the action $U^{i,j}_t$ is randomly generated according to the distribution $\Gamma_t^{i,j}(\rv{P}_t^{i,j})$. More precisely, the system dynamics for this game are given by:

\begin{align}
\label{virdyn1}\rv{X}_{t+1} &= f_t(\rv{X}_t,\rv{U}_t^{1:2},\rv{W}_t^s)\\
\label{virdyn2}\rv{P}^i_{t+1} &= \xi_{t+1}^i(\rv{P}_t^{1:2},\rv{U}_t^{1:2},\rv{Y}_{t+1}^{1:2}) & i = 1,2,\\
\label{virdyn3}\rv{Y}_{t+1}^{i,j} &= h_{t+1}^{i,j}(\rv{X}_{t+1},\rv{U}_t^{1:2},\rv{W}_{t+1}^{i,j}) & i = 1,2; j = 1,\dots,N_i,\\
\label{virdyn4}\rv{U}_t^{i,j} &= \rand(\U_t^{i,j},\Gamma_t^{i,j}(\rv{P}_t^{i,j}),\rv{K}_t^{i,j}) & i = 1,2; j = 1,\dots,N_i\\
\label{virdyn5}\rv{Z}_{t+1} &= \zeta_{t+1}(\rv{P}_t^{1:2},\rv{U}_t^{1:2},\rv{Y}_{t+1}^{1:2}),
\end{align}
where the functions $f_t,h_t^{i,j}$, $\xi_t^i,\rand$ and $\zeta_{t}$  are the same as in $\gm{G}$.

%For a prescription $\gamma_t^i$, the probability of selecting an action $u_t^i$ when the private information is $p_t^i$ is denoted by $\gamma_t^i(p_t^i ; u_t^i)$.

%% Note that equations (\ref{virdyn1}-\ref{virdyn4}) describes a well-defined dynamic system.

\paragraph{Strategies} In virtual game $\gm{G}_v$, virtual players use the common information $\rv{C}_t$ to select their prescriptions at time $t$. The $i$th virtual player selects its prescription according to a control law $\chi_t^i$, i.e. $\Gamma_t^i = \chi_t^i(\rv{C}_t)$.
%\begin{align}
%\Gamma_t^i = \chi_t^i(\rv{C}_t).
%\end{align}
 For virtual player $i$, the collection of control laws over the entire time horizon $\vct{\chi}^i = (\chi_1^i,\dots,\chi_T^i)$ is referred to as its control strategy. Let   $\mathcal{H}_t^i$ be the set of all possible control laws for virtual player $i$ at time $t$  and let  $\mathcal{H}^i$ be the set of all possible control strategies for virtual player $i$, i.e. $\mathcal{H}^i = \mathcal{H}_1^i \times \dots \times \mathcal{H}_T^i$. 
The total cost associated with the game for a strategy profile $(\vct{\chi}^1,\vct{\chi}^2)$ is
\begin{align}
\mathcal{J}(\vct{\chi}^1,\vct{\chi}^2)=\E^{(\vct{\chi}^1,\vct{\chi}^2)}\left[\sum_{t=1}^T c_t(\rv{X}_t,\rv{U}_t^{1},\rv{U}_t^2)\right], \label{eq:virtualJv}
\end{align}
where the function $c_t$ is the same as in Game $\gm{G}$.

Note that any strategy $\chi^i \in \H^i$ is equivalent to the strategy $\tilde{\chi}^i \in \tilde{\H}^i$ that satisfies the following condition: for each time $t$ and for each realization of common information $\vct{c}_t \in \C_t$,
\begin{align}
\tilde{\chi}_t^i(\vct{c}_t,\gamma_{1:t-1}^{1:2}) = \chi_t^i(\vct{c}_t) \quad \forall\;\gamma_{1:t-1}^{1:2} \in \mathcal{B}_{1:t-1}^{1:2}.
\end{align}
Hence, with slight abuse of notation, we can say that the strategy space $\H^i$ in the virtual game $\gm{G}_v$ is a subset of the strategy space $\tilde{\H}^i$ in the expanded game $\gm{G}_e$.

The following lemma establishes a connection between the original game $\gm{G}$ and the virtual game $\gm{G}_v$ constructed above. 

\begin{lemma}\label{virtlemma}
Let $S^u(\gm{G}_v)$ and $S^l(\gm{G}_v)$ be, respectively, the upper and lower values of the virtual game $\gm{G}_v$. Then, \[S^l(\gm{G}) = S^l(\gm{G}_v) ~~\mbox{and}~~ S^u(\gm{G}) = S^u(\gm{G}_v).\] 
Consequently, if a Team Nash equilibrium exists in the original game $\gm{G}$, then $S(\gm{G})=S^l(\gm{G}_v) = S^u(\gm{G}_v).$ 
\end{lemma}
\begin{proof}
For a given strategy $g^i \in \mathcal{G}^i$, let us define a strategy $\chi^i \in \H^i$ in the following manner. For each time $t$, each instance of common information $c_t \in \C_t$ and player $j = 1,\dots,N_i$, let
\begin{align}
    \gamma_t^{i,j} \doteq g_t^{i,j}(c_t,\cdot),
\end{align}
and let $\gamma_t^i \doteq (\gamma_t^{i,1},\dots,\gamma_t^{i,N_i})$. Note that the partial function $g_t^{i,j}(c_t,\cdot)$ is a mapping from $\mathcal{P}_t^{i,j}$ to $\Delta \U_t^{i,j}$. Let
\begin{align}
    \chi_t^i(c_t) \doteq \gamma_t^i.
\end{align}
We will denote this correspondence between strategies in Game $\gm{G}$ and  $\gm{G}_v$ with $\mathcal{M}^i: \mathcal{G}^i \rightarrow \mathcal{H}^i$, $i=1,2$, i.e., $\chi^i = \mathcal{M}^i(g^i)$. One can easily verify that the mapping $\mathcal{M}^i$ is bijective. Further, for every $g^1 \in \G^1$ and $g^2 \in \G^2$, we have
\begin{align}
J(g^1,g^2) = \mathcal{J}(\mathcal{M}^1(g^1),\mathcal{M}^2(g^2)).
\end{align}
We refer the reader to Appendix A of \cite{nayyar2014common} for a proof of the above statement.

Therefore, for any strategy $g^1 \in \G^1$, we have
\begin{align}
\sup_{g^2\in \G^2}J(g^1,g^2) &= \sup_{g^2\in \G^2}\mathcal{J}(\mathcal{M}^1(g^1),\mathcal{M}^2(g^2))\\
&= \sup_{\chi^2\in \H^2}\mathcal{J}(\mathcal{M}^1(g^1),\chi^2).
\end{align}
Consequently,
\begin{align}
\inf_{g^1\in\G^1}\sup_{g^2\in \G^2}J(g^1,g^2) &= \inf_{g^1\in\G^1}\sup_{\chi^2\in \H^2}\mathcal{J}(\mathcal{M}^1(g^1),\chi^2)\\
&=  \inf_{\chi^1\in\H^1}\sup_{\chi^2\in \H^2}\mathcal{J}(\chi^1,\chi^2).
\end{align}
This implies that $S^u(\gm{G}) = S^u(\gm{G}_v)$. We can similarly prove that $S^l(\gm{G}) = S^l(\gm{G}_v)$.
{\begin{remark}\label{virtorigremark}
We can also show that a strategy $g^1$ is a min-max strategy in Game $\gm{G}$ if and only if $\mathcal{M}^1(g^1)$ is a min-max strategy in Game $\gm{G}_v$. A similar statement holds for max-min strategies as well.
\end{remark}}
\end{proof}

We will now establish the relationship between the upper and lower values of the virtual game $\gm{G}_v$ and the expanded virtual game $\gm{G}_e$. To do so, we define the following mappings between the strategies in games $\gm{G}_v$ and $\gm{G}_{e}$. To do so, we define the following mappings between the strategies in games $\gm{G}_v$ and $\gm{G}_{e}$. 

%\red{I changed the def; please check}
\begin{definition}\label{def:rho}
Let $\varrho^i: \tilde{\H}^1 \times  \tilde{\H}^2 \rightarrow \H^i$ be an operator that maps a strategy profile  $(\vct{\tilde{\chi}}^1,\vct{\tilde{\chi}}^2)$ in virtual game $\gm{G}_e$ to a strategy  $ \vct{{\chi}}^i$ for virtual player $i$ in  game $\gm{G}_v$ as follows: For $t=1,2,\ldots,T,$
%$\tilde{\chi}_t^i$ and $\tilde{\chi}_{1:t-1}^{1:2}$ to $\chi_t^i: \C_t \rightarrow \mathcal{B}_t^i$ such that
\begin{align}
\chi_t^i(\vct{c}_t) \doteq \tilde{\chi}_t^i(\vct{c}_t, \tilde{\gamma}_{1:t-1}^{1:2}),
\end{align}
where  $\tilde{\gamma}_s^j = \tilde{\chi}_s^j(\vct{c}_s,\tilde{\gamma}_{1:s-1}^{1:2})$ for every $1\leq s \leq t-1$ and $j = 1,2$.
%Also, denote $\vct{\chi}^i \doteq \varrho^i(\vct{\tilde{\chi}}^1,\vct{\tilde{\chi}}^2)$ where
%\begin{align}
%\vct{\chi}^i \doteq \left(\varrho_t^i(\tilde{\chi}_1^i), \dots, \varrho_t^i(\tilde{\chi}_T^i,\tilde{\chi}_{1:T-1}^{1:2})\right).
%\end{align}
We denote the ordered pair $(\varrho^1,\varrho^2)$ by $\varrho$.
\end{definition}
The mapping $\varrho$ is defined in such a way that the strategy profile  $(\vct{\tilde{\chi}}^1,\vct{\tilde{\chi}}^2)$  and the strategy profile  $\varrho(\vct{\tilde{\chi}}^1,\vct{\tilde{\chi}}^2)$ induce identical dynamics in the respective games $\gm{G}_e$ and $\gm{G}_v$.

\begin{lemma}\label{evolequi}
Let $(\vct{{\chi}}^1,\vct{{\chi}}^2)$ and $(\vct{\tilde{\chi}}^1,\vct{\tilde{\chi}}^2)$ be  strategy profiles for games $\gm{G}_v$ and $\gm{G}_{e}$, such that $\vct{\chi}^i = \varrho^i(\vct{\tilde{\chi}}^1,\vct{\tilde{\chi}}^2)$, $i=1,2$. Then, 
%\red{Can we exclude the red part??for player $i = 1,2$ and at each time $1\leq t \leq T$,
%\begin{align}
%\rv{X}_t &= \tilde{\rv{X}}_t & \rv{Y}_t = \tilde{\rv{Y}}_t^i \nonumber \\
%\rv{U}_t^i &= \tilde{\rv{U}}_t^i & \rv{P}_t^i = \tilde{\rv{P}}_t^i\nonumber\\
%\rv{C}_t &= \tilde{\rv{C}}_t & \Gamma_t^i = \tilde{\Gamma}_t^i,
%\end{align}
%with probability 1. Note that the variables with a tilde correspond to the expanded game $\gm{G}_e$. Consequently, we have}
\begin{align}
\mathcal{J}(\vct{{\chi}}^1,\vct{{\chi}}^2) = {\mathcal{J}}(\vct{\tilde{\chi}}^1,\vct{\tilde{\chi}}^2).
\end{align}
\end{lemma}
\begin{proof}
Let us consider the evolution of the virtual game $\gm{G}_v$ under the strategy profile $(\chi^1,\chi^2)$ and the expanded virtual game $\gm{G}_e$ under the strategy profile $(\tilde{\chi}^1,\tilde{\chi}^2)$. Let the primitive variables and the randomization variables $K_t^{i,j}$ in both games be identical. The variables such as the state, action and information variables in the expanded game $\gm{G}_e$ are distinguished from those in the virtual game $\gm{G}_v$ by means of a tilde. For instance, $X_t$ is the state in game $\gm{G}_v$ and $\tilde{X}_t$ is the state in game $\gm{G}_e$.

We will prove by induction that the system evolution in both these games is identical over the entire horizon. This is clearly true at the end of time $t=1$ because the state, observations and the common and private information variables are identical in both games. Moreover, since $\vct{\chi}^i = \varrho^i(\vct{\tilde{\chi}}^1,\vct{\tilde{\chi}}^2)$, $i=1,2$, the strategies $\chi^i_1$ and $\tilde{\chi}^i_1$ are identical by definition (see Definition \ref{def:rho}). Thus, the prescriptions and actions at $t=1$ are also identical.

For induction, assume that the system evolution in both games is identical until the end of time $t$. Then, $$\rv{X}_{t+1} = f_t(\rv{X}_t, \rv{U}_t^{1:2},\rv{W}_t^s) = f_t(\tilde{\rv{X}}_t, \tilde{\rv{U}}_t^{1:2},\rv{W}_t^s) = \tilde{\rv{X}}_{t+1}.$$ Using equations (\ref{obseq}), (\ref{privevol}) and (\ref{commonevol}), we can similarly argue that $\rv{Y}_{t+1}^i = \tilde{\rv{Y}}_{t+1}^i$, $\rv{P}_{t+1}^i = \tilde{\rv{P}}_{t+1}^i$ and $\rv{C}_{t+1} = \tilde{\rv{C}}_{t+1}$. Since $\vct{\chi}^i = \varrho^i(\vct{\tilde{\chi}}^1,\vct{\tilde{\chi}}^2)$, we also have 
\begin{align}
\tilde{\Gamma}_{t+1}^i &= \tilde{\chi}_{t+1}^i(\tilde{\rv{C}}_{t+1},\tilde{\Gamma}_{1:t}^{1:2}) \stackrel{a}{=} \chi_{t+1}^i(\tilde{\rv{C}}_{t+1}) \stackrel{b}{=} \Gamma_{t+1}^i.
\end{align}
Here, equality $(a)$ follows from the construction of the mapping $\varrho^i$ (see Definition \ref{def:rho}) and equality $(b)$ follows from the fact that $C_{t+1} = \tilde{C}_{t+1}$. Further, 
\begin{align}
\rv{U}_{t+1}^{i,j} = \rand(\U_t^{i,j},\Gamma_{t+1}^{i,j}(\rv{P}_{t+1}^{i,j}),\rv{K}_{t+1}^{i,j}) &= \rand(\U_t^{i,j},\tilde{\Gamma}_{t+1}^{i,j}(\tilde{\rv{P}}_{t+1}^{i,j}),\rv{K}_{t+1}^{i,j}) \\
&= \tilde{\rv{U}}_{t+1}^{i,j}.
\end{align}
Thus, by induction, the hypothesis is true for every $1\leq t \leq T$. This proves that the virtual and expanded games have identical dynamics under strategy profiles $(\vct{\chi}^1,\chi^2)$ and $(\tilde{\vct{\chi}}^1,\tilde{\chi}^2)$.

Since the virtual and expanded games have the same cost structure, having identical dynamics ensures that strategy profiles $(\vct{\chi}^1,\chi^2)$ and $(\tilde{\vct{\chi}}^1,\tilde{\chi}^2)$ have the same expected cost in games $\gm{G}_v$ and $\gm{G}_e$, respectively. Therefore, $\mathcal{J}(\vct{\chi}^1,\chi^2) = {\mathcal{J}}(\tilde{\vct{\chi}}^1,\tilde{\chi}^2)$.
%\begin{align}
%\mathcal{J}(\vct{\chi}) = \tilde{\mathcal{J}}(\tilde{\vct{\chi}}).
%\end{align}
\end{proof}

\subsection{Proof of Theorem \ref{origvirt}}\label{origvirtproof}
%\blue{I edited this proof. Please verify.}\\
%\blue{verified}

For any strategy $\chi^1 \in \H^1$, we have 
\begin{align}
\sup_{\tilde{\chi}^2 \in \tilde{\H}^2}\mathcal{J}({\chi}^1,\tilde{\chi}^2) \geq \sup_{{\chi}^2 \in {\H}^2}\mathcal{J}({\chi}^1,{\chi}^2), \label{eq:thm1c}
\end{align}
because $\H^2 \subseteq \tilde{\H}^2$. Further,
\begin{align}
\sup_{\tilde{\chi}^2 \in \tilde{\H}^2}\mathcal{J}({\chi}^1,\tilde{\chi}^2) &=\sup_{\tilde{\chi}^2 \in \tilde{\H}^2}\mathcal{J}(\varrho^1(\chi^1,\tilde{\chi}^2),\varrho^2(\chi^1,\tilde{\chi}^2)). \label{eq:thm1a} \\
&= \sup_{\tilde{\chi}^2 \in \tilde{\H}^2}\mathcal{J}({\chi}^1,\varrho^2(\chi^1,\tilde{\chi}^2))\\
&\leq \sup_{{\chi}^2 \in {\H}^2}\mathcal{J}({\chi}^1,{\chi}^2), \label{eq:thm1b}
\end{align} 
where the first equality is due to Lemma \ref{evolequi}, the second equality is because $\varrho^1(\chi^1,\tilde{\chi}^2) = \chi^1$ and the last inequality is due to the fact that $\varrho^2(\chi^1,\tilde{\chi}^2) \in {\H}^2$ for any $\tilde{\chi}^2 \in \tilde{\H}^2$.

%Notice that for any strategy $\tilde{\chi}^2 \in \tilde{\H}^2$, $\varrho^1(\chi^1,\tilde{\chi}^2) = \chi^1$. Therefore, using Lemma \ref{evolequi}, we have
%\begin{align}
%\sup_{\tilde{\chi}^2 \in \tilde{\H}^2}\mathcal{J}({\chi}^1,\tilde{\chi}^2) &= \sup_{\tilde{\chi}^2 \in \tilde{\H}^2}\mathcal{J}({\chi}^1,\varrho^2(\chi^1,\tilde{\chi}^2))\\
%&\leq \sup_{{\chi}^2 \in {\H}^2}\mathcal{J}({\chi}^1,{\chi}^2).
%\end{align}
Combining \eqref{eq:thm1c} and \eqref{eq:thm1b}, we obtain that
\begin{align}
\sup_{{\chi}^2 \in {\H}^2}\mathcal{J}({\chi}^1,{\chi}^2) = \sup_{\tilde{\chi}^2 \in \tilde{\H}^2}\mathcal{J}({\chi}^1,\tilde{\chi}^2). \label{eq:thm1d}
\end{align}
Now, 
%since $\H^1 \subseteq \tilde{\H}^1$, we have that
\begin{align}
S^u(\gm{G}_e) &:=\inf_{\tilde{\chi}^1 \in \tilde{\H}^1}\sup_{\tilde{\chi}^2 \in \tilde{\H}^2}\mathcal{J}(\tilde{\chi}^1,\tilde{\chi}^2) \\
\label{infsupineq} &\leq \inf_{{\chi}^1 \in {\H}^1}\sup_{\tilde{\chi}^2 \in \tilde{\H}^2}\mathcal{J}({\chi}^1,\tilde{\chi}^2)\\
&= \inf_{{\chi}^1 \in {\H}^1}\sup_{{\chi}^2 \in {\H}^2}\mathcal{J}({\chi}^1,{\chi}^2), \label{eq:thm1e}\\
&=: S^u(\gm{G}_v).
\end{align}
where the inequality (\ref{infsupineq}) is true since $\H^1 \subseteq \tilde{\H}^1$ and the   equality  in \eqref{eq:thm1e}  follows from \eqref{eq:thm1d}. 
Therefore, $S^u(\gm{G}_e) \leq S^u(\gm{G}_v)$. We can use similar arguments to show that $S^l(\gm{G}_v) \leq S^l(\gm{G}_e).$

\section{Proof of Lemma \ref{infstate}}\label{infstateproof}
{\begin{definition}[Notation for Prescriptions]\label{behavepres}
For virtual player $i$ at time $t$, let $\gamma^i = (\gamma^{i,1},\dots,\gamma^{i,N_i}) \in \mathcal{B}_t^i$ be a prescription. We use $\gamma^{i,j}(p_t^{i,j} ; u_t^{i,j})$ to denote the probability assigned to action $u_t^{i,j}$ by the distribution $\gamma^{i,j}(p_t^{i,j})$.
We will also use the following notation:
\begin{align*}
    \gamma^i(p_t^i;u_t^i) \doteq \prod_{j = 1}^{N_i}\gamma^{i,j}(p_t^{i,j};u_t^{i,j})~~\forall p_t^i\in \P_t^i,u_t^i \in \U_t^i.
\end{align*}
Here, $\gamma^i(p_t^i;u_t^i)$ is the probability assigned to a team action $u_t^i$ by the prescription $\gamma^i$ when the team's private information is $p_t^i$.
\end{definition}}
We begin with defining the following transformations for each time $t$. Recall that $\mathcal{S}_t$ is the set of all possible common information beliefs at time $t$ and $\mathcal{B}_t^i$ is the prescription space for virtual player $i$  at time $t$.
\begin{definition}\label{fdef}
\begin{enumerate}
\item Let $\mathscr{P}_t^j: \mathcal{S}_t \times \mathcal{B}_t^1 \times \mathcal{B}_t^2 \to \Delta(\mathcal{Z}_{t+1} \times \X_{t+1} \times \P_{t+1}^1 \times \P_{t+1}^2)$ be defined as
\begin{align}
\mathscr{P}_t^j(\pi_t,&\gamma_t^{1:2}; z_{t+1}, x_{t+1},p_{t+1}^{1:2}) \label{jointprob}\\
&:= \sum_{\vct{x}_t,\vct{p}^{1:2}_t,\vct{u}_t^{1:2}}\pi_t(\vct{x}_t,\vct{p}_t^{1:2})\gamma_t^1(\vct{p}_t^1;u_t^1)\gamma_t^2( \vct{p}_t^2; u_t^2)\Py[\vct{x}_{t+1}, \vct{p}_{t+1}^{1:2},\vct{z}_{t+1} \mid \vct{x}_t,\vct{p}_t^{1:2},\vct{u}_t^{1:2}].\label{condindepeq}
\end{align}
We will use $\mathscr{P}_t^j(\pi_t,\gamma_t^{1:2})$ as a shorthand for the probability distribution $\mathscr{P}_t^j(\pi_t,\gamma_t^{1:2}; \cdot, \cdot, \cdot)$.
The distribution $\mathscr{P}_t^j(\pi_t,\gamma_t^{1:2})$ can be viewed as a joint distribution over the variables $Z_{t+1},X_{t+1}, P_{t+1}^{1:2}$ if the distribution on $X_t, P^{1:2}_t$ is $\pi_t$ and prescriptions $\gamma^{1:2}_t$ are chosen by the virtual players at time $t$.
\item Let $\mathscr{P}_t^m: \mathcal{S}_t \times \mathcal{B}_t^1 \times \mathcal{B}_t^2 \to \Delta\mathcal{Z}_{t+1}$ be defined as
\begin{align}
\mathscr{P}_t^m(\pi_t,\gamma_t^{1:2}; z_{t+1}) = \sum_{x_{t+1},p_{t+1}^{1:2}} \mathscr{P}_t^j(\pi_t,\gamma_t^{1:2}; z_{t+1}, x_{t+1},p_{t+1}^{1:2}). \label{marginalprob}
\end{align}
{The distribution $\mathscr{P}_t^m(\pi_t,\gamma_t^{1:2})$ is the marginal distribution of the variable $Z_{t+1}$ obtained from the joint distribution $\mathscr{P}_t^j(\pi_t,\gamma_t^{1:2})$ defined above.}
\item Let $F_t: \mathcal{S}_t \times \mathcal{B}_t^1 \times \mathcal{B}_t^2 \times \mathcal{Z}_{t+1} \to \mathcal{S}_{t+1}$ be defined as
\begin{align}
F_t(\pi_t,\gamma_t^{1:2},z_{t+1})= 
\begin{cases}
\frac{\mathscr{P}_t^j(\pi_t,\gamma_t^{1:2};z_{t+1},\cdot,\cdot)}{\mathscr{P}_t^m(\pi_t,\gamma_t^{1:2};z_{t+1})} &\text{if } \mathscr{P}_t^m(\pi_t,\gamma_t^{1:2};z_{t+1}) > 0\\
G_t(\pi_t,\gamma_t^{1:2},z_{t+1}) & \text{otherwise},
\end{cases}
\end{align}
where $G_t: \mathcal{S}_t \times \mathcal{B}_t^1 \times \mathcal{B}_t^2 \times \mathcal{Z}_{t+1} \to \mathcal{S}_{t+1}$ can be any arbitrary measurable mapping. {One such mapping is the one that maps every element $\pi_t,\gamma_t^{1:2},z_{t+1}$ to the uniform distribution over the finite space $\mathcal{X}_{t+1} \times \P_{t+1}^1 \times \P_{t+1}^2$.}
\end{enumerate}
\end{definition}
Let the collection of transformations $F_t$ that can be constructed using the method described in Definition \ref{fdef} be denoted by $\mathscr{B}$. Note that the transformations $\mathscr{P}_t^j, \mathscr{P}_t^m$ and $F_t$ do not depend on the strategy profile $(\tilde{\chi}^1, \tilde{\chi}^2)$ because the term $\Py[\vct{x}_{t+1}, \vct{p}_{t+1}^{1:2},\vct{z}_{t+1} \mid \vct{x}_t,\vct{p}_t^{1:2},\vct{u}_t^{1:2}]$ in (\ref{condindepeq}) depends only on the system dynamics (see equations (\ref{virdyn1}) -- (\ref{virdyn5})) and not on the strategy profile $(\tilde{\chi}^1,\tilde{\chi}^2)$.

{Consider a strategy profile $(\tilde{\chi}^1, \tilde{\chi}^2)$.  Note that the number of possible realizations of common information and prescription history  under $(\tilde{\chi}^1, \tilde{\chi}^2)$  is finite. } Let $c_{t+1},\gamma_{1:t}^{1:2}$ be a realization of the common information and prescription history at time $t+1$ with non-zero probability of occurrence under  $(\tilde{\chi}^1, \tilde{\chi}^2)$. For this realization of virtual players' information, the common information based belief on the state and private information at time $t+1$ is given by
%\blue{should the P have strategies in superscript??}
\begin{align}
\nonumber&\pi_{t+1}(x_{t+1},p_{t+1}^{1:2}) \\
&= \nonumber\Py^{(\tilde{\chi}^1,\tilde{\chi}^2)}[X_{t+1} = \vct{x}_{t+1}, P_{t+1}^{1:2} = \vct{p}_{t+1}^{1:2} \mid \vct{c}_{t+1},\gamma_{1:t}^{1:2}]\\
&= \nonumber\Py^{(\tilde{\chi}^1,\tilde{\chi}^2)}[X_{t+1} = \vct{x}_{t+1}, P_{t+1}^{1:2} = \vct{p}_{t+1}^{1:2}\mid \vct{c}_{t},\gamma_{1:t-1}^{1:2},\vct{z}_{t+1},\gamma_t^{1:2}]\\
\label{b1}&= \frac{\Py^{(\tilde{\chi}^1,\tilde{\chi}^2)}[X_{t+1} = \vct{x}_{t+1}, P_{t+1}^{1:2} = \vct{p}_{t+1}^{1:2},Z_{t+1} = \vct{z}_{t+1} \mid \vct{c}_{t},\gamma_{1:t}^{1:2}]}{\Py^{(\tilde{\chi}^1,\tilde{\chi}^2)}[Z_{t+1} = \vct{z}_{t+1}\mid\vct{c}_t,\gamma_{1:t}^{1:2}]}.
\end{align}

Notice that the expression (\ref{b1}) is well-defined, that is, the denominator is non-zero because of our assumption that the realization $c_{t+1},\gamma_{1:t}^{1:2}$ has non-zero probability of occurrence. Let us consider the numerator in the expression (\ref{b1}). For convenience, we will denote it with $\Py^{(\tilde{\chi}^1,\tilde{\chi}^2)}[\vct{x}_{t+1}, \vct{p}_{t+1}^{1:2},\vct{z}_{t+1} \mid \vct{c}_{t},\gamma_{1:t}^{1:2}]$. We have
\begin{align}
\nonumber&\Py^{(\tilde{\chi}^1,\tilde{\chi}^2)}[\vct{x}_{t+1}, \vct{p}_{t+1}^{1:2},\vct{z}_{t+1} \mid \vct{c}_{t},\gamma_{1:t}^{1:2}]\\
&= \sum_{\vct{x}_t,\vct{p}^{1:2}_t,\vct{u}_t^{1:2}}\pi_t(\vct{x}_t,\vct{p}_t^{1:2})\gamma_t^1( \vct{p}_t^1;\vct{u}_t^1)\gamma_t^2( \vct{p}_t^2; {u}_t^2)\Py^{(\tilde{\chi}^1,\tilde{\chi}^2)}[\vct{x}_{t+1}, \vct{p}_{t+1}^{1:2},\vct{z}_{t+1} \mid \vct{c}_{t},\gamma_{1:t}^{1:2},\vct{x}_t,\vct{p}_t^{1:2},\vct{u}_t^{1:2}]\\
&= \label{condindepeq1}\sum_{\vct{x}_t,\vct{p}^{1:2}_t,\vct{u}_t^{1:2}}\pi_t(\vct{x}_t,\vct{p}_t^{1:2})\gamma_t^1(\vct{p}_t^1;u_t^1)\gamma_t^2( \vct{p}_t^2; u_t^2)\Py[\vct{x}_{t+1}, \vct{p}_{t+1}^{1:2},\vct{z}_{t+1} \mid \vct{x}_t,\vct{p}_t^{1:2},\vct{u}_t^{1:2}]\\
&= \mathscr{P}_t^j(\pi_t,\gamma_t^{1:2};z_{t+1}, x_{t+1},p_{t+1}^{1:2}),
\end{align}
{where $\pi_t$ is the common information belief on $X_t, P_t^1,P_t^2$ at time $t$ given the realization\footnote{Note that the belief $\Py^{(\tilde{\chi}^1,\tilde{\chi}^2)}[x_t, p_t^{1:2} \mid c_t,\gamma_{1:t-1}^{1:2}] = \Py^{(\tilde{\chi}^1,\tilde{\chi}^2)}[x_t, p_t^{1:2} \mid c_t,\gamma_{1:t}^{1:2}]$ because $\gamma_t^i = \tilde{\chi}_t^i( c_t,\gamma_{1:t-1}^{1:2})$, $i=1,2$. } $c_t,\gamma_{1:t-1}^{1:2}$ and $\mathscr{P}_t^j$ is as defined in Definition \ref{fdef}.}
%\begin{align}
%&\vct{X}_{t+1} = f_t(\vct{X}_t,\vct{U}_t^{1:2},\rv{W}_t^s)\\
%&\vct{Z}_{t+1} = \zeta_{t+1}(\vct{P}_t^{1:2},\vct{U}_t^{1:2},\vct{Y}_{t+1}^{1:2})\\
%%\vct{P}_{t+1}^1 &= \xi_{t+1}^1(\vct{P}_t^1,\vct{U}_t^1,\vct{Y}_{t+1}^1)\\
%%\vct{P}_{t+1}^2 &= \xi_{t+1}^2(\vct{P}_t^2,\vct{U}_t^2,\vct{Y}_{t+1}^2)\\
%&\vct{P}_{t+1}^i = \xi_{t+1}^i(\vct{P}_t^i,\vct{U}_t^i,\vct{Y}_{t+1}^i), &i = 1,2, \\
%%\vct{U}_t^1 &= \kappa(\gamma_t^1(\vct{P}_t^1),\rv{V}_t^1)\\
%%\vct{U}_t^2 &= \kappa(\gamma_t^1(\vct{P}_t^2),\rv{V}_t^2)\\
%&\vct{U}_t^i = \kappa(\Gamma_t^i(\vct{P}_t^i),\rv{V}_t^i), &i = 1,2, \\
%%\vct{y}_{t+1}^1 &= h_{t+1}^1(\vct{x}_{t+1},\rv{W}_{t+1}^1)\\
%%\vct{y}_{t+1}^2 &= h_{t+1}^2(\vct{x}_{t+1},\rv{W}_{t+1}^2)
%&\vct{Y}_{t+1}^i = h_{t+1}^i(\vct{X}_{t+1}, U_t^{1:2},\rv{W}_{t+1}^i), &i = 1,2.
%\end{align}
The equality in (\ref{condindepeq1}) is due to the structure of the system dynamics in game $\gm{G}_e$ described by equations (\ref{virdyn1}) -- (\ref{virdyn5}). Similarly, the denominator in (\ref{b1}) satisfies
\begin{align}
0 < \nonumber\Py^{(\tilde{\chi}^1,\tilde{\chi}^2)}[\vct{z}_{t+1} \mid \vct{c}_{t},\gamma_{1:t}^{1:2}] &= \sum_{x_{t+1},p_{t+1}^{1:2}} \mathscr{P}_t^j(\pi_t,\gamma_t^{1:2};z_{t+1}, x_{t+1},p_{t+1}^{1:2})\\
&= \mathscr{P}_t^m(\pi_t,\gamma_t^{1:2};z_{t+1}), \label{marginalprob}
\end{align}
where $\mathscr{P}_t^m$ is as defined is Definition \ref{fdef}. Thus, from equation (\ref{b1}), we have
\begin{equation}\label{beltrans}
\pi_{t+1} = \frac{\mathscr{P}_t^j(\pi_t,\gamma_t^{1:2};z_{t+1},\cdot,\cdot)}{\mathscr{P}_t^m(\pi_t,\gamma_t^{1:2},z_{t+1})} = F_t(\pi_t, \gamma_t^{1:2};\vct{z}_{t+1}),
\end{equation}
where $F_t$ is as defined in Definition \ref{fdef}. {Since the relation (\ref{beltrans}) holds for every realization $c_{t+1}, \gamma_{1:t}^{1:2}$ that has non-zero probability of occurrence under $(\tilde{\chi}^1,\tilde{\chi}^2)$,} we can conclude that  the common information belief $\Pi_t$ evolves \emph{almost surely} as
\begin{align}
\Pi_{t+1} = F_t(\Pi_t, \Gamma_t^{1:2},\vct{Z}_{t+1}), ~~ t \geq 1,
\end{align}
under the strategy profile $(\tilde{\chi}^1,\tilde{\chi}^2)$.
%where $F_t$ is the strategy-independent transformation defined in Definition \ref{fdef}.

The expected cost at time $t$ can be expressed as follows
\begin{align}
\E^{(\tilde{\vct{\chi}}^1,\tilde{\vct{\chi}}^2)}[c_t(X_t,U_t^1,U_t^2)]
& = \E^{(\tilde{\vct{\chi}}^1,\tilde{\vct{\chi}}^2)}[\E[c_t(X_t,U_t^1,U_t^2) \mid C_t,\Gamma_{1:t}^{1:2}]]\\
%&=\E^{(\tilde{\vct{\chi}}^1,\tilde{\vct{\chi}}^2)}[\E\left[c_t\left(\rv{X}_t,\kappa(\Gamma_t^i(\rv{P}_t^{i}),\rv{V}_t^i)\right) \mid \rv{C}_t,\Gamma_{1:t}^{1:2} \right]]\\
&=  \E^{(\tilde{\vct{\chi}}^1,\tilde{\vct{\chi}}^2)}[\tilde{c}_t(\Pi_t,\Gamma_t^1,\Gamma_t^2)],
\end{align}
where the function $\tilde{c}_t$ is as defined as
\begin{align}
\label{tildec}\tilde{c}_t(\pi,\gamma^1,\gamma^2) :=&\sum_{\vct{x}_t,\vct{p}_t^{1:2},\vct{u}_t^{1:2}}c_t(\vct{x}_t, \vct{u}_t^1,\vct{u}_t^2) \pi(\vct{x}_t,\vct{p}_t^1,\vct{p}_t^2)\gamma^1(\vct{p}_t^1 ; \vct{u}_t^1)\gamma^2(\vct{p}_t^2;\vct{u}_t^2).
\end{align}
Therefore, the total cost can be expressed as
\begin{align}
\E^{(\tilde{\vct{\chi}}^1,\tilde{\vct{\chi}}^2)}&\left[\sum_{t=1}^T c_t(\rv{X}_t,\rv{U}_t^{1},\rv{U}_t^2)\right]
%= \; &\E^{(\tilde{\vct{\chi}}^1,\tilde{\vct{\chi}}^2)}\left[\E\left[\sum_{t=1}^T c_t\left(\rv{X}_t,\kappa(\Gamma_t^i(\rv{P}_t^{i}),\rv{V}_t^i)\right) \mid \rv{C}_t,\Gamma_{1:t-1}^{1:2} \right]\right] \nonumber\\
 = \E^{(\tilde{\vct{\chi}}^1,\tilde{\vct{\chi}}^2)}\left[\sum_{t=1}^T \tilde{c}_t(\Pi_t,\Gamma_t^1,\Gamma_t^2)\right].
\end{align}
%The function $\tilde{c}_t$ for a given $\pi_t, \gamma_t^1,\gamma_t^2$ is
%\begin{align}
%\nonumber\sum_{\vct{x}_t,\vct{p}_t^1,\vct{p}_t^2,\vct{u}_t^1,\vct{u}_t^2}\pi_t(\vct{x}_t,\vct{p}_t^1,\vct{p}_t^2)\gamma_t^1(\vct{u}_t^1 &\mid \vct{p}_t^1)\gamma_t^2(\vct{u}_t^2 \mid \vct{p}_t^2)\\
%\label{tildec}& \times c_t(\vct{x}_t, \vct{u}_t^1,\vct{u}_t^2).
%\end{align}
%Note that the cost $\tilde{c}_t$ is a trilinear function in $\pi_t, \gamma_t^1$ and $\gamma_t^2$.

\section{Domain Extension and Proof of Lemma \ref{equiexistlemma}}\label{equiexistlemmaproof}
\subsection{The Max-min Dynamic Program}
The maximizing virtual player (virtual player 2) in game $\gm{G}_e$ solves the following dynamic program. Define $V^l_{T+1}(\pi_{T+1}) = 0$ for every  $\pi_{T+1}$. In a backward inductive manner, at each time $t \leq T$ and for each possible  common information belief $\pi_t$ and prescriptions $\gamma^1_t, \gamma^2_t$, define the lower cost-to-go function $w_t^l$ and the lower value function $V_t^l$ as
\begin{align}
\label{lowercost}w^l_t(\pi_t,\gamma_t^1,\gamma_t^2) 
&\doteq \tilde{c}_t(\pi_t,\gamma_t^1,\gamma_t^2) + \E[V^l_{t+1}(F_t(\pi_{t},\gamma_t^{1:2},\rv{Z}_{t+1}))\mid \pi_t,\gamma_t^{1:2}],\\
\label{maxequa}V_t^l(\pi_t) &\doteq \min_{{\gamma}_t^1} \max_{\gamma_t^2}w_t^l(\pi_t,\gamma_t^1,\gamma_t^2).
\end{align}

\subsection{Domain Extension}
Recall that $\mathcal{S}_t$ is the set of all probability distributions over the finite set $\mathcal{X}_{t} \times \mathcal{P}_{t}^1 \times \mathcal{P}_{t}^2$.
%{Let the common information based belief simplex at time $t$ be $\mathcal{S}_t$. 
Define 
\begin{align}
\bar{\mathcal{S}}_t := \{\alpha\pi_t : 0\leq\alpha\leq 1, \pi_t \in \mathcal{S}_t\}.
\end{align} 
%\blue{Discuss....\\}
The functions $\tilde{c}_t$ in (\ref{tildec}), $\mathscr{P}_t^j$ in (\ref{jointprob}), $\mathscr{P}_t^m$ in (\ref{marginalprob}) and $F_t$ in (\ref{beltrans}) were defined for any $\pi_t \in \mathcal{S}_t$. We will extend the domain of the argument $\pi_t$ in these functions to $\bar{\mathcal{S}}_t$ as follows. For any {$\gamma_t^i \in \mathcal{B}_t^i, i = 1,2$, $z_{t+1} \in \mathcal{Z}_{t+1}$,} $0 \leq \alpha \leq 1$ and $\pi_t \in {\mathcal{S}}_t$, let
\begin{enumerate}
\item $\tilde{c}_t(\alpha\pi_t,\gamma_t^1,\gamma_t^2) := \alpha\tilde{c}_t(\pi_t,\gamma_t^1,\gamma_t^2)$
\item $\mathscr{P}_t^j(\alpha\pi_t,\gamma_t^{1:2}) := \alpha \mathscr{P}_t^j(\pi_t,\gamma_t^{1:2})$
\item $\mathscr{P}_t^m(\alpha\pi_t,\gamma_t^{1:2}) := \alpha \mathscr{P}_t^m(\pi_t,\gamma_t^{1:2})$
\item $
F_t(\alpha\pi_t,\gamma_t^{1:2},z_{t+1}) := 
\begin{cases}
F_t(\pi_t,\gamma_t^{1:2},z_{t+1}) &\text{if } \alpha > 0\\
\bm{0} & \text{if } \alpha = 0,
\end{cases}
$
\end{enumerate}
where $\bm{0}$ is a zero-vector of size $|\mathcal{X}_{t} \times \mathcal{P}_{t}^1 \times \mathcal{P}_{t}^2|$.

% \blue{Notice that the definitions above }

%We can extend the domain of the argument $\pi_t$ in the function $\tilde{c}_t$ in (\ref{tildec}) to $\bar{\mathcal{S}}_t$ without changing its definition. Similarly, we can extend the domain of the argument $\pi_t$  in the transformations $\mathscr{P}_t^j$ and $\mathscr{P}_t^m$ to $\bar{\mathcal{S}}_t$ with the same definitions as in equations (\ref{jointprob}) and (\ref{marginalprob}). Using the definition in (\ref{beltrans}), the domain of the argument $\pi_t$ in the transformation $F_t$ can also be extended to $\bar{\mathcal{S}}_t$ with the following modification: for any pair of prescriptions $\gamma_{t}^{1:2}$ and common observation $z_{t+1}$,
%\begin{equation}
%F_t(\bm{0}, \gamma_t^{1:2},z_{t+1}) = \bm{0}.
%\end{equation}
Having extended the domain of the above  functions, we can also extend the domain of the argument $\pi_t$ in the functions $w_t^u(\cdot), w_t^l(\cdot), V_t^u(\cdot), V_t^l(\cdot)$ defined in the dynamic programs of  Section \ref{dpsec}. First, for any  $0 \leq \alpha \leq 1$ and  $\pi_{T+1} \in {\mathcal{S}}_{T+1}$, define $V^u_{T+1}(\alpha\pi_{T+1}) :=0$.  We can then define the following functions for every  $t \leq T$ in a backward inductive manner:  For any {$\gamma_t^i \in \mathcal{B}_t^i, i = 1,2$, } $0 \leq \alpha \leq 1$ and $\pi_t \in {\mathcal{S}}_t$, let
\begin{align}\label{eq:w_ext}
w^u_t (\alpha\pi_t,\gamma_t^1,\gamma_t^2) &:= \tilde{c}_t(\alpha\pi_t,\gamma_t^1,\gamma_t^2) + \sum_{z_{t+1}}\big[\mathscr{P}_t^m(\alpha\pi_t,\gamma_t^{1:2};z_{t+1})V^u_{t+1}(F_t(\alpha\pi_t,\gamma_t^{1:2},z_{t+1}))\big] \\
V^u_t(\alpha\pi_t) &:= \inf_{{\gamma}_t^1} \sup_{\gamma_t^2}w_t^u(\alpha\pi_t,\gamma_t^1,\gamma_t^2).
\end{align}
Note that when $\alpha=1$, the above definition of $w^u_t$ is equal to the definition of $w^u_t$ in equation (\ref{minequa}) of the dynamic program. We can similarly extend $w^l_t$ and $V^l_t$. 
%without modifying their definitions (\blue{In Section \ref{dpsec}, the second term is an expectation. Expectation for $\pi_t$ that is not a distribution is not well-defined. Should we replace the expectation in the definition with a sum in Section \ref{dpsec}?} \red{This might be a messy sum. Perhaps we should just explicitly extend them here?}). 
\begin{lemma}\label{scalinglemma}
For any constant $0 \leq \alpha \leq 1$ and any $\pi_t \in \bar{\mathcal{S}}_t$, we have $\alpha V_t^u(\pi_t) = V_t^u(\alpha \pi_t)$ and $\alpha V_t^l(\pi_t) = V_t^l(\alpha \pi_t)$.
%\begin{align}
%\alpha V_t^u(\pi_t) &= V_t^u(\alpha \pi_t)\\
%\alpha V_t^l(\pi_t) &= V_t^l(\alpha \pi_t).
%\end{align}
\end{lemma}
\begin{proof}
{The proof can be easily obtained from the above definitions of the extended functions. }
\end{proof}

\subsection{Proof of Lemma \ref{equiexistlemma}}
We will first prove inductively that the function $w_t^u: \bar{\mathcal{S}}_t\times\mathcal{B}_t^1\times \mathcal{B}_t^2 \to \R$ is continuous. Let us as assume that value function $V_{t+1}^u$ is continuous for some $t \leq T$. Note that this assumption clearly holds at $t = T$. At time $t$, we have
\begin{align}
w^u_t(\pi_t,\gamma_t^1,\gamma_t^2) &= \tilde{c}_t(\pi_t,\gamma_t^1,\gamma_t^2) + \E[V^u_{t+1}(F_t(\pi_{t},\gamma_t^{1:2},\rv{Z}_{t+1}))\mid \pi_t,\gamma_t^{1:2}]\\
&=  \tilde{c}_t(\pi_t,\gamma_t^1,\gamma_t^2) +\sum_{z_{t+1}}\mathscr{P}_t^m( \pi_t,\gamma_{t}^{1:2};z_{t+1})V_{t+1}^u(F_t(\pi_{t},\gamma_t^{1:2},{z}_{t+1}))\\
&=  \tilde{c}_t(\pi_t,\gamma_t^1,\gamma_t^2) +\sum_{z_{t+1}}V_{t+1}^u(\mathscr{P}_t^j(\pi_{t},\gamma_t^{1:2};{z}_{t+1},\cdot,\cdot)),\label{simpval}
\end{align}
where the last inequality follows from the homogeneity property of the value functions in Lemma \ref{scalinglemma} and the structure of the belief update in \eqref{beltrans}. The first term in \eqref{simpval} is clearly continuous (see \ref{tildec}). Also, the transformation $\mathscr{P}_t^j(\pi_t,\gamma_t^{1:2};z_{t+1},\cdot,\cdot)$ defined in \eqref{jointprob} is a continuous function. Therefore, by our induction hypothesis, the composition $V_{t+1}^u(\mathscr{P}_t^j(\cdot))$ is continuous in $(\pi_t,\gamma_t^{1:2})$ for every common observation $z_{t+1}$. Thus, $w_t^u$ is continuous in its arguments. To complete our inductive argument, we need to show that $V_t^u$ is a continuous function and to this end, we will use the Berge's maximum theorem (Lemma 17.31) in \cite{guide2006infinite}. Since $w_t^u$ is continuous and $\mathcal{B}_t^2$ is compact, we can use the maximum theorem to conclude that the following function
\begin{align}\label{vtu}
v_t^u(\pi_t,\gamma_t^1) \doteq \sup_{\gamma_t^2}w^u_t(\pi_t,\gamma_t^1,\gamma_t^2),
\end{align}
is continuous in $\pi_t,\gamma_t^1$. Once again, we can use the maximum theorem to conclude that
\begin{align}
V_t^u(\pi_t) = \inf_{\gamma_t^1} v_t^u(\pi_t,\gamma_t^1)
\end{align}
is continuous in $\pi_t$. This concludes our inductive argument. Also, because of the continuity of $v_t^u$ in \eqref{vtu}, we can use the measurable selection condition (see Condition 3.3.2 in \cite{hernandez2012discrete}) to prove the existence of the measurable mapping $\Xi_t^1(\pi_t)$ as defined in Lemma \ref{equiexistlemma}. A similar argument can be made to establish the existence of a maxminimizer and a measurable mapping $\Xi_t^2(\pi_t)$ as defined in Lemma \ref{equiexistlemma}. This concludes the proof of Lemma \ref{equiexistlemma}.

\section{Proof of Theorem \ref{dp}}\label{dpproof}
 Let us first define a distribution $\tilde{\Pi}_t$ over the space $\mathcal{X}_t \times \P_t^1 \times \P_t^2$ in the following manner. The distribution $\tilde{\Pi}_t$, given $C_t,\Gamma_{1:t-1}^{1:2}$, is recursively obtained using the following relation
\begin{align}
\tilde{\Pi}_1(x_1,p_1^1,p_1^2) &= \Py[X_1 = x_1,P_1^1 = p_1^1,P_1^2 = p_1^2 \mid C_1 ] ~ \forall\; x_1,p_1^1,p_1^2,\\
\tilde{\Pi}_{\tau+1} &= F_\tau(\tilde{\Pi}_\tau, \Gamma_\tau^{1},\Gamma_\tau^2,\vct{Z}_{\tau+1}), ~~ \tau \geq 1,
\end{align}
where $F_\tau$ is as defined in Definition \ref{fdef} in Appendix \ref{infstateproof}. We refer to this distribution as the strategy-independent common information belief (SI-CIB).

Let $\tilde{\chi}^1 \in \tilde{\mathcal{H}}^1$ be any strategy for virtual player 1 in game $\gm{G}_e$. Consider the problem of obtaining virtual player 2's best response to the strategy $\tilde{\chi}^1$ with respect to the cost $\mathcal{J}(\tilde{\chi}^1 ,\tilde{\chi}^2)$ defined in \eqref{eq:virtualJ}. This problem can be formulated as a Markov decision process (MDP) with common information and prescription history $C_t,\Gamma_{1:t-1}^{1:2}$ as the state. The control action at time $t$ in this MDP is $\Gamma_t^2$, which is selected based on the information $C_t,\Gamma_{1:t-1}^{1:2}$ using strategy $\tilde{\chi}^2 \in \mathcal{H}^2$.
The evolution of the state $C_t,\Gamma_{1:t-1}^{1:2}$ of this MDP is as follows
\begin{align}
\{C_{t+1},\Gamma_{1:t}^{1:2}\} = \{C_t,Z_{t+1},\Gamma_{1:t-1}^{1:2}, \tilde{\chi}^1_t(C_t,\Gamma_{1:t-1}^{1:2}),\Gamma_t^2\},
\end{align}
where 
\begin{align}\label{zstatevol}
\Py^{(\tilde{\vct{\chi}}^1,\tilde{\vct{\chi}}^2)}[Z_{t+1} = z_{t+1} \mid C_t,\Gamma_{1:t-1}^{1:2}, \Gamma_t^2] = \mathscr{P}_t^m[\tilde{\Pi}_t,\Gamma_t^1,\Gamma_t^2;z_{t+1}],
\end{align}
almost surely. Here, $\Gamma_t^1 = \tilde{\chi}^1_t(C_t,\Gamma_{1:t-1}^{1:2})$ and the transformation $\mathscr{P}_t^m$ is as defined in Definition \ref{fdef} in Appendix \ref{infstateproof}.  Notice that due to Lemma \ref{infstate}, the common information belief $\Pi_t$ associated with any strategy profile ${(\tilde{\vct{\chi}}^1,\tilde{\vct{\chi}}^2)}$ is equal to $\tilde{\Pi_t}$ almost surely. This results in the state evolution equation in \eqref{zstatevol}.
The objective of this MDP is to maximize, for a given $\tilde{\chi}^1$, the following cost
\begin{align}
\E^{(\tilde{\vct{\chi}}^1,\tilde{\vct{\chi}}^2)}\left[\sum_{t=1}^T \tilde{c}_t(\tilde{\Pi}_t,\Gamma_t^1,\Gamma_t^2)\right],
\end{align}
where $\tilde{c}_t$ is as defined in equation (\ref{tildec}). Due to Lemma \ref{infstate}, the total expected cost defined above is equal to the cost ${\mathcal{J}}(\tilde{\vct{\chi}}^1,\tilde{\vct{\chi}}^2)$ defined in \eqref{eq:virtualJ}.

The MDP described above can be solved using the following dynamic program. For every realization of virtual players' information $c_{T+1},\gamma_{1:T}^{1:2}$, define $$V^{\tilde{\chi}^1}_{T+1}(c_{T+1},\gamma_{1:T}^{1:2}) := 0.$$
In a backward inductive manner, for each time $t \leq T$ and each realization $c_{t},\gamma_{1:t-1}^{1:2}$, define
\begin{align}
V^{\tilde{\chi}^1}_{t}(c_{t},\gamma_{1:t-1}^{1:2}) &:= \sup_{\gamma_t^2}[\tilde{c}_t(\tilde{\pi}_t,\gamma_t^1,\gamma_t^2)+ \E[V^{\tilde{\chi}^1}_{t+1}(c_{t},Z_{t+1},\gamma_{1:t}^{1:2}) \mid c_t,\gamma_{1:t}^{1:2}]]\label{thm3e1},
\end{align}
where $\gamma_t^1 = \tilde{\chi}^1_t(c_{t},\gamma_{1:t-1}^{1:2})$ and {$\tilde{\pi}_t$ is the SI-CIB associated with the information $c_{t},\gamma_{1:t-1}^{1:2}$}. Note that the measurable selection condition (see condition 3.3.2 in \cite{hernandez2012discrete}) holds for the dynamic program described above. Thus, the value functions $V^{\tilde{\chi}^1}_{t}(\cdot)$ are measurable and there exists a measurable best-response strategy for player 2 which is a solution to the dynamic program described above. Therefore, we have
\begin{align}
\label{dpcommon}\sup_{\tilde{\chi}^2}\mathcal{J}(\tilde{\chi}^1,\tilde{\chi}^2) = \E V^{\tilde{\chi}^1}_{1}(C_{1}).
\end{align}
\begin{claim}\label{upperclaim}
For any strategy $\tilde{\chi}^1 \in \tilde{\mathcal{H}}^1$ and for any realization of virtual players' information $c_{t},\gamma_{1:t-1}^{1:2}$, we have
\begin{align}
\label{uppervalineq}V^{\tilde{\chi}^1}_{t}(c_{t},\gamma_{1:t-1}^{1:2}) \geq V_t^u(\tilde{\pi}_t),
\end{align}
where $V_t^u$ is as defined in (\ref{minequa}) and $\tilde{\pi}_t$ is the SI-CIB belief associated with the instance $c_{t},\gamma_{1:t-1}^{1:2}$. As a consequence, we have
\begin{align}\label{upineq}
\sup_{\tilde{\chi}^2}\mathcal{J}(\tilde{\chi}^1,\tilde{\chi}^2) \geq \E V^{u}_{1}(\Pi_{1}).
\end{align}
\end{claim}
\begin{proof}
The proof is by backward induction. Clearly, the claim is true at time $t = T+1$. Assume that the claim is true for all times greater than $t$. Then we have
\begin{align*}
V^{\tilde{\chi}^1}_{t}(c_{t},\gamma_{1:t-1}^{1:2}) 
&= \sup_{\gamma_t^2}[\tilde{c}_t(\tilde{\pi}_t,\gamma_t^1,\gamma_t^2) + \E[V^{\tilde{\chi}^1}_{t+1}(c_{t},Z_{t+1},\gamma_{1:t}^{1:2}) \mid c_t,\gamma_{1:t}^{1:2}]]\\
&\geq \sup_{\gamma_t^2}[\tilde{c}_t(\tilde{\pi}_t,\gamma_t^1,\gamma_t^2)
 + \E[V^{u}_{t+1}(F_t(\tilde{\pi}_t,\gamma_t^{1:2},Z_{t+1})) \mid c_t,\gamma_{1:t}^{1:2}]]\\
&\geq V_t^u(\tilde{\pi}_t).
\end{align*}
The first equality follows from the definition in (\ref{thm3e1}) and the inequality after that follows from the induction hypothesis. The last inequality is a consequence of the definition of the value function $V_t^u$. This completes the induction argument.
Further, using Claim \ref{upperclaim} and the result in \eqref{dpcommon}, we have
\begin{align*}
\sup_{\tilde{\chi}^2}\mathcal{J}(\tilde{\chi}^1,\tilde{\chi}^2) = \E V^{\tilde{\chi}^1}_{1}(C_{1}) \geq \E V^{u}_{1}(\tilde{\Pi}_{1}) = \E V^{u}_{1}(\Pi_{1}).
\end{align*}
\end{proof}

We can therefore say that
\begin{align}\label{suone}
S^u(\gm{G}_e) &= \inf_{\tilde{\chi}^1}\sup_{\tilde{\chi}^2}\mathcal{J}(\tilde{\chi}^1,\tilde{\chi}^2) \geq \inf_{\tilde{\chi}^1}\E V^{u}_{1}(\Pi_{1}) = \E V^{u}_{1}(\Pi_{1}).
\end{align}
Further, for the strategy $\tilde{\chi}^{1*}$ defined in Definition \ref{stratdef}, the inequalities (\ref{uppervalineq}) and (\ref{upineq}) hold with equality. We can prove this using an inductive argument similar to the one used to prove Claim \ref{upperclaim}. Therefore, we have
\begin{align}\label{sutwo}
S^u(\gm{G}_e) &= \inf_{\tilde{\chi}^1}\sup_{\tilde{\chi}^2}\mathcal{J}(\tilde{\chi}^1,\tilde{\chi}^2) \leq \sup_{\tilde{\chi}^2}\mathcal{J}(\tilde{\chi}^{1*},\tilde{\chi}^2)  = \E V^{\tilde{\chi}^{1*}}_{1}(C_{1})  = \E V^{u}_{1}(\Pi_{1}).
\end{align}
Combining \eqref{suone} and \eqref{sutwo}, we have
\begin{align*}
S^u(\gm{G}_e) = \E V^{u}_{1}(\Pi_{1}).
\end{align*}
Thus, the inequality in \eqref{sutwo} holds with equality which leads us to the result that the strategy $\tilde{\chi}^{1*}$ is a min-max strategy in game $\gm{G}_e$.
A similar argument can be used to show that
\begin{align*}
S^l(\gm{G}_e) = \E V^{l}_{1}(\Pi_{1}),
\end{align*}
and that the strategy $\tilde{\chi}^{2*}$ defined in Definition \ref{stratdef} is a max-min strategy in game $\gm{G}_e$.

\section{Information Structures}\label{additionalinf}

{\subsection{Team 2 does not control the state}
Consider an instance of Game $\gm{G}$ in which the state evolution and the players' observations are given by
\begin{align*}
&\rv{X}_{t+1} = f_t(\rv{X}_t, \rv{U}_t^1,\rv{W}_t^s);&\rv{Y}_t^{i,j} = h_t^{i,j}(\rv{X}_t, \rv{U}_{t-1}^1,\rv{W}_t^{i,j}).
\end{align*}
Further, let the information structure of the players be such that the common and private information  evolve as
\begin{align}
\rv{Z}_{t+1} &= C_{t+1}\setminus C_t = \zeta_{t+1}(\rv{P}_t^{1:2},\rv{U}_t^{1},\rv{Y}_{t+1}^{1:2})\\
\rv{P}^{i}_{t+1} &= \xi_{t+1}^{i}(\rv{P}_t^{1:2},\rv{U}_t^{1},\rv{Y}_{t+1}^{1:2}).
\end{align}
In this model, Team 2's actions do not affect the system evolution and Team 2's past actions are not used by any of the players to select their current actions.
Any information structure that satisfies these conditions satisfies Assumption \ref{onesidecontrolassum}, see Appendix \ref{infproof3} for a proof.}

\subsection{Global and local states}
Consider a system in which the system state $X_t = (X_t^0,X_t^1,X_t^2)$ comprises of a global state $X_t^0$ and a local state $X_t^i = (X_t^{i,1},\dots,X_t^{i,N_i})$ for Team $i = 1,2$. The state evolution is given by
\begin{align}
\rv{X}^0_{t+1} &= f_t^1(\rv{X}_t^0, X_t^1, \rv{U}_t^1, U_t^2,\rv{W}_t^{s,1})\\
\rv{X}^1_{t+1} &= f_t^2(\rv{X}_t^0, X_t^1, \rv{U}_t^1, U_t^2,\rv{W}_t^{s,2})\\
\rv{X}^2_{t+1} &= f_t^3(\rv{X}_t^0, \rv{U}_t^1, U_t^2,\rv{W}_t^{s,3}).
%(\rv{X}^0_{t+1},\rv{X}^1_{t+1},\rv{X}^2_{t+1}) = (f_t^1(\rv{X}_t^0, X_t^1, \rv{U}_t^1, U_t^2,\rv{W}_t^{s,1}),f_t^1(\rv{X}_t^0, X_t^1, \rv{U}_t^1, U_t^2,\rv{W}_t^{s,2}),f_t^1(\rv{X}_t^0, X_t^1, \rv{U}_t^1, U_t^2,\rv{W}_t^{s,3})).
\end{align}
Note that Team 2's current local state does not affect the state evolution. Further, we have
\begin{align}
C_t &=\{X_{1:t}^0,U_{1:t-1}^1,U_{1:t-1}^2\}\\
P_t^{1,j} &= \{X_t^{1,j}\} \quad \forall j = 1,\dots,N_1\\
P_t^{2,j} &= \{X_t^{2,j}\} \quad \forall j = 1,\dots,N_2.
\end{align}
Under this system dynamics and information structure, we can prove that Assumption \ref{onesidecontrolassum} holds. The proof of this is provided in Appendix \ref{infproof4}.

\section{Information Structures that Satisfy Assumption \ref{onesidecontrolassum}: Proofs}\label{infproof}
For each model in \ref{infexample}, we will accordingly construct transformations $Q_t: \mathcal{S}_t \times \mathcal{B}_t^{1} \times \mathcal{Z}_{t+1} \to \R^{|\mathcal{X}_{t+1} \times \mathcal{P}_{t+1}^{1:2}|}$ and $R_t:\mathcal{S}_t \times \mathcal{B}_t^{1} \times \mathcal{Z}_{t+1} \to \R$ at each time $t$ such that
\begin{subequations}
\begin{align}
\label{constra} \mathscr{P}_t^m(\pi_t,\gamma_t^{1:2};z_{t+1}) &\leq R_t(\pi_t,\gamma_t^{1},z_{t+1}) \quad \forall \pi_t,\gamma_t^{1:2},z_{t+1},\\
\label{constrb}\frac{\mathscr{P}_t^j(\pi_t,\gamma_t^{1:2},z_{t+1},\cdot)}{\mathscr{P}_t^m(\pi_t,\gamma_t^{1:2};z_{t+1})} &= \frac{Q_t(\pi_t,\gamma_t^{1},z_{t+1})}{R_t(\pi_t,\gamma_t^{1},z_{t+1})} \quad \forall \mathscr{P}_t^m(\pi_t,\gamma_t^{1:2};z_{t+1}) > 0.
\end{align}
\end{subequations}
Note that the transformations $Q_t$ and $R_t$ do \emph{not} make use of virtual player 2's prescription $\gamma_t^2$. Following the methodology in Definition \ref{fdef}, we define $F_t$ as
\begin{align}
\label{tempfdef}F_t(\pi_t,\gamma_t^{1:2},z_{t+1})&= 
\begin{cases}
\frac{\mathscr{P}_t^j(\pi_t,\gamma_t^{1:2},z_{t+1},\cdot)}{\mathscr{P}_t^m(\pi_t,\gamma_t^{1:2};z_{t+1})} &\text{if } \mathscr{P}_t^m(\pi_t,\gamma_t^{1:2};z_{t+1}) > 0\\
G_t(\pi_t,\gamma_t^{1:2},z_{t+1}) & \text{otherwise},
\end{cases}
\end{align}
where the transformation $G_t$ is chosen to be
\begin{align}
G_t(\pi_t,\gamma_t^{1:2},z_{t+1})&= 
\begin{cases}
\frac{Q_t(\pi_t,\gamma_t^{1},z_{t+1})}{R_t(\pi_t,\gamma_t^{1},z_{t+1})} &\text{if } R_t(\pi_t,\gamma_t^{1},z_{t+1}) > 0\\
\mathscr{U}(\X_{t+1}\times\mathcal{P}_{t+1}^{1:2}) & \text{otherwise},
\end{cases}
\end{align}
where $\mathscr{U}(\X_{t+1}\times \mathcal{P}_{t+1}^{1:2})$ is the uniform distribution over the space $\X_{t+1}\times\mathcal{P}_{t+1}^{1:2}$. Since the transformations $Q_t$ and $R_t$ satisfy \eqref{constra} and \eqref{constrb}, we can simplify the expression for the transformation $F_t$ in (\ref{tempfdef}) to obtain the following
\begin{align}
F_t(\pi_t,\gamma_t^{1:2},z_{t+1})&= 
\begin{cases}
\frac{Q_t(\pi_t,\gamma_t^{1},z_{t+1})}{R_t(\pi_t,\gamma_t^{1},z_{t+1})} &\text{if } R_t(\pi_t,\gamma_t^{1},z_{t+1}) > 0\\
\mathscr{U}(\X_{t+1}\times \mathcal{P}_{t+1}^{1:2}) & \text{otherwise}.
\end{cases}
\end{align}
This concludes the construction of an update rule $F_t$ in the class $\mathscr{B}$ that does not use virtual player 2's prescription $\gamma_t^2$. We will now describe the the construction of the transformations $Q_t$ and $R_t$ for each information structure in Section \ref{infexample}.

\subsection{All players in Team 2 have the same information}\label{infproof1}
In this case, Team 2 does not have any private information and any instance of the common observation $z_{t+1}$ includes Team 2's action at time $t$ (denote it with $\hat{u}_t^2$). The corresponding transformation $\mathscr{P}_t^j$ (see Definition \ref{fdef}) has the following form.
\begin{align}
&\mathscr{P}_t^j(\pi_t,\gamma_t^{1:2};z_{t+1},x_{t+1},p_{t+1}^1) \\
&= \sum_{\vct{x}_t,\vct{p}^{1:2}_t,\vct{u}_t^{1:2}}\pi_t(\vct{x}_t,\vct{p}_t^{1:2})\gamma_t^1(\vct{p}_t^1;u_t^1)\gamma_t^2( \vct{p}_t^2; u_t^2)\Py[\vct{x}_{t+1}, \vct{p}_{t+1}^{1:2},\vct{z}_{t+1} \mid \vct{x}_t,\vct{p}_t^{1:2},\vct{u}_t^{1:2}]\\
&= \gamma_t^2( \varnothing; \hat{u}_t^2)\sum_{\vct{x}_t,\vct{p}^{1}_t,\vct{u}_t^{1}}\pi_t(\vct{x}_t,\vct{p}_t^{1})\gamma_t^1(\vct{p}_t^1;u_t^1)\Py[\vct{x}_{t+1}, \vct{p}_{t+1}^{1},\vct{z}_{t+1} \mid \vct{x}_t,\vct{p}_t^{1},\vct{u}_t^{1},\hat{u}_t^2]\\
\label{qdef1}&=:  \gamma_t^2( \varnothing; \hat{u}_t^2)Q_t(\pi_t,\gamma_t^{1},z_{t+1};x_{t+1},p_{t+1}^1).
\end{align}
Here, we use the fact that Team 2 does not have any private information and $\hat{u}_t^2$ is a part of the common observation $z_{t+1}$. Similarly, the corresponding transformation $\mathscr{P}_t^m$ (see Definition \ref{fdef}) has the following form.
\begin{align}
\mathscr{P}_t^m(\pi_t,\gamma_t^{1:2};z_{t+1}) &=  \gamma_t^2( \varnothing; \hat{u}_t^2)\sum_{x_{t+1},p_{t+1}^1} Q_t(\pi_t,\gamma_t^{1},z_{t+1};x_{t+1},p_{t+1}^1)\\
\label{rdef1}&=:  \gamma_t^2( \varnothing; \hat{u}_t^2) R_t(\pi_t,\gamma_t^{1},z_{t+1}).
\end{align}
Using the results \eqref{qdef1} and \eqref{rdef1}, we can easily conclude that the transformations $Q_t$ and $R_t$ defined above satisfy the conditions \eqref{constra} and \eqref{constrb}.

\subsection{Team 2's observations become common information with a delay of one-step}\label{infproof2}
In this case, any instance of the common observation $z_{t+1}$ includes Team 2's private information at time $t$ (denote it with $\hat{p}_t^2$) and Team 2's action at time $t$ (denote it with $\hat{u}_t^2$). The corresponding transformation $\mathscr{P}_t^j$ (see Definition \ref{fdef}) has the following form.
\begin{align}
&\mathscr{P}_t^j(\pi_t,\gamma_t^{1:2};z_{t+1},x_{t+1},p_{t+1}^{1:2}) \\
&= \sum_{\vct{x}_t,\vct{p}^{1:2}_t,\vct{u}_t^{1:2}}\pi_t(\vct{x}_t,\vct{p}_t^{1:2})\gamma_t^1(\vct{p}_t^1;u_t^1)\gamma_t^2( \vct{p}_t^2; u_t^2)\Py[\vct{x}_{t+1}, \vct{p}_{t+1}^{1:2},\vct{z}_{t+1} \mid \vct{x}_t,\vct{p}_t^{1:2},\vct{u}_t^{1:2}]\\
&= \gamma_t^2( \hat{p}_t^2; \hat{u}_t^2)\sum_{\vct{x}_t,\vct{p}^{1}_t,\vct{u}_t^{1}}\pi_t(\vct{x}_t,\vct{p}_t^{1},\hat{p}_t^2)\gamma_t^1(\vct{p}_t^1;u_t^1)\Py[\vct{x}_{t+1}, \vct{p}_{t+1}^{1:2},\vct{z}_{t+1} \mid \vct{x}_t,\vct{p}_t^{1},\hat{p}_t^2,\vct{u}_t^{1},\hat{u}_t^2]\\
\label{qdef2}&=: \gamma_t^2( \hat{p}_t^2; \hat{u}_t^2)Q_t(\pi_t,\gamma_t^{1},z_{t+1};x_{t+1},p_{t+1}^{1:2}).
\end{align}
Here, we use the fact that both $\hat{p}_t^2$ and $\hat{u}_t^2$ are part of the common observation $z_{t+1}$. Similarly, the corresponding transformation $\mathscr{P}_t^m$ (see Definition \ref{fdef}) has the following form.
\begin{align}
\mathscr{P}_t^m(\pi_t,\gamma_t^{1:2};z_{t+1}) &=   \gamma_t^2( \hat{p}_t^2; \hat{u}_t^2)\sum_{x_{t+1},p_{t+1}^{1:2}} Q_t(\pi_t,\gamma_t^{1},z_{t+1};x_{t+1},p_{t+1}^{1:2})\\
\label{rdef2}&=:  \gamma_t^2( \hat{p}_t^2; \hat{u}_t^2) R_t(\pi_t,\gamma_t^{1},z_{t+1}).
\end{align}
Using the results \eqref{qdef2} and \eqref{rdef2}, we can conclude that the transformations $Q_t$ and $R_t$ defined above satisfy the conditions \eqref{constra} and \eqref{constrb}.

\subsection{Team 2 does not control the state}\label{infproof3}
In this case, the corresponding transformation $\mathscr{P}_t^j$ (see Definition \ref{fdef}) has the following form.
\begin{align}
&\mathscr{P}_t^j(\pi_t,\gamma_t^{1:2};z_{t+1},x_{t+1},p_{t+1}^{1:2})  \\
&= \sum_{\vct{x}_t,\vct{p}^{1:2}_t,\vct{u}_t^{1:2}}\pi_t(\vct{x}_t,\vct{p}_t^{1:2})\gamma_t^1(\vct{p}_t^1;u_t^1)\gamma_t^2( \vct{p}_t^2; u_t^2)\Py[\vct{x}_{t+1}, \vct{p}_{t+1}^{1:2},\vct{z}_{t+1} \mid \vct{x}_t,\vct{p}_t^{1:2},\vct{u}_t^{1:2}]\\
\label{u2indep}&= \sum_{\vct{x}_t,\vct{p}^{1:2}_t,\vct{u}_t^{1:2}}\pi_t(\vct{x}_t,\vct{p}_t^{1:2})\gamma_t^1(\vct{p}_t^1;u_t^1)\gamma_t^2( \vct{p}_t^2; u_t^2)\Py[\vct{x}_{t+1}, \vct{p}_{t+1}^{1:2},\vct{z}_{t+1} \mid \vct{x}_t,\vct{p}_t^{1:2},\vct{u}_t^{1}]\\
&=\label{sum1}\sum_{\vct{x}_t,\vct{p}^{1:2}_t,\vct{u}_t^{1}}\pi_t(\vct{x}_t,\vct{p}_t^{1:2})\gamma_t^1(\vct{p}_t^1;u_t^1)\Py[\vct{x}_{t+1}, \vct{p}_{t+1}^{1:2},\vct{z}_{t+1} \mid \vct{x}_t,\vct{p}_t^{1:2},\vct{u}_t^{1}]\\
\label{qdef3}&=:Q_t(\pi_t,\gamma_t^{1},z_{t+1};x_{t+1},p_{t+1}^{1:2}).
\end{align}
Note that \eqref{u2indep} holds because $u_t^2$ does not influence the evolution of state, common and private information and \eqref{sum1} follows by summing over $u_t^2$. Similarly, the corresponding transformation $\mathscr{P}_t^m$ (see Definition \ref{fdef}) has the following form.
\begin{align}
\mathscr{P}_t^m(\pi_t,\gamma_t^{1:2};z_{t+1}) &=  \sum_{x_{t+1},p_{t+1}^{1:2}} Q_t(\pi_t,\gamma_t^{1},z_{t+1};x_{t+1},p_{t+1}^{1:2})\\
\label{rdef3}&=:  R_t(\pi_t,\gamma_t^{1},z_{t+1}).
\end{align}
Using the results \eqref{qdef3} and \eqref{rdef3}, we can conclude that the transformations $Q_t$ and $R_t$ satisfy the conditions \eqref{constra} and \eqref{constrb}.
\subsection{Global and local states}\label{infproof4}
In this case, the private information variables of each player are part of the system state $X_t$ because $P_t^1 = X_t^1$ and $P_t^2 = X_t^2$. Therefore, the common information belief $\Pi_t$ is formed only on the system state. 
%We will inductively prove that the common information belief on the state
%\begin{align}
%\Pi_t(x_t^0, x_t^1,x_t^2) = \Py[X_t^0 = x_t^0, X_t^1 = x_t^1 \mid C_t, \Gamma_{1:t-1}^{1:2}] \Py[X_t^2 = x_t^2 \mid X_t^0 =x_t^0, C_t, \Gamma_{1:t-1}^{1:2}],
%\end{align}
%almost surely.
Let us first define a collection of beliefs with a particular structure
\begin{align} 
\mathcal{S}'_t \doteq \left\{\pi \in \mathcal{S}_t: \pi(x_0,x_1,x_2) = \mathbbm{1}_{\hat{x}}(x_0)\pi^{1|0}(x_1\mid x_0)\pi^2(x_2) ~~\forall x_0,x_1,x_2\right\},
\end{align}
where $\pi^{1|0}$ and $\pi^2$ denote the respective conditional and marginal distributions formed using the joint distribution $\pi$, and $\hat{x}$ is some realization of the global state $X_t^0$. Under the dynamics and information structure in this case, we can show that the belief $\pi_t$ computed recursively using the transformation $F_t$ as in \eqref{beltrans} always lies in $\mathcal{S}'_t$ (for an appropriate initial distribution on the state $X_1$). Therefore, we restrict our attention to beliefs in the restricted set $\mathcal{S}_t'$. Let $\pi_t \in \mathcal{S}'_t$ such that
\begin{align}
    \pi_t(x_0,x_1,x_2) = \mathbbm{1}_{\hat{x}}(x_0)\pi_t^{1|0}(x_1\mid x_0)\pi_t^2(x_2)~~\forall x_0,x_1,x_2.
\end{align}

In this case, any instance of the common observation $z_{t+1}$ comprises of the global state $\hat{x}_{t+1}^0$ and players' actions (denoted by $\hat{u}_t^1,\hat{u}_t^2$ for Teams 1 and 2 respectively). The corresponding transformation $\mathscr{P}_t^j$ (see Definition \ref{fdef}) has the following form.
\begin{align}
&\mathscr{P}_t^j(\pi_t,\gamma_t^{1:2};z_{t+1},x_{t+1})  \\
&= \sum_{\vct{x}_t,\vct{u}_t^{1:2}}\pi_t(\vct{x}_t)\gamma_t^1(\vct{x}_t^1;u_t^1)\gamma_t^2( \vct{x}_t^2; u_t^2)\Py[\vct{x}_{t+1},\vct{z}_{t+1} \mid \vct{x}_t,\vct{u}_t^{1:2}]\\
%&=\sum_{\vct{x}_t,\vct{u}_t^{1:2}}\pi^{0,1}_t(\vct{x}^0_t,x_t^1)\pi_t^2(x_t^2 \mid x_t^0)\gamma_t^1(\vct{x}_t^1;u_t^1)\gamma_t^2( \vct{x}_t^2; u_t^2)\Py[\vct{x}^0_{t+1},x_{t+1}^1,u_t^{1:2} \mid \vct{x}^0_t,x_t^1,\vct{u}_t^{1:2}]\Py[\vct{x}^2_{t+1} \mid \vct{x}^0_t,\vct{u}_t^{1:2}]\\
%&= \left(\sum_{\vct{x}^1_t}\pi^{0,1}_t(\vct{x}^0_t,x_t^1)\gamma_t^1(\vct{x}_t^1;u_t^1)\Py[\vct{x}^0_{t+1},x_{t+1}^1,u_t^{1:2} \mid \vct{x}^0_t,x_t^1,\vct{u}_t^{1:2}]\right) \left(\sum_{\vct{x}^2_t}\pi_t^2(x_t^2 \mid x_t^0)\gamma_t^2( \vct{x}_t^2; u_t^2)\Py[\vct{x}^2_{t+1} \mid \vct{x}^0_t,\vct{u}_t^{1:2}]\right)\\
\label{probeq}&= \left(\sum_{x_t^0,{x}^1_t}\mathbbm{1}_{\hat{x}}(x_t^0)\pi^{1|0}_t(x_t^1\mid x_t^0)\gamma_t^1(\vct{x}_t^1;\hat{u}_t^1)\Py[\vct{x}_{t+1},\vct{z}_{t+1} \mid \vct{x}^0_t,x_t^1,\hat{u}_t^{1:2}]\right) \left(\sum_{\vct{x}^2_t}\pi_t^2(x_t^2 )\gamma_t^2( \vct{x}_t^2; \hat{u}_t^2)\right)\\
\label{probeq}&= \left(\sum_{{x}^1_t}\pi^{1|0}_t(x_t^1\mid \hat{x})\gamma_t^1(\vct{x}_t^1;\hat{u}_t^1)\Py[\vct{x}_{t+1},\vct{z}_{t+1} \mid \hat{x},x_t^1,\hat{u}_t^{1:2}]\right) \left(\sum_{\vct{x}^2_t}\pi_t^2(x_t^2 )\gamma_t^2( \vct{x}_t^2; \hat{u}_t^2)\right)\\
\label{qdef4}&=: Q_t(\pi_t,\gamma_t^{1},z_{t+1};x_{t+1})\left(\sum_{\vct{x}^2_t}\pi_t^2(x_t^2 )\gamma_t^2( \vct{x}_t^2; \hat{u}_t^2)\right).
\end{align}
Similarly, the corresponding transformation $\mathscr{P}_t^m$ (see Definition \ref{fdef}) has the following form.
\begin{align}
\mathscr{P}_t^m(\pi_t,\gamma_t^{1:2};z_{t+1}) &=  \left(\sum_{\vct{x}^2_t}\pi_t^2(x_t^2 )\gamma_t^2( \vct{x}_t^2; \hat{u}_t^2)\right)\sum_{x_{t+1}} Q_t(\pi_t,\gamma_t^{1},z_{t+1};x_{t+1})\\
\label{rdef4}&=:  \left(\sum_{\vct{x}^2_t}\pi_t^2(x_t^2 )\gamma_t^2( \vct{x}_t^2; \hat{u}_t^2)\right)R_t(\pi_t,\gamma_t^{1},z_{t+1}).
\end{align}
Using the results \eqref{qdef4} and \eqref{rdef4}, we can conclude that the transformations $Q_t$ and $R_t$ satisfy the conditions \eqref{constra} and \eqref{constrb}.

% Further, using the fact that
% \begin{align}
% \Py[\vct{x}_{t+1},\vct{z}_{t+1} \mid \vct{x}^0_t,x_t^1,\vct{u}_t^{1:2}] = \mathbbm{1}_{\hat{x}_{t+1}^0}(x_{t+1}^0)\Py[x_{t+1}^1,u_t^{1:2} \mid \vct{x}^0_{t+1},\vct{x}^0_t,x_t^1,\vct{u}_t^{1:2}]\Py[\vct{x}_{t+1}^2 \mid \vct{x}^0_t,\vct{u}_t^{1:2}]
% \end{align}
% in \eqref{probeq}, we can easily verify that $\pi_{t+1} \in \mathcal{S}_{t+1}'$.

\section{Proof of Theorem \ref{strategy}}
\label{strategyproof}
Let $\tilde{\chi}^{1*} \in \tilde{\H}^1$ be the min-max strategy for virtual player 1 in the expanded virtual game $\gm{G}_e$ as described in Section \ref{onesidedp}. Note that the strategy $\tilde{\chi}^{1*}$ uses only the common information $c_t$ and virtual player 1's past prescriptions $\gamma_{1:t-1}^1$. This is because the CIB update $F_t$ in Assumption \ref{onesidecontrolassum} does not depend on virtual player 2's prescription $\gamma_t^2$. Let us define a strategy $\chi^{1*} \in \H^1$ for virtual player 1 in the virtual game $\gm{G}_v$ defined in Appendix \ref{virtgameproof}. At each time $t$ and for each instance $c_t \in \C_t$,
\begin{align}
\chi_t^{1*}(\vct{c}_t) \doteq \Xi_t^1(\pi_t).
\end{align}
Here, $\Xi_t^1$ is the mapping obtained by solving the min-max dynamic program (see Lemma \ref{equiexistlemma}) and $\pi_t$ is computed using the following relation
\begin{align}
\pi_1(x_1,p_1^1,p_1^2) &= \Py[X_1 = x_1, P_1^1 = p_1^1,P_1^2=p_1^2 \mid C_1 = c_1] ~ \forall \; x_1,p_1^1,p_1^2\\
\pi_{\tau + 1} &= F_\tau(\pi_\tau, \Xi_\tau^1(\pi_\tau),z_{\tau+1}), ~  1 \leq \tau < t,
\end{align}
where $F_t$ is the belief update transformation in Assumption \ref{onesidecontrolassum}. Note that the prescription $\chi^{1*}_t(c_t)$ is the same as the one obtained in the ``Get prescription" step in Algorithm \ref{alg:example} for common information $c_t$ in the $t$-th iteration.

Using Definition \ref{def:rho}, we have
\begin{equation}\label{structeq}
\vct{\chi}^{1*}  = \varrho^1(\tilde{\vct{\chi}}^{1*},\tilde{\vct{\chi}}^{2})
\end{equation}
for any strategy $\tilde{\vct{\chi}}^{2} \in \tilde{H}^2$.
Based on this observation and the fact that for a given $\chi^2 \in \mathcal{H}^2$, $\varrho^2(\tilde{\chi}^1,\chi^2) = \chi^2$ for every $\tilde{\chi}^1 \in \tilde{\mathcal{H}}^1$, we have
\begin{equation}
(\vct{\chi}^{1*},\vct{\chi}^{2}) = \varrho(\tilde{\vct{\chi}}^{1*},{\vct{\chi}}^{2}).\label{anychi}
\end{equation}
Further, due to Theorem \ref{origvirt}, we have
\begin{align}
    S^u(\gm{G}_v) &\geq S^u(\gm{G}_e) \\
    &\stackrel{a}{=} \sup_{\tilde{\chi}^{2}\in \tilde{\mathcal{H}}^2}{\mathcal{J}}(\tilde{\vct{\chi}}^{1*},{\tilde{\chi}}^{2}) \\
    &\stackrel{b}{\geq} \sup_{{\chi}^{2}\in {\mathcal{H}}^2}{\mathcal{J}}(\tilde{\vct{\chi}}^{1*},{{\chi}}^{2}) \\
    &\stackrel{c}{=} \sup_{{\chi}^{2}\in {\mathcal{H}}^2}{\mathcal{J}}({\vct{\chi}}^{1*},{{\chi}}^{2}) \\
    &\geq S^u(\gm{G}_v).
\end{align}
where the equality in $(a)$ is because $\tilde{\vct{\chi}}^{1*}$ is a min-max strategy of $\gm{G}_e$. Inequality  $(b)$ holds because $\mathcal{H}^2 \subseteq \tilde{\mathcal{H}}^2$. Equality in $(c)$ is a consequence of the result in \eqref{anychi} and Lemma \ref{evolequi} in Appendix \ref{virtgameproof}. The last inequality simply follows from the definition of the upper value of the virtual game. Therefore, all the inequalities in the display above must hold with equality. Hence, ${\vct{\chi}}^{1*}$ must be a min-max strategy of game $\gm{G}_v$ and $S^u(\gm{G}_v) = S^u(\gm{G}_e)$.

From Lemma \ref{virtlemma} in Appendix \ref{virtgameproof}, we know that $S^u(\gm{G}_v) = S^u(\gm{G})$ and thus, $S^u(\gm{G}) = S^u(\gm{G}_e)$. Further, we can verify from the description of the strategy $g^{1*}$ in Algorithm \ref{alg:example} that $\chi^{1*} = \mathcal{M}^1(g^{1*})$, where $\mathcal{M}^1$ is the transformation defined in the proof of Lemma \ref{virtlemma} in Appendix \ref{virtgameproof}. Based on the argument in Remarks \ref{virtorigremark} in Appendix \ref{virtgameproof}, we can conlude that the strategy $g^{1*}$ is a min-max strategy in Game $\gm{G}$.

\section{Solving the DP: Methods and Challenges}\label{dpsolve}
In the previous section, we provided a dynamic programming characterization of the upper value $S^u(\gm{G})$ and a min-max strategy $g^{1*}$ in the original game $\gm{G}$. We now describe an approximate dynamic programming methodology \cite{bertsekas1996neuro} that can be used to compute them. The methodology we propose offers broad guidelines for a computational approach. We do not make any claim about the degree of approximation achieved by the proposed methodology. In Section \ref{dpstruct}, we discuss some structural properties of the cost-to-go and value functions in the dynamic program that may be useful for computational simplification. We illustrate our approach and the challenges involved with the help of an example in Section \ref{specialcases}.

At each time $t$, let $\hat{V}_{t+1}(\pi_{t+1},\theta_{t+1})$ be a an approximate representation of the upper value function $V_{t+1}^u(\pi_{t+1})$ in a suitable parametric form, where $\theta_{t+1}$ is a vector representing the parameters. We proceed as follows.

\paragraph{Sampling the belief space} At time $t$, we sample a set $\mathscr{S}_t$ of belief points from the set $\mathcal{S}_t$ (the set of all CIBs). A simple sampling approach would be to uniformly sample the space $\mathcal{S}_t$. Using the arguments in Appendix 5 of \cite{kartik2020upper}, we can say that the value function $V_t^u$ is uniformly continuous in the CIB $\pi_t$. Thus, if the set $\mathscr{S}_t$ is sufficiently large, computing the value function only at the belief points in $\mathscr{S}_t$ may be enough for obtaining an $\epsilon$-approximation of the value function. It is likely that the required size of $\mathscr{S}_t$ to ensure an $\epsilon$-approximation will be prohibitively large.
For solving POMDPs, more intelligent strategies for sampling the belief points, known as forward exploration heuristics, exist in literature \cite{smith2004heuristic,horak2017heuristic,xie2020optimally}. While these methods rely on certain structural properties of the value functions (such as convexity), it may be possible to adapt them for our dynamic program and make the CIB sampling process more efficient. A precise understanding of such exploration heuristics and the relationship between the approximation error $\epsilon$ and the associated computational burden is needed. This is a problem for future work.
\paragraph{Compute value function at each belief point}
Once we have a collection of points $\mathscr{S}_t$, we can then approximately compute the value $V_t^u(\pi_t)$ for each $\pi_t \in \mathscr{S}_t.$ For each belief vector $\pi_t \in \mathscr{S}_t$, we will compute $\bar{V}_t(\pi_t)$ which is given by
\begin{align} 
\hat{w}_t(\pi_t,\gamma_t^1,\gamma_t^2) &\doteq \tilde{c}_t(\pi_t,\gamma_t^1,\gamma_t^2) +\E[\hat{V}_{t+1}(F_t(\pi_{t},\gamma_t^{1},\rv{Z}_{t+1}),\theta_{t+1})\mid \pi_t,\gamma_t^{1:2}]\\
%\hat{w}_t(\pi_t,\gamma_t^1,\gamma_t^2) &\doteq \tilde{c}_t(\pi_t,\gamma_t^1,\gamma_t^2) + \sum_{\vct{z}_{t+1}}\gamma_t^2(\vct{u}_t^2)\hat{V}_{t+1}\left({\vct{Q}_{t}(\pi_t,\gamma_t^1,\vct{z}_{t+1})},\theta_{t+1}\right)\\
\bar{V}_t(\pi_t) &\doteq \min_{{\gamma}_t^1}\max_{\gamma_t^2}\hat{w}_t(\pi_t,\gamma_t^1,\gamma_t^2).\label{barv}
\end{align}
% Here, the the prescriptions $\gamma_t^1$ and $\gamma_t^2$ are parameterized using the factored form in Definition \ref{behavepres}, i.e., for each player $j$ in Team $i$ and for each instance of private information $p_t^{i,j}$, we have a set of parameters representing a distribution over the action space $\U_t^{i,j}$. This form was also used in \cite{foerster2019bayesian}.
For a given belief point $\pi_t$, one approach for solving the min-max problem in \eqref{barv} is to use the Gradient Descent Ascent (GDA) method \cite{lin2020gradient,heusel2017gans}. This can however lead to local optima because in general, the cost $\hat{w}_t$ is neither convex nor concave in the respective prescriptions. In some cases such as when Team 2 has only one player, the inner maximizing problem in \eqref{barv} can be substantially simplified and we will discuss this in Section \ref{dpstruct}.

\paragraph{Interpolation} For solving a particular instance of the min-max problem in \eqref{barv}, we need an estimate of the value function $V_{t+1}^u$. Generally, knowing just the value of the function at different points may be insufficient and we may need additional information like the first derivatives/gradients especially when we are using gradient based methods like GDA. In that case, it will be helpful to choose an appropriate parametric form (such as neural networks) for $\hat{V}_t$ that has desirable properties like continuity or differentiability.
The parameters $\theta_t$ can be obtained by solving the following regression problem.
\begin{align}\label{regression}
\min_{\theta_t}\sum_{\pi_t \in \mathscr{S}_t} (\hat{V}_t(\pi_t,\theta_t) - \bar{V}_t(\pi_t))^2.
\end{align}
Standard methods such as gradient descent can be used to solve this regression problem.

\subsection{Structural Properties of the DP}\label{dpstruct}
We will now prove that under Assumption \ref{onesidecontrolassum}, the cost-to-go function $w_t^u$ (see \eqref{uppercost}) in the min-max dynamic program defined in Section \ref{onesidedp} is linear in virtual player 2's prescription in its product form $\gamma_t^2(p_t^2;u_t^2)$ (see Definition \ref{behavepres} in Appendix \ref{infstateproof}).
\begin{lemma}\label{linearlemma}
{The cost-to-go $w_t^u(\pi_t,\gamma_t^1,\gamma_t^2)$ is linear in $\gamma_t^2(p_t^2;u_t^2)$ (see Definition \ref{behavepres} in Appendix \ref{infstateproof}), i.e., }there exists a function $a_t:\P_t^2 \times \U_t^2 \times \mathcal{S}_t \times \mathcal{B}_t^1 \to \R$ such that
\begin{align*}
    w^u_t(\pi_t,\gamma_t^1,\gamma_t^2)= \sum_{p_t^2,u_t^2}\pi_t(p_t^2)\gamma_t^2(p_t^2;u_t^2)a_t(p_t^2,u_t^2,\pi_t,\gamma_t^1).
\end{align*}
Here, $\pi_t(p_t^2)$ denotes the marginal probability of $p_t^2$ with respect to the CIB $\pi_t$.
\end{lemma}
\begin{proof}
We have
\begin{align}
w^u_t(\pi_t,\gamma_t^1,\gamma_t^2) &= \tilde{c}_t(\pi_t,\gamma_t^1,\gamma_t^2) + \E[V^u_{t+1}(F_t(\pi_{t},\gamma_t^{1:2},\rv{Z}_{t+1}))\mid \pi_t,\gamma_t^{1:2}]\\
&\label{wu}=  \tilde{c}_t(\pi_t,\gamma_t^1,\gamma_t^2) +\sum_{z_{t+1}}\mathscr{P}_t^m( \pi_t,\gamma_{t}^{1:2};z_{t+1})V_{t+1}^u(F_t(\pi_{t},\gamma_t^{1},{z}_{t+1}))\\
&\label{oneshoteam2}\doteq \sum_{p_t^2,u_t^2}\pi_t(p_t^2)\gamma_t^2(p_t^2;u_t^2)a_t(p_t^2,u_t^2,\pi_t,\gamma_t^1),
\end{align}
where the last equality uses the fact that the belief update $F_t$ does not depend on $\gamma_t^2$ and both $\tilde{c}_t$ (see \eqref{tildec}) and $\mathscr{P}_t^m$ (see \eqref{marginalprob}) are linear in $\gamma_t^2(p_t^2;u_t^2)$. We then rearrange terms in \eqref{wu} and appropriately define $a_t$ to obtain \eqref{oneshoteam2}. Here, $\pi_t(p_t^2)$ represents the marginal distribution of the second player's private information with respect to the common information based belief $\pi_t$.
\end{proof}
\begin{remark}
When there is only one player in Team 2, then the prescription $\gamma_t^2$ coincides with its product form in Definition \ref{behavepres}. Then the cost-to-go $w_t^u(\pi_t,\gamma_t^1,\gamma_t^2)$ is linear in $\gamma_t^2$ as well.
\end{remark}

At any given time $t$, a \emph{pure} prescription for virtual player 2 in Game $\gm{G}_e$ is a prescription $(\gamma^{2,1},\dots,\gamma^{2,N_2}) \in \mathcal{B}_t^2$ such that for every $j$ and every instance $p_t^{2,j} \in \P_t^{2,j}$, the distribution $\gamma^{2,j}(p_t^{2,j})$ is degenerate. Let $\mathcal{Q}_t^2$ be the set of all such {pure prescriptions} at time $t$ for virtual player 2. 
\begin{lemma}\label{finitelemma}
Under Assumption \ref{onesidecontrolassum}, for any given $\pi_t \in \mathcal{S}_t$ and $\gamma_t^1 \in \mathcal{B}_t^1$ the cost-to-go function in \eqref{uppercost} satisfies
\begin{align*}
    &\max_{\gamma_t^2 \in \mathcal{B}_t^2}w^u_t(\pi_t,\gamma_t^1,\gamma_t^2)= \max_{\gamma_t^2 \in \mathcal{Q}_t^2}w^u_t(\pi_t,\gamma_t^1,\gamma_t^2).
\end{align*}
\end{lemma}
\begin{proof}
Using Lemma \ref{linearlemma}, for a fixed value of CIB $\pi_t$ and virtual player 1's prescription $\gamma_t^1$, the inner maximization problem in \eqref{minequa} can be viewed as a single-stage team problem. In this team problem, the are $N_2$ players and the state of the system is $P_t^2$ which is distributed according to $\pi_t$. Player $j$'s information in this team is $P_t^{2,j}$. Using this information, the player selects a distribution $\delta U_t^{2,j}$ which is then used to randomly generate an action $U_t^{2,j}$. The reward for this team is given by $a_t$ defined in Lemma \ref{linearlemma}. It can be easily shown that for a single-stage team problem, we can restrict to deterministic strategies without loss of optimality. It is easy to see that the set of all deterministic strategies in the single-stage team problem described above corresponds to the collection of pure prescriptions $\mathcal{Q}_t^2$. This concludes the proof of Lemma \ref{finitelemma}.
\end{proof}

Lemma \ref{finitelemma} allows us to transform a non-concave maximization problem over the continuous space $\mathcal{B}_t^2$ in \eqref{minequa} and \eqref{barv} into a maximization problem over a finite set $\mathcal{Q}_t^2$. This is particularly useful when players in Team 2 do not have any private information because in that case, $\mathcal{Q}_t^2$ has the same size as the action space $\U_t^{2}$. Other structural properties and simplifications are discussed in Appendix \ref{linearlemma}.

\section{Proof of Lemma \ref{structlemma}}\label{structlemmaproof}

%Suppose $h^1$ is any equilibrium strategy for player 1 (we know one exists due to Proposition \ref{kuhnremark}). Since we are dealing with a zero-sum game, we know that (a)  $h^1$ achieves the infimum in $\inf_{g^1 \in \mathcal{G}^1}\sup_{g^2 \in \mathcal{G}^2}J(g^1,g^2)$, (b) any strategy achieving the infimum in the above inf-sup problem will be an equilibrium strategy for player 1 \cite{osborne1994course}. 

For proving Lemma \ref{structlemma}, it will be helpful to split Team 1's strategy $g^1$ into two parts $(g^{1,1},g^{1,2})$ where $g^{1,1} = (g^{1,1}_1,\dots,g^{1,1}_T)$ and $g^{1,2} = (g^{1,2}_1,\dots,g^{1,2}_T)$. The set of all such strategies for Player $j$ in Team 1 will be denoted by $\mathcal{G}^{1,j}$.
We will prove the lemma using the following claim.
\begin{claim}\label{claim:one}
Consider any arbitrary strategy $g^{1} = (g^{1,1},g^{1,2})$ for Team 1. Then there exists a strategy $\bar{g}^{1,1}$ for Player 1 in Team 1 such that, for each $t$, $\bar{g}^{1,1}_t$ is a function  of $X_t$ and $I^2_t$ and 
\[
J((\bar{g}^{1,1},g^{1,2}),g^2) = J((g^{1,1},g^{1,2}), g^2), ~~\forall g^2 \in \mathcal{G}^2.
\]

%\emph{for every strategy $g^2 \in \mathcal{G}^2$.}
\end{claim}

Suppose that the above claim is true. Let $(h^{1,1},h^{1,2})$ be a min-max strategy for Team 1.
Due to Claim \ref{claim:one}, there exists a strategy $\bar{h}^{1,1}$ for Player 1 in Team 1 such that, for each $t$, $\bar{h}^{1,1}_t$ is a function only of $X_t$ and $I^2_t$ and 
\[
J((\bar{h}^{1,1},h^{1,2}),g^2) = J((h^{1,1},h^{1,2}), g^2),
\]
\emph{for every strategy $g^2 \in \mathcal{G}^2$.} Therefore, we have that 
\begin{align*}
\sup_{g^2 \in \mathcal{G}^2}J((\bar{h}^{1,1},h^{1,2}),g^2) = \sup_{g^2 \in \mathcal{G}^2}J((h^{1,1},h^{1,2}),g^2) = \inf_{g^1 \in \mathcal{G}^1}\sup_{g^2 \in \mathcal{G}^2}J(g^1,g^2).
\end{align*} 
Thus, $(\bar{h}^{1,1},h^{1,2})$ is a min-max strategy for Team 1 wherein Player 1 uses only the current state and Player 2's information.

{\emph{Proof of Claim \ref{claim:one}:}} We now proceed to prove Claim \ref{claim:one}. 
Consider any arbitrary strategy $g^1 = (g^{1,1},g^{1,2})$ for Team 1.  Let $\iota_t^2 = \{u_{1:t-1}^{1:2},y_{1:t}^2\}$ be a realization of Team 2's information $I_t^2$ (which is the same as Player 2's information in Team 1). Define the distribution $\Psi_t(\iota_t^2)$ over the space $(\prod_{\tau = 1}^t\X_{\tau} )\times\U_{t}^{1,1}$ as follows:
\begin{align*}
\Psi_t(\iota_t^2; x_{1:t},u_{t}^{1,1}) \doteq \Py^{g^1,h^2}[X_{1:t},U^{1,1}_{t} = (x_{1:t},u_{t}^{1,1}) \mid I_t^2 = \iota_t^2],
\end{align*}
%\begin{align*}
%&\Psi_t(\iota_t^2; x_{1:t},u_{1:t}^1) \doteq \notag \\
%~~&\Py^{(g^1,openloop)}[X_{1:t},U^1_{1:t} = (x_{1:t},u_{1:t}^1) \mid I_t^2 = \iota_t^2],
%\end{align*}
if $\iota_t^2$ is \emph{feasible}, that is $\Py^{g^1,h^2}[I_t^2 = \iota_t^2] > 0$, under the \emph{open-loop} strategy $h^2 \doteq (u^2_{1:t-1})$ for Team 2. Otherwise, define $\Psi_t(\iota_t^2; x_{1:t},u_{t}^{1,1})$ to be the uniform
%\begin{align}
%&\Psi_t(\iota_t^2; x_{1:t},u_{1:t}^1) \\
%&= \begin{cases}
%\Py^{g^1,h^2}[X_{1:t},U_{1:t} = x_{1:t},u_{1:t}^1 \mid I_t^2 = \iota_t^2] & \text{if $\iota_t^2$ is feasible under}\\
%&\text{some open-loop $h^2$}\\
%\mathscr{U}(\cdot) & \text{otherwise},
%\end{cases}
%\end{align}
 distribution over the space $(\prod_{\tau = 1}^t\X_{\tau} )\times\U_{t}^{1,1}$.

\begin{lemma}\label{claimJ1}
Let $g^1$ be Team 1's strategy and let $g^2$ be an arbitrary strategy for Team 2. Then for any realization $x_{1:t},u_{t}^{1,1}$ of the variables $X_{1:t},U^{1,1}_{t}$, we have
\begin{align*}
\Py^{g^1,g^2}[X_{1:t},U^{1,1}_{t} = (x_{1:t},u_{t}^{1,1}) \mid I_t^2] = \Psi_t(I_t^2; x_{1:t},u_{t}^{1,1}),
\end{align*}
almost surely.
\end{lemma}
\begin{proof}
{From Team 2's perspective, the system evolution can be seen in the following manner. The system state at time $t$ is $S_t = (X_{1:t},U_{t}^{1,1},I^2_t)$. Team 2 obtains a partial observation $Y_{t}^2$ of the state at time $t$. Using information $\{Y_{1:t}^2,U_{1:t-1}^{1:2}\}$, Team 2 then selects an action $U_t^2$. The state then evolves in a controlled Markovian manner (with dynamics that depend on $g^1$). Thus, from Team 2's perspective,  this system is a partially observable Markov decision process (POMDP). The claim in the Lemma then follows from the standard result in POMDPs that the belief on the state given the player's information does not depend on the player's strategy \cite{kumar2015stochastic}.}
\end{proof}

For any instance $\iota_t^2$ of Team 2's information $I_t^2$, define the distribution $\Phi_t(\iota_t^2)$ over the space $\X_t \times \U_t^{1,1}$ as follows
\begin{align}
\Phi_t(\iota_t^2; x_t, u_t^{1,1}) = {\sum_{x_{1:t-1}}\Psi_t(\iota_t^2;x_{1:t},u_{t}^{1,1})}.\label{eq:phi}
\end{align}
Define strategy $\bar{g}^{1,1}$ for Player 1 in Team 1 such that for any realization $x_t, \iota_t^2$ of state $X_t$ and Player 2's information $I_t^2$ at time $t$, the probability of selecting an action $u_t^{1,1}$ at time $t$ is
%\begin{align}
%\bar{h}_t^1(x_t,\iota_{t}^2;u_t^1) &= \sum_{\bar{p}_t^1}\Py^{h^1}[\bar{p}_t^1 \mid x_t,\iota_t^2]h^1_t(\iota_t^1; u_t^1)\\
%&=\sum_{\bar{p}_t^1}\Py^{h^1}[\bar{p}_t^1 \mid x_t,\iota_t^2]\Py^{h^1}[u_t^1 \mid \iota_t^1],
%\end{align}
\begin{align}\label{eq:defgbar}
\bar{g}_t^{1,1}(x_t,\iota_{t}^2;u_t^{1,1}) \doteq 
\begin{cases}
\frac{\Phi_t(\iota_t^2; x_t,u_t^{1,1})}{\sum_{u_t^{1,1'}}\Phi_t(\iota_t^2; x_t,u_t^{1,1'})} & \text{if } \sum_{u_t^{1,1'}}\Phi_t(\iota_t^2; x_t,u_t^{1,1'}) > 0\\
\mathscr{U}(\cdot) & \text{otherwise},
\end{cases}
\end{align}
where $\mathscr{U}(\cdot)$ denotes the uniform distribution over the action space $\U_t^{1,1}$. Notice that the construction of the strategy $\bar{g}^{1,1}$ does not involve Team 2's strategy $g^2$.

\begin{lemma}\label{aseqstrat}
For any strategy $g^2$ for Team 2, we have
\begin{align*}
\Py^{(g^1,g^2)}[U_t^{1,1} = u_t^{1,1} \mid X_t, I_t^2] = \bar{g}_t^{1,1}(X_t,I_{t}^2;u_t^{1,1})
\end{align*}
almost surely for every $u_t^{1,1} \in \U_t^{1,1}$.
\end{lemma}
\begin{proof}
Let $x_t, \iota_t^2$ be a realization that has a non-zero probability of occurrence under the strategy profile $(g^1,g^2)$. Then using Lemma \ref{claimJ1}, we have
\begin{align}
\Py^{(g^1,g^2)}[X_{1:t},U^{1,1}_{t} = (x_{1:t},u_{t}^{1,1}) \mid \iota_t^2] = \Psi_t(\iota_t^2; x_{1:t},u_{t}^{1,1}), \label{eq:lemma14}
\end{align}
for every realization $x_{1:t-1}$ of states $X_{1:t-1}$ and $u_{t}^{1,1}$ of action $U_{t}^{1,1}$. Summing over all $x_{1:t-1},u^{1,1}_{t}$ and using \eqref{eq:phi} and \eqref{eq:lemma14},  we have
\begin{align}
 \Py^{(g^1,g^2)}[X_t = x_t \mid I_t^2 = \iota_t^2] =\sum_{u_t^{1,1}}\Phi_t(\iota_t^2; x_t,u_t^{1,1}).
\end{align}
The left hand side of the above equation is positive since $x_t,i^2_t$ is a realization of positive probability under the strategy profile $(g^1,g^2)$.

Using Bayes' rule,  \eqref{eq:phi}, \eqref{eq:defgbar} and \eqref{eq:lemma14}, we obtain
\begin{align}
\nonumber\Py^{g^1,g^2}[U_t^{1,1} = u_t^{1,1} \mid X_t = x_t, I_t^2 = \iota_t^2]
&= \frac{\Phi_t(\iota_t^2; x_t,u_t^{1,1})}{\sum_{u_t^{1,1'}}\Phi_t(\iota_t^2; x_t,u_t^{1,1'})}= \bar{g}_t^{1,1}(x_t,\iota_{t}^2;u_t^{1,1}) .
\end{align}
This concludes the proof of the lemma.
\end{proof}

Let us define $\bar{g}^1 = (\bar{g}^{1,1},g^{1,2})$, where $\bar{g}^{1,1}$ is as defined in \eqref{eq:defgbar}. We can now show that the strategy $\bar{g}^1$ satisfies
\[
J(\bar{g}^1,g^2) = J(g^1, g^2),
\]
for every strategy $g^2 \in \mathcal{G}^2$. 
Because of the structure of the cost function in \eqref{eq:totalcost}, it is sufficient to show that for each time $t$, the random variables   $(X_t,U_t^1,U_t^2,I^2_t)$ have the same joint distribution under strategy profiles $(g^1,g^2)$ and $(\bar{g}^1,g^2)$. We prove this by induction. It is easy to verify that at time $t=1$, $(X_1,U_1^1,U_1^2,I^2_1)$ have the same joint distribution under strategy profiles $(g^1,g^2)$ and $(\bar{g}^1,g^2)$.

Now assume that at time $t$, 
\begin{align}
\label{jointeq}\Py^{g^1,g^2}[x_t,u_t^1,u_t^2,\iota_t^2] = \Py^{\bar{g}^1,g^2}[x_t,u_t^1,u_t^2,\iota_t^2],
\end{align}
 for any realization of state, actions and Team 2's information $x_t,u_t^1,u_t^2, \iota_t^2$. Let $\iota_{t+1}^2 = (\iota_t^2,u_t^{1:2},y_{t+1}^2)$. Then we have
\begin{align}
\Py^{g^1,g^2}[x_{t+1},\iota_{t+1}^2] &= \nonumber\sum_{\bar{x}_t}\Py[x_{t+1},y_{t+1}^2 \mid \bar{x}_t,u_t^{1:2},\iota_t^2]\Py^{g^1,g^2}[\bar{x}_t,u_t^{1:2},\iota_t^2]\\
&= \label{indhyp}\sum_{\bar{x}_t}\Py[x_{t+1},y_{t+1}^2 \mid \bar{x}_t,u_t^{1:2},\iota_t^2]\Py^{\bar{g}^1,g^2}[\bar{x}_t,u_t^{1:2},\iota_t^2]\\
&=\Py^{\bar{g}^1,g^2}[x_{t+1},\iota_{t+1}^2].\label{indhyp2}
\end{align}
The equality in (\ref{indhyp}) is due to the induction hypothesis. Note that the conditional distribution $\Py[x_{t+1},\iota_{t+1}^2 \mid {x}_t,{u}_t^1,u_t^2,\iota_t^2]$ does not depend on players' strategies (see equations (\ref{statevol}) and (\ref{obseq})).
 
At $t+1$, for any realization $x_{t+1},u_{t+1}^1,u_{t+1}^2,\iota_{t+1}^2$ that has non-zero probability of occurrence under the strategy profile $(g^1,g^2)$, we have
\begin{align}
\label{constarg}&\Py^{g^1,g^2}[x_{t+1},u_{t+1}^1,u_{t+1}^2,\iota_{t+1}^2] \\
&= \Py^{g^1,g^2}[x_{t+1},\iota_{t+1}^2]g_t^2(\iota_{t+1}^2;u_{t+1}^2)g_t^{1,2}(\iota_{t+1}^2;u_{t+1}^{1,2})\Py^{g^1,g^2}[u_{t+1}^{1,1}\mid x_{t+1},\iota_{t+1}^2]\\
&= \Py^{g^1,g^2}[x_{t+1},\iota_{t+1}^2]g_t^2(\iota_{t+1}^2;u_{t+1}^2)g_t^{1,2}(\iota_{t+1}^2;u_{t+1}^{1,2})\bar{g}_t^{1,1}(x_{t+1},\iota_{t+1}^2;u_{t+1}^{1,1})\label{constarg1}\\
&= \Py^{\bar{g}^1,g^2}[x_{t+1},\iota_{t+1}^2]g_t^2(\iota_{t+1}^2;u_{t+1}^2)g_t^{1,2}(\iota_{t+1}^2;u_{t+1}^{1,2})\bar{g}_t^{1,1}(x_{t+1},\iota_{t+1}^2;u_{t+1}^{1,1})\label{constarg3}\\
&= \Py^{\bar{g}^1,g^2}[x_{t+1},\iota_{t+1}^2]g_t^2(\iota_{t+1}^2;u_{t+1}^2)g_t^{1,2}(\iota_{t+1}^2;u_{t+1}^{1,2})\Py^{\bar{g}^1,g^2}[u_{t+1}^{1,1}\mid x_{t+1},\iota_{t+1}^2]\label{constarg2}\\
&= \Py^{\bar{g}^1,g^2}[x_{t+1},u_{t+1}^1,u_{t+1}^2,\iota_{t+1}^2]\label{constarg4},
\end{align}
where the equality in \eqref{constarg} is a consequence of the chain rule and the manner in which players randomize their actions. Equality in \eqref{constarg1} follows from Lemma \ref{aseqstrat} and the equality in \eqref{constarg3} follows from the result in \eqref{indhyp2}.
%\red{you have used lemma 14, please mention it....
%The last inequality follows from the construction of strategy $\bar{g}^1$.} 
%Consider the first term in \eqref{constarg2}, andCombining equations (\ref{constarg2}) and (\ref{indhyp}), we have
%\begin{align*}
%\Py^{g^1,g^2}[x_{t+1},u_{t+1}^1,u_{t+1}^2,\iota_{t+1}^2] = \Py^{\bar{g}^1,g^2}[x_{t+1},u_{t+1}^1,u_{t+1}^2,\iota_{t+1}^2].
%\end{align*}
Therefore, by induction, the equality in \eqref{jointeq} holds for all $t$. This concludes the proof of Claim \ref{claim:one}. \qed

\section{Numerical Experiments: Details}\label{numappend}

\subsection{Game Model}

\paragraph{States} There are two players in Team 1 and one player in Team 2. We will refer to Player 2 in Team 1 as the defender and the Player in Team 2 as the attacker. Player 1 in Team 1 will be referred to as the signaling player. The state space of the system is $\X_t = \{(l,a),(r,a),(l,p),(r,p)\}$. For convenience, we will denote these four states with $\{0,1,2,3\}$ in the same order. Recall that $l$ and $r$ represent the two entities and $a$ and $p$ denote the active and passive states of the attacker. Therefore, if the state $X_t = (l,a)$, this means that the vulnerable entity is $l$ and the attacker is active. Similar interpretation applies for other states.

\paragraph{Actions} Each player in Team 1 has two actions and let $\U_t^{1,1} = \U_t^{1,2} = \{\alpha,\beta\}$. Player 2's action $\alpha$ corresponds to defending entity $l$ and $\beta$ corresponds to defending entity $r$. Player 1's actions are meant only for signaling. The player (attacker) in Team 2 has three actions, $\U_t^{2,1} = \{\alpha,\beta,\mu\}$. For the attacker, $\alpha$ represents the targeted attack on entity $l$, $\beta$ represents the targeted attack on entity $r$ and $\mu$ represents the blanket attack on both entities.

\paragraph{Dynamics} In this particular example, the actions of the players in Team 1 do not affect the state evolution. Only the attacker's actions affect the state evolution. We consider this type of setup for simplicity and one can let the players in Team 1 control the state evolution as well. On the other hand, note that the evolution of the CIB is controlled only by Team 1 (the signaling player) and not by Team 2 (attacker). The transition probabilities are provided in Table \ref{trans}. We interpret these transitions below.

Whenever the attacker is passive, i.e. $X_t = 2$ or 3, then the attacker will remain passive with probability 0.7. In this case, the system state does not change. With probability 0.3, the attacker will become active. Given that the attacker becomes active, the next state $X_{t+1}$ may either be 0 or 1 with probability 0.5 each. In this state, even the attacker's actions do not affect the transitions.

When the current state $X_t = 0$, the following cases are possible:
\begin{enumerate}
    \item Attacker plays $\alpha$: In this case, the attacker launches a successful targeted attack. The next state $X_{t+1}$ is either 0 or 1 with probability 0.5.
    \item Attacker plays $\beta$: The next state $X_{t+1}$ is 2 with probability 1. Thus, the attacker becomes passive if it targets the invulnerable entity.
    \item Attacker plays $\mu$: The state does not change with probability $0.7$. With probability 0.3, $X_{t+1} = 2$.
\end{enumerate}
The attacker's actions can be interpreted similarly when the state $X_t = 1$.

\begin{table}[t]
\caption{The transition probabilities $\Py[X_{t+1} \mid X_t, U_t^{2}]$. Note that the dynamics do not depend on time $t$ and Team 1's actions.}
\label{trans}
%\vskip 0.15in
\begin{center}
\begin{small}
\begin{sc}
\begin{tabular}{lccccr}
\toprule
 & $X_t = 0$ & $X_t = 1$ & $X_t=2$ & $X_t=3$ \\
\midrule
$\alpha$   & (0.5,0.5,0.0,0.0)& (0.0,0.0,0.0,1.0)& (0.15,0.15,0.7,0.0)& (0.15,0.15,0.0,0.7)\\
$\beta$   & (0.0,0.0,1.0,0.0)& (0.5,0.5,0.0,0.0)& (0.15,0.15,0.7,0.0) & (0.15,0.15,0.0,0.7)\\
$\mu$   & (0.7,0.0,0.3,0.0)& (0.0,0.7,0.0,0.3)& (0.15,0.15,0.7,0.0) & (0.15,0.15,0.0,0.7)\\
\bottomrule
\end{tabular}
\end{sc}
\end{small}
\end{center}
\vskip -0.1in
\end{table}

\paragraph{Cost} Player 1's actions in Team 1 do not affect the cost in this example (again for simplicity). The cost as a function of the state and the other players' actions is provided in Table \ref{costable}.

\begin{table}[t]
\caption{The cost $c(X_t,U_t^{1,2},U_t^{2,1})$. Note that the actions of Player 1 in Team 1 do not affect the cost. Each row corresponds to a pair of actions $(U_t^{1,2},U_t^{2,1})$ and each column corresponds to a state $X_t$.}
\label{costable}
%\vskip 0.15in
\begin{center}
\begin{small}
\begin{sc}
\begin{tabular}{lccccr}
\toprule
 & $X_t = 0$ & $X_t = 1$ & $X_t=2$ & $X_t=3$ \\
\midrule
$(\alpha,\alpha)$   & 15 & 0& 0& 0\\
$(\alpha,\beta)$   & 0& 15& 0 & 0\\
$(\alpha,\mu)$   & 10& 20& 0 & 0\\
$(\beta,\alpha)$   & 15 & 0& 0& 0\\
$(\beta,\beta)$   & 0& 15& 0 & 0\\
$(\beta,\mu)$   & 20& 10& 0 & 0\\
\bottomrule
\end{tabular}
\end{sc}
\end{small}
\end{center}
\vskip -0.1in
\end{table}

Note that when the attacker is passive, the system does not incur any cost. If the attacker launches a successful targeted attack, then the cost is 15. If the attacker launches a failed targeted attack, then the cost is 0. When the attacker launches a blanket attack, if the defender happens to be defending the vulnerable entity, the cost incurred is 10. Otherwise, the cost incurred is 20.

We consider discounted cost with a horizon of $T=15$. Therefore, the cost function $c_t(\cdot) = \delta^{t-1}c(\cdot)$ at time $t$, where the function $c(\cdot)$ is defined in Table \ref{costable}. The discount factor $\delta$ in our problem is $0.9$.

\subsection{Architecture and Implementation}
At any time $t$, we represent the value functions $V_t^u$ using a deep fully-connected ReLU network. The network (denoted by $\hat{V}_t$ with parameters $\theta_t$) takes the CIB $\pi_t$ as input and returns the value $\hat{V}_t(\pi_t,\theta_t)$. The state $X_t$ in our example takes only four values and thus, we sample the CIB belief space uniformly. For each CIB point in the sampled set, we compute an estimate of the value function $\bar{V}_t$ as in \eqref{barv}. Since Team 2 has only one player, we can use the structural result in Section \ref{dpstruct} to simplify the inner maximization problem in \eqref{barv} and solve the outer minimization problem in \eqref{barv} using gradient descent. The neural network representation of the value function is helpful for this as we can compute gradients using backpropagation. Since this minimization is not convex, we may not converge to the global optimum. Therefore, we use multiple initializations and pick the best solution. Finally, the interpolation step is performed by training the neural network $\hat{V}_t$ with $\ell_2$ loss.

We illustrate this iterative procedure in Section \ref{valueapproxfigs}. For each time $t$, we compute the upper value functions $\bar{V}_t$ at the sampled belief points using the estimated value function $\hat{V}_{t+1}$. In the figures in Section \ref{valueapproxfigs}, these points are labeled ``Value function" and plotted in blue. The weights of the neural network $\hat{V}_t$ are then adjusted by regression to obtain an approximation. This approximated value function is plotted in orange. For the particular example discussed above, we can obtain very close approximations of the value functions. We can also observe that the value function converges due to discounting as the horizon of the problem increases.

\subsection{The Value Function and Team 1's Strategy}
The value function computed using the methodology described above is depicted in Figure \ref{valuefunc}. Since the the initial belief is given by $\pi_1(0) = \pi_1(1) = 0.5$, we can conclude that the min-max value of the game is 65.2.
We note that the value function in Figure \ref{valuefunc} is neither convex nor concave. As discussed earlier in Section \ref{specialcases}, the trade-off between signaling and secrecy can be understood from the value function. Clearly, revealing too much (or too little) information about the hidden state is unfavorable for Team 1. The most favorable belief according to the value function seems to be $\pi_1(0) = 0.28$ (or 0.72). 

We also compute a min-max strategy for Team 1 and the corresponding best-response\footnote{This is not to be confused with the max-min strategy for the attacker.} for the attacker in Team 2. The min-max strategy is computed using Algorithm \ref{alg:example}. 
\begin{enumerate}
    \item In the best-response, the attacker decides to launch a targeted attack (while it is active) on entity $l$ if $0\leq\pi_1(0)<0.28$ and on entity $r$ if $0.72<\pi_1(0)\leq 1$. For all other belief-states it launches a blanket attack. Clearly, the attacker launches a targeted attack only if it is confident enough.
    \item According to our computed min-max strategy, Player 2 in Team 1 simply computes the CIB and chooses to defend the entity that is more likely to be vulnerable.
    \item The signaling strategy for Player 1 in Team 1 is more complicated. This can be seen from the strategy depicted in Figure \ref{strategy_fig}. The signaling strategy can be divided into three broad regions as shown in Figure \ref{strategy_fig}. 
    \begin{enumerate}
        \item The grey region is where the signaling can be arbitrary. This is because the attacker in these belief states is confident enough about which entity is vulnerable and will decide to launch a targeted attack immediately.
        \item In the brown region, the signaling is negligible. This is because even slight signaling might make the attacker confident enough to launch a targeted attack.
        \item In the yellow region, there is substantial signaling. In this region, the attacker uses the blanket attack and signaling can help Player 2 in Team 1 defend better.
    \end{enumerate}
      
\end{enumerate}

We also observe an interesting pattern in the signaling strategy. Consider the CIB $\pi_1(0) = 0.6$ in Figure \ref{strategy_fig}. Let the true state be 0. In this scenario, Player 1 in Team 1 knows that the state is 0 while Player 2 believes that state 1 is more likely. Therefore, Player 1's signaling strategy must be such that it quickly resolves this \emph{mismatch} between the state and Player 2's belief. However, when the true state is 1, Player 2's belief is consistent with the state. Too much signaling in this case can make the attacker more confident about the system state. We note that there is an asymmetry w.r.t. the state in the signaling requirement, i.e., signaling is favorable when the state is 0 but unfavorable when the state is 1. The asymmetric signaling strategy for states 0 and 1 depicted in Figure \ref{strategy_fig} achieves this goal of asymmetric signaling efficiently.

% Additional details such as the hyperparameters used can be found in the code attached. This code should run on Google Colab without any issues.

\twocolumn
\subsection{Value Function Approximations}\label{valueapproxfigs}
In this Section, we present successive approximations of the value functions obtained using our methodology in Section \ref{dpsolve}.
\begin{figure}[ht]
%\vskip 0.2in
\begin{center}
\centerline{\includegraphics[width=0.7\columnwidth]{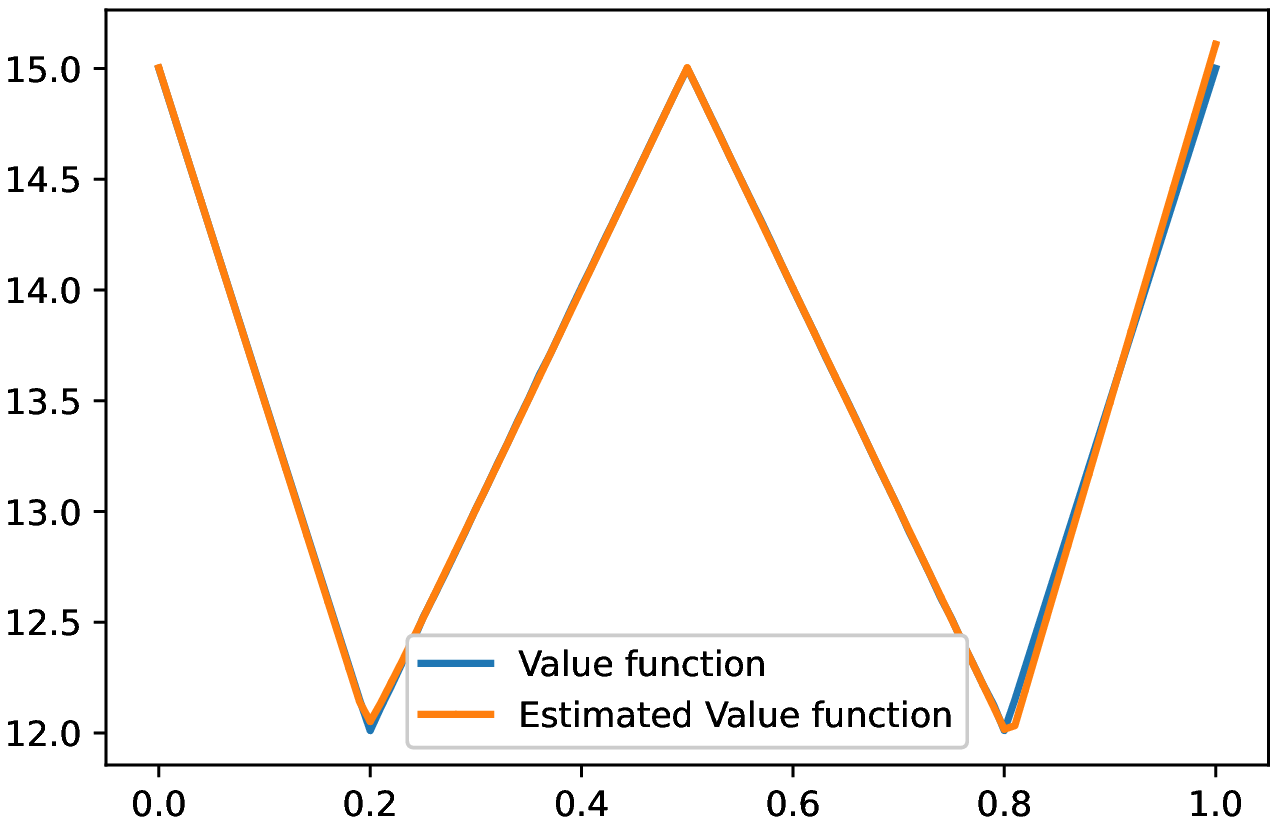}}
\caption{Estimated value function $\hat{V}_{15}$}
\label{icml-historical}
\end{center}
\vskip -0.2in
\end{figure}

\begin{figure}[ht]
%\vskip 0.2in
\begin{center}
\centerline{\includegraphics[width=0.7\columnwidth]{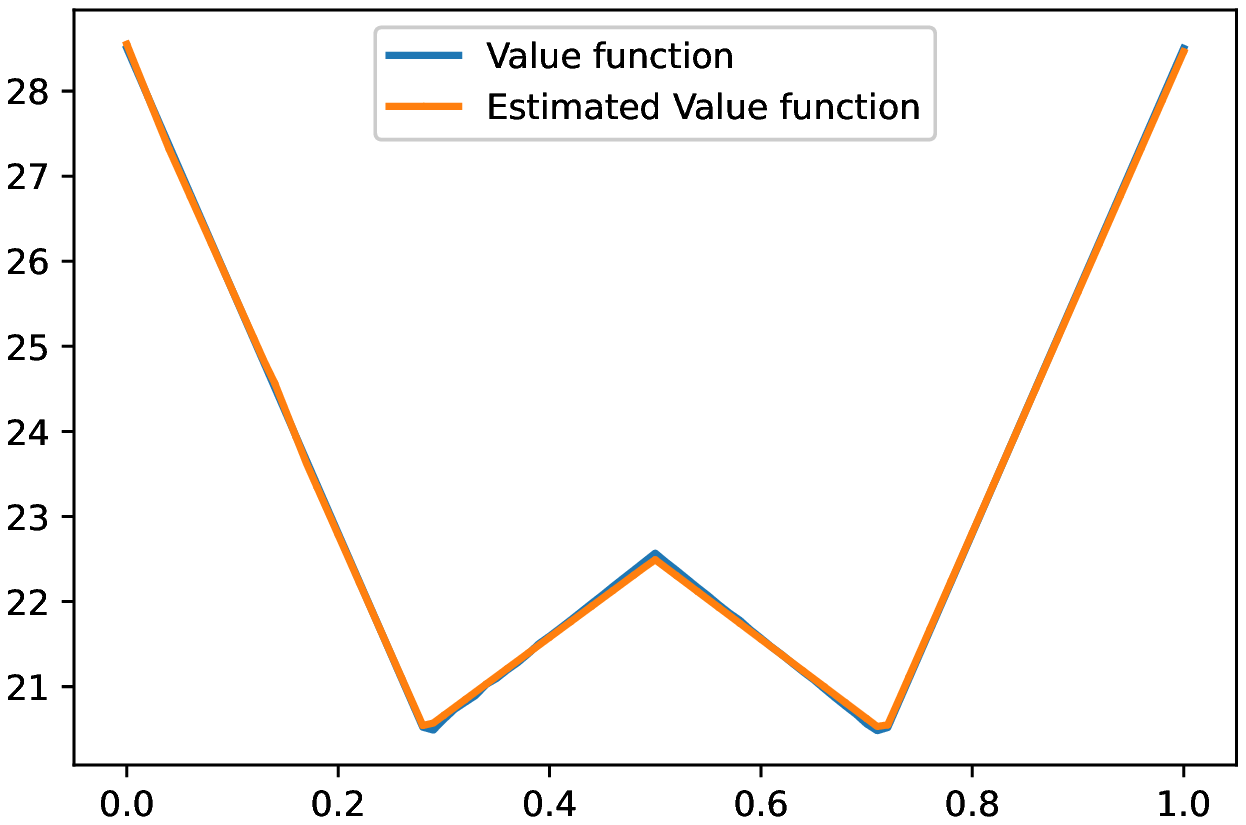}}
\caption{Estimated value function $\hat{V}_{14}$}
\label{icml-historical}
\end{center}
\vskip -0.2in
\end{figure}

\begin{figure}[ht]
%\vskip 0.2in
\begin{center}
\centerline{\includegraphics[width=0.7\columnwidth]{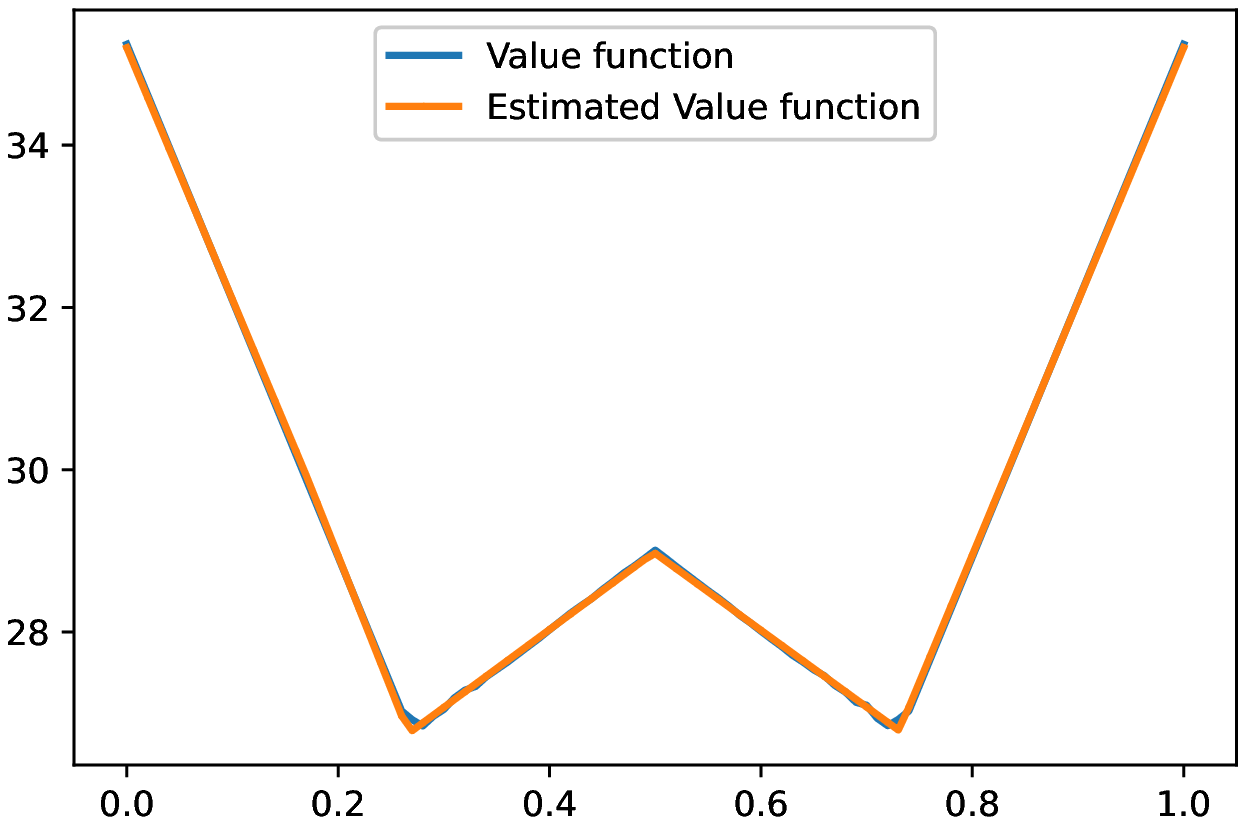}}
\caption{Estimated value function $\hat{V}_{13}$}
\label{icml-historical}
\end{center}
\vskip -0.2in
\end{figure}

\begin{figure}[ht]
%\vskip 0.2in
\begin{center}
\centerline{\includegraphics[width=0.7\columnwidth]{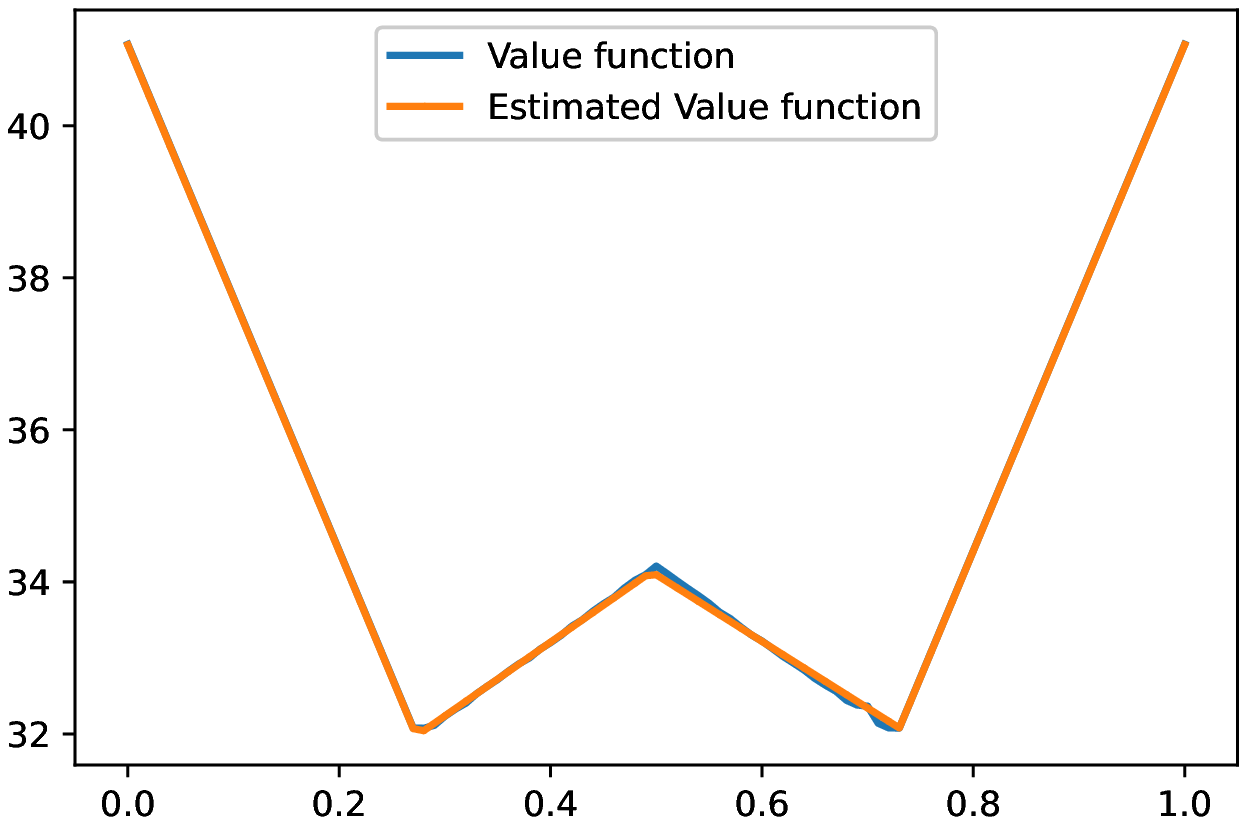}}
\caption{Estimated value function $\hat{V}_{12}$}
\label{icml-historical}
\end{center}
\vskip -0.2in
\end{figure}

\begin{figure}[ht]
%\vskip 0.2in
\begin{center}
\centerline{\includegraphics[width=0.7\columnwidth]{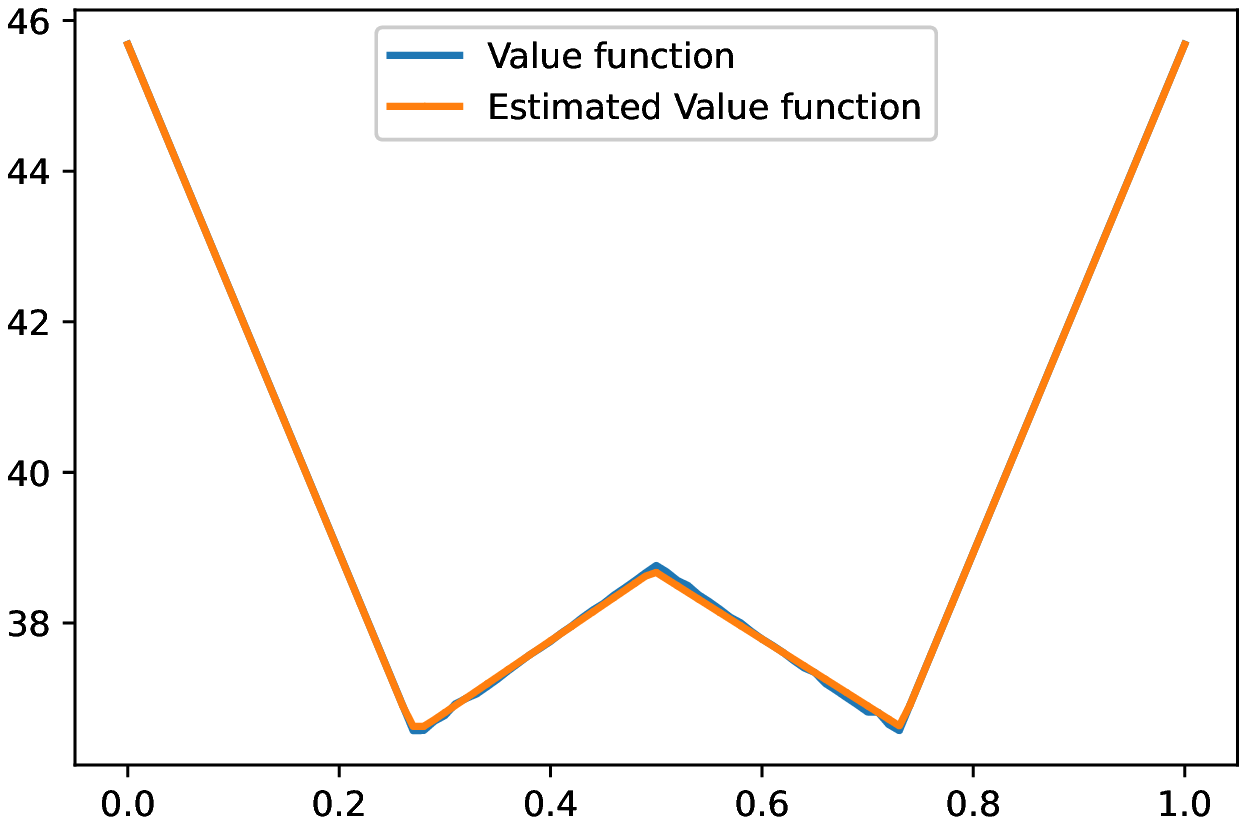}}
\caption{Estimated value function $\hat{V}_{11}$}
\label{icml-historical}
\end{center}
\vskip -0.2in
\end{figure}

\begin{figure}[ht]
%\vskip 0.2in
\begin{center}
\centerline{\includegraphics[width=0.7\columnwidth]{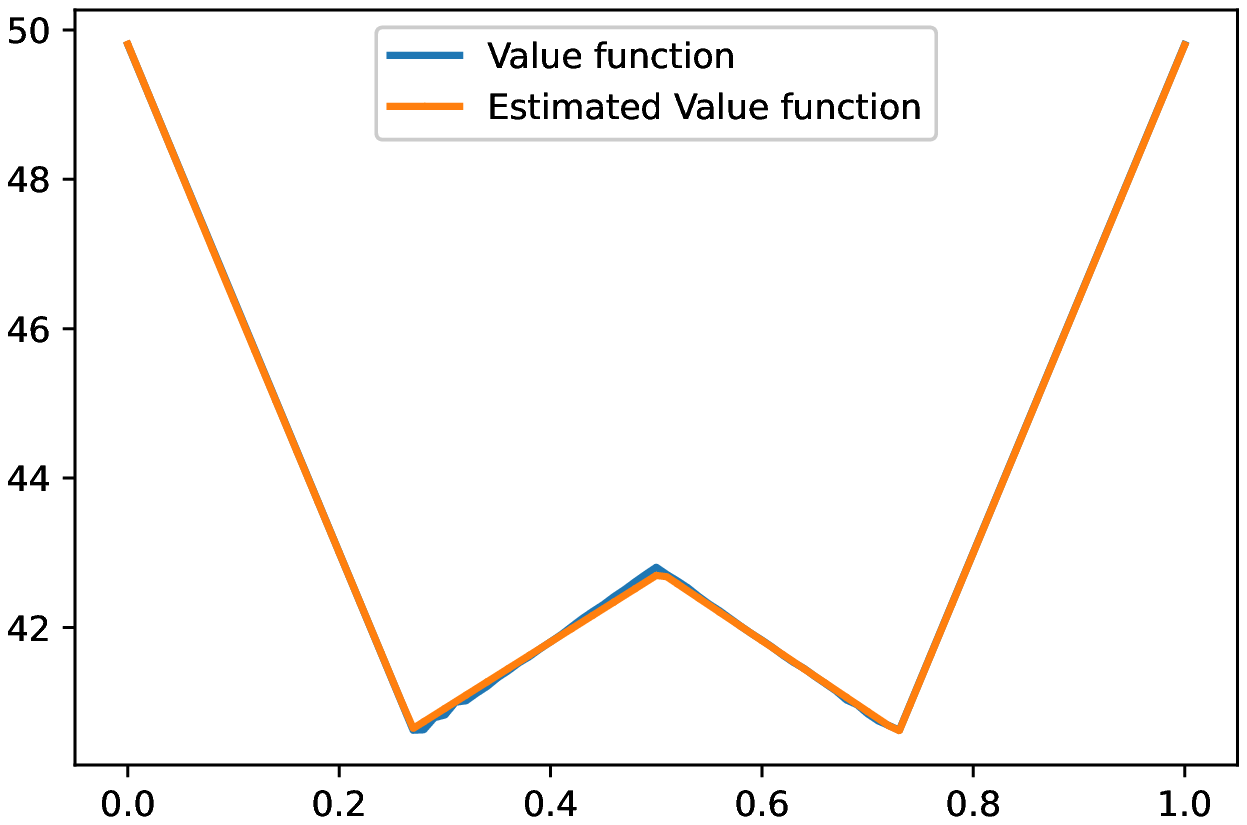}}
\caption{Estimated value function $\hat{V}_{10}$}
\label{icml-historical}
\end{center}
\vskip -0.2in
\end{figure}

\begin{figure}[ht]
%\vskip 0.2in
\begin{center}
\centerline{\includegraphics[width=0.7\columnwidth]{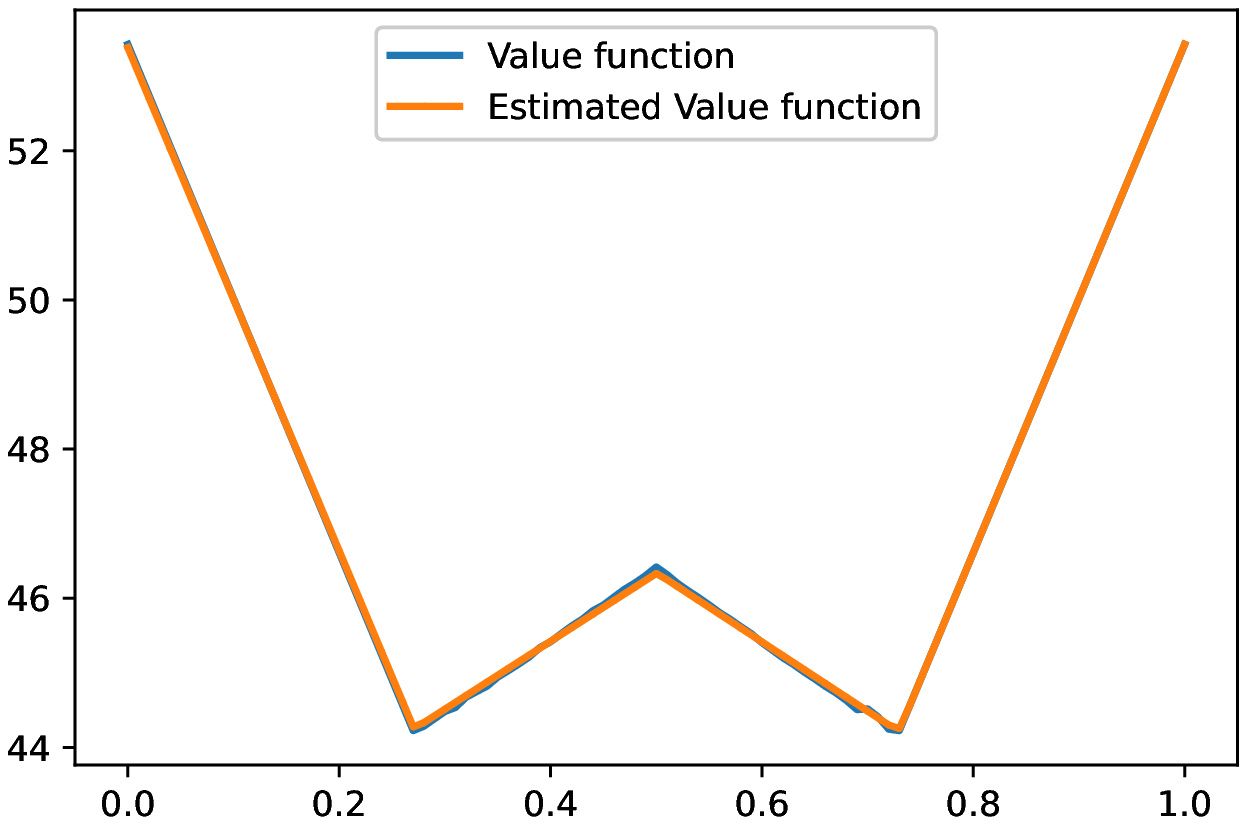}}
\caption{Estimated value function $\hat{V}_{9}$}
\label{icml-historical}
\end{center}
\vskip -0.2in
\end{figure}

\begin{figure}[ht]
%\vskip 0.2in
\begin{center}
\centerline{\includegraphics[width=0.7\columnwidth]{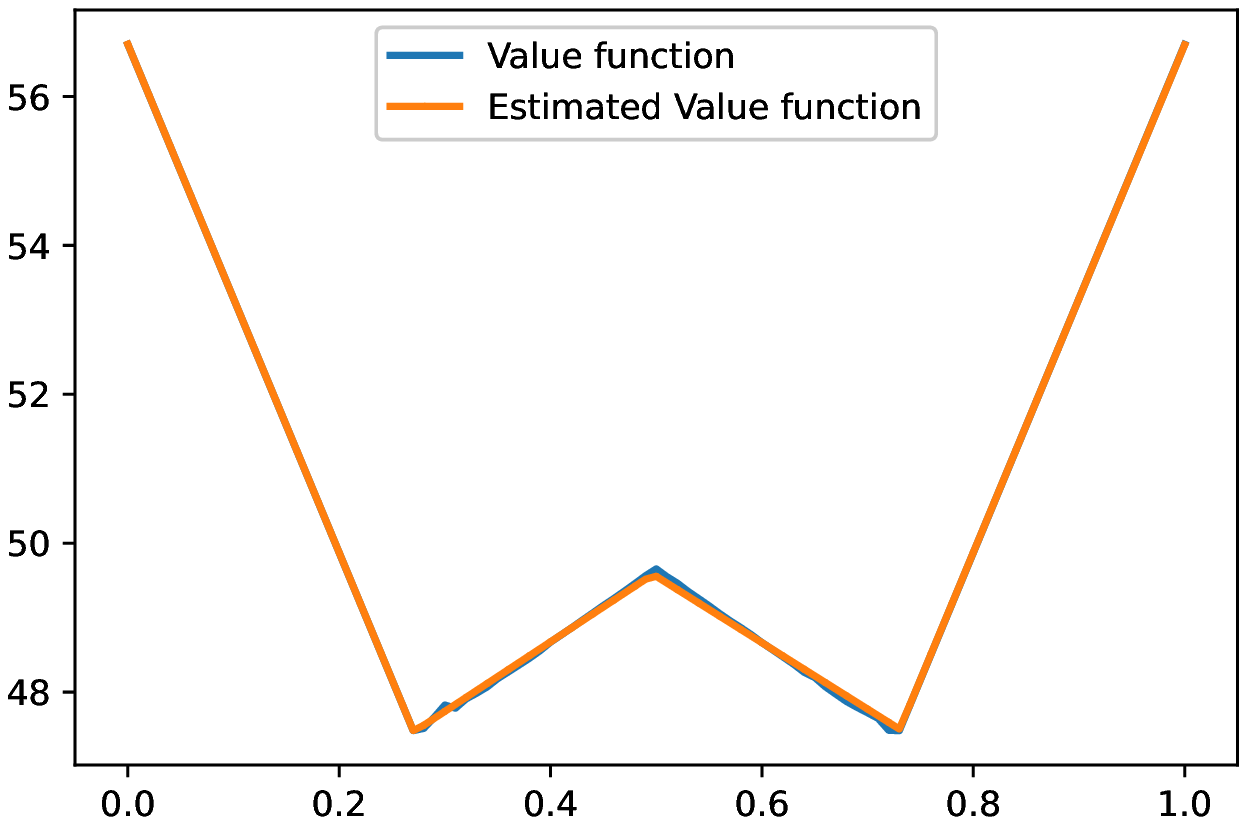}}
\caption{Estimated value function $\hat{V}_{8}$}
\label{icml-historical}
\end{center}
\vskip -0.2in
\end{figure}

\begin{figure}[ht]
%\vskip 0.2in
\begin{center}
\centerline{\includegraphics[width=0.7\columnwidth]{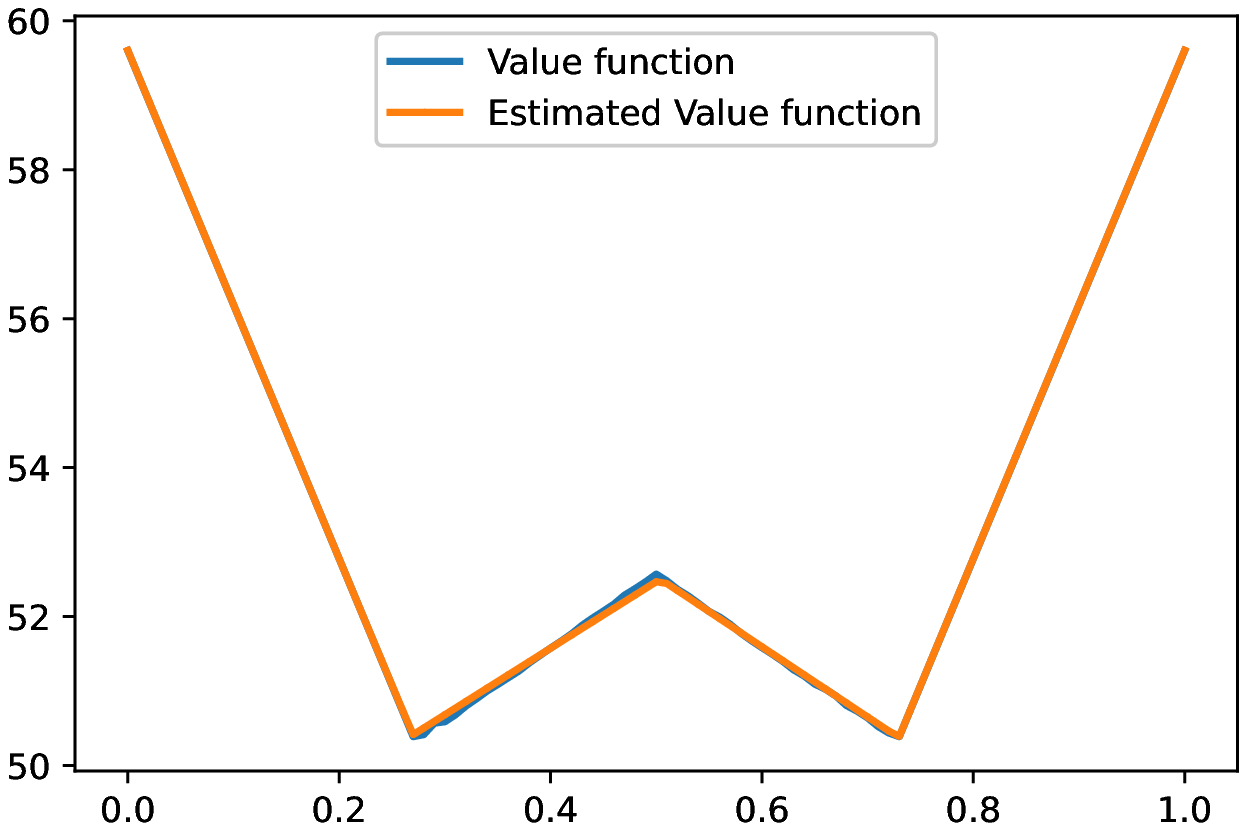}}
\caption{Estimated value function $\hat{V}_{7}$}
\label{icml-historical}
\end{center}
\vskip -0.2in
\end{figure}

\begin{figure}[ht]
%\vskip 0.2in
\begin{center}
\centerline{\includegraphics[width=0.7\columnwidth]{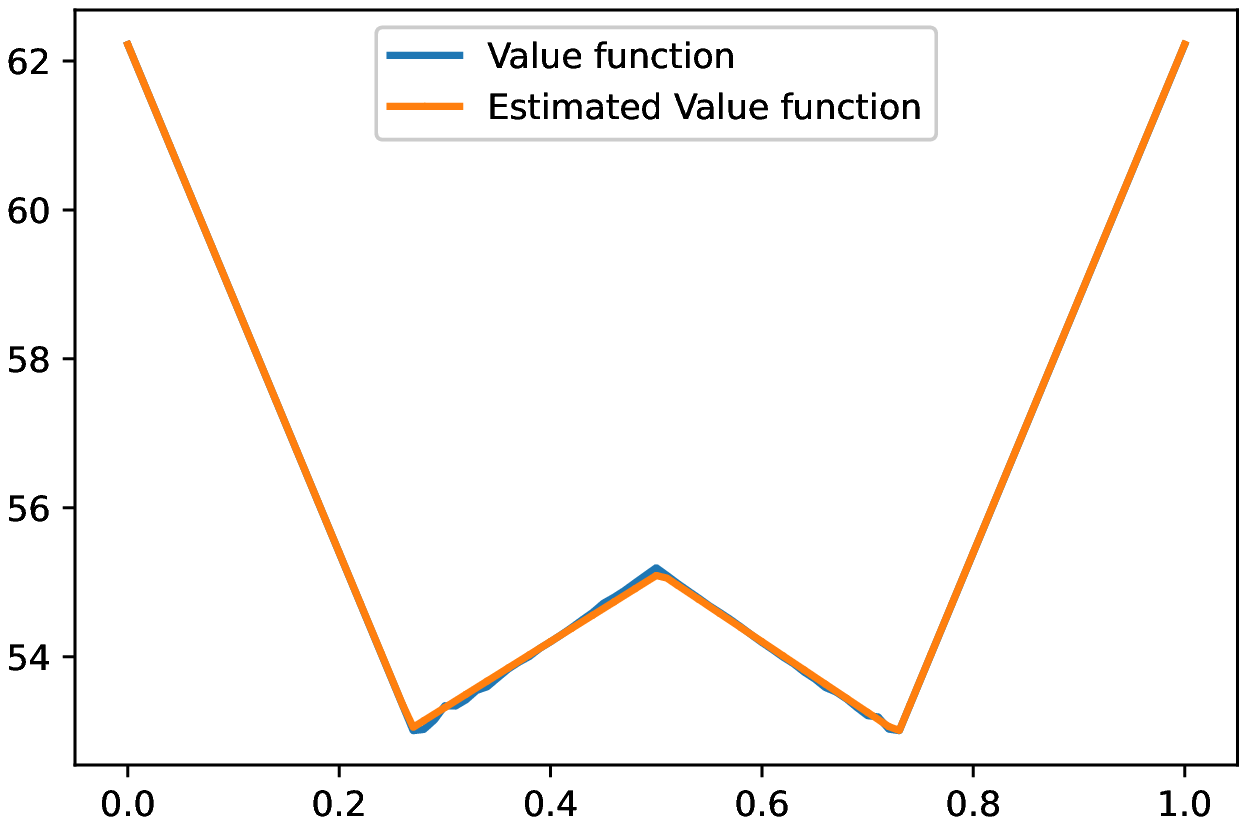}}
\caption{Estimated value function $\hat{V}_{6}$}
\label{icml-historical}
\end{center}
\vskip -0.2in
\end{figure}

\begin{figure}[ht]
%\vskip 0.2in
\begin{center}
\centerline{\includegraphics[width=0.7\columnwidth]{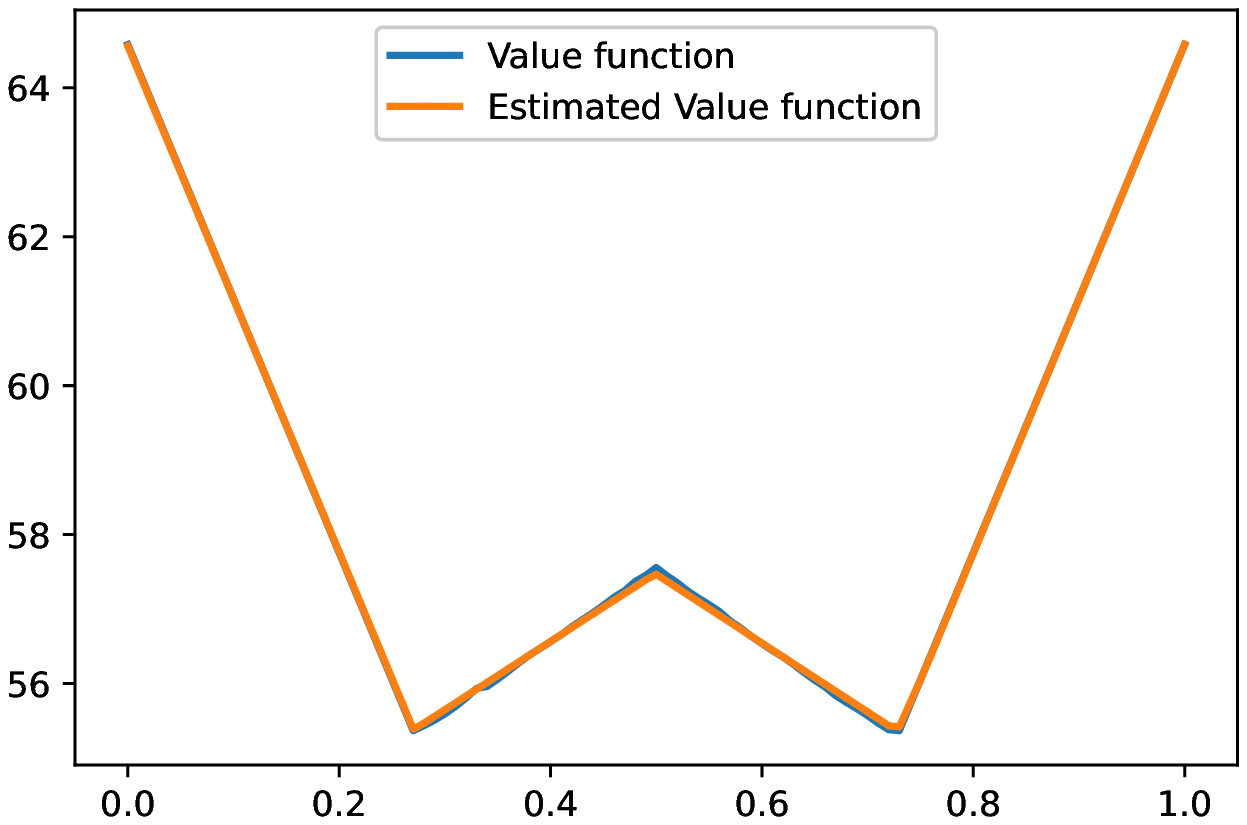}}
\caption{Estimated value function $\hat{V}_{5}$}
\label{icml-historical}
\end{center}
\vskip -0.2in
\end{figure}

\begin{figure}[ht]
%\vskip 0.2in
\begin{center}
\centerline{\includegraphics[width=0.7\columnwidth]{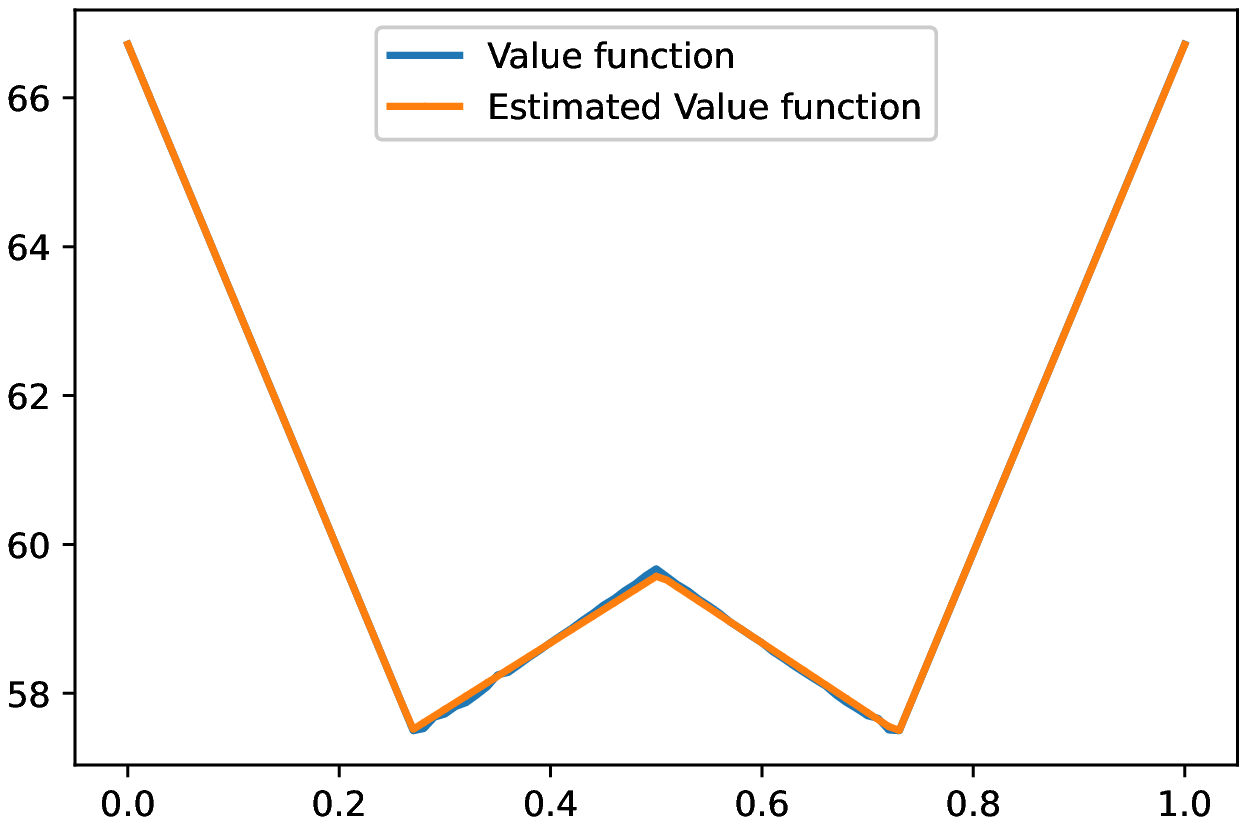}}
\caption{Estimated value function $\hat{V}_{4}$}
\label{icml-historical}
\end{center}
\vskip -0.2in
\end{figure}

\begin{figure}[ht]
%\vskip 0.2in
\begin{center}
\centerline{\includegraphics[width=0.7\columnwidth]{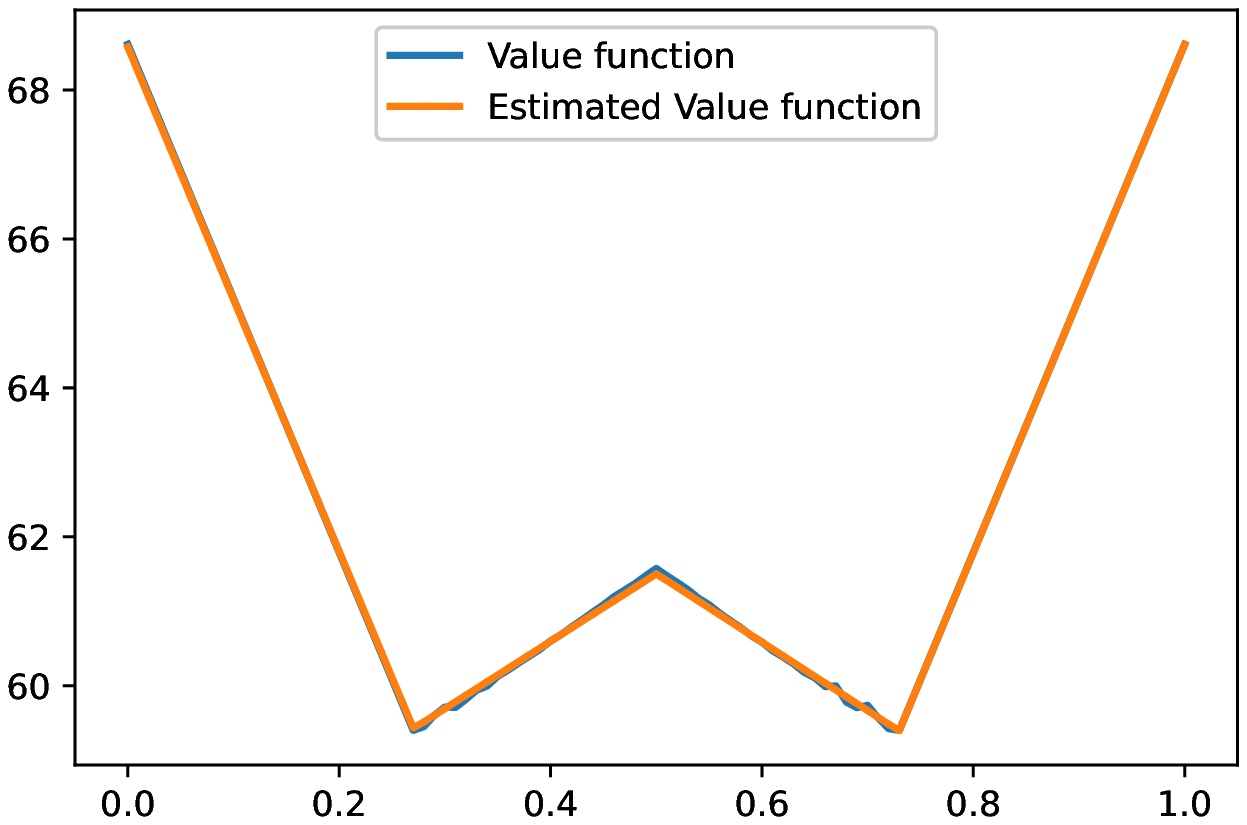}}
\caption{Estimated value function $\hat{V}_{3}$}
\label{icml-historical}
\end{center}
\vskip -0.2in
\end{figure}

\begin{figure}[ht]
%\vskip 0.2in
\begin{center}
\centerline{\includegraphics[width=0.7\columnwidth]{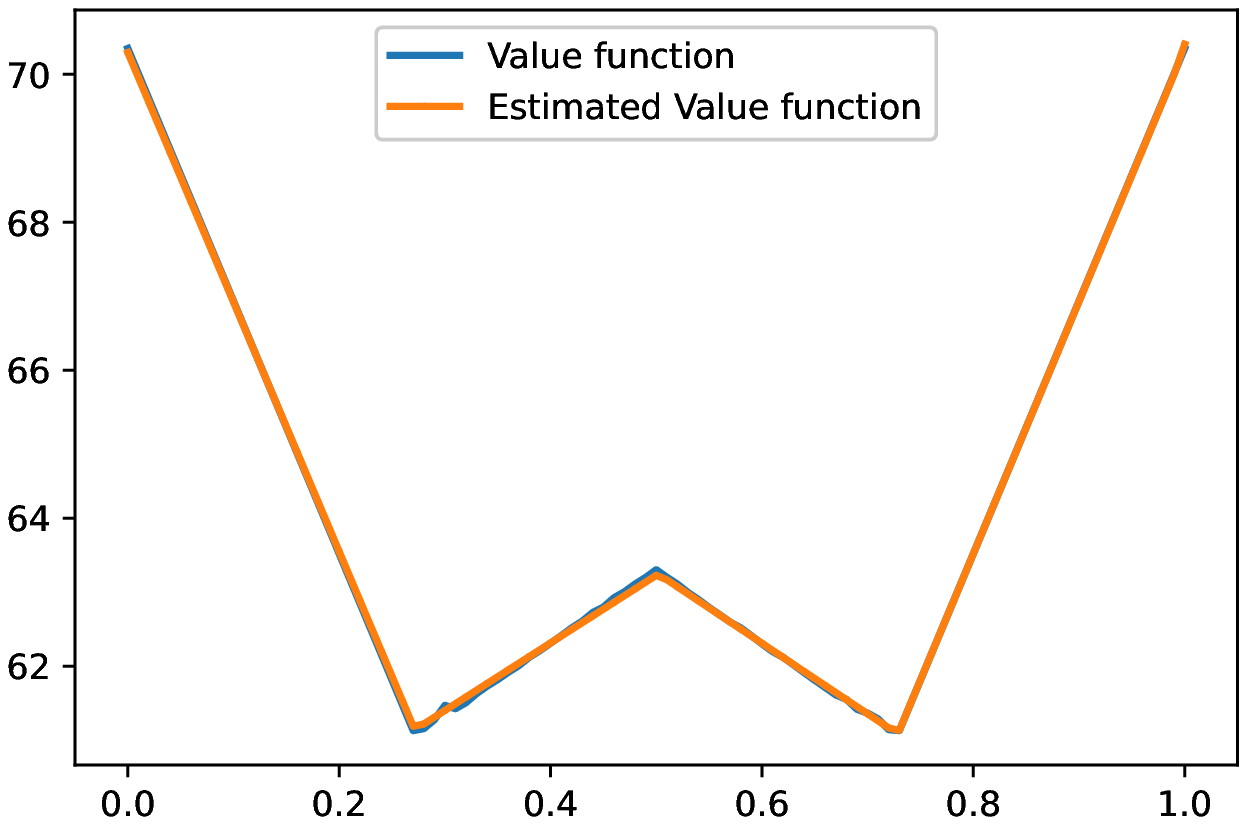}}
\caption{Estimated value function $\hat{V}_{2}$}
\label{icml-historical}
\end{center}
\vskip -0.2in
\end{figure}

\begin{figure}[ht]
%\vskip 0.2in
\begin{center}
\centerline{\includegraphics[width=0.7\columnwidth]{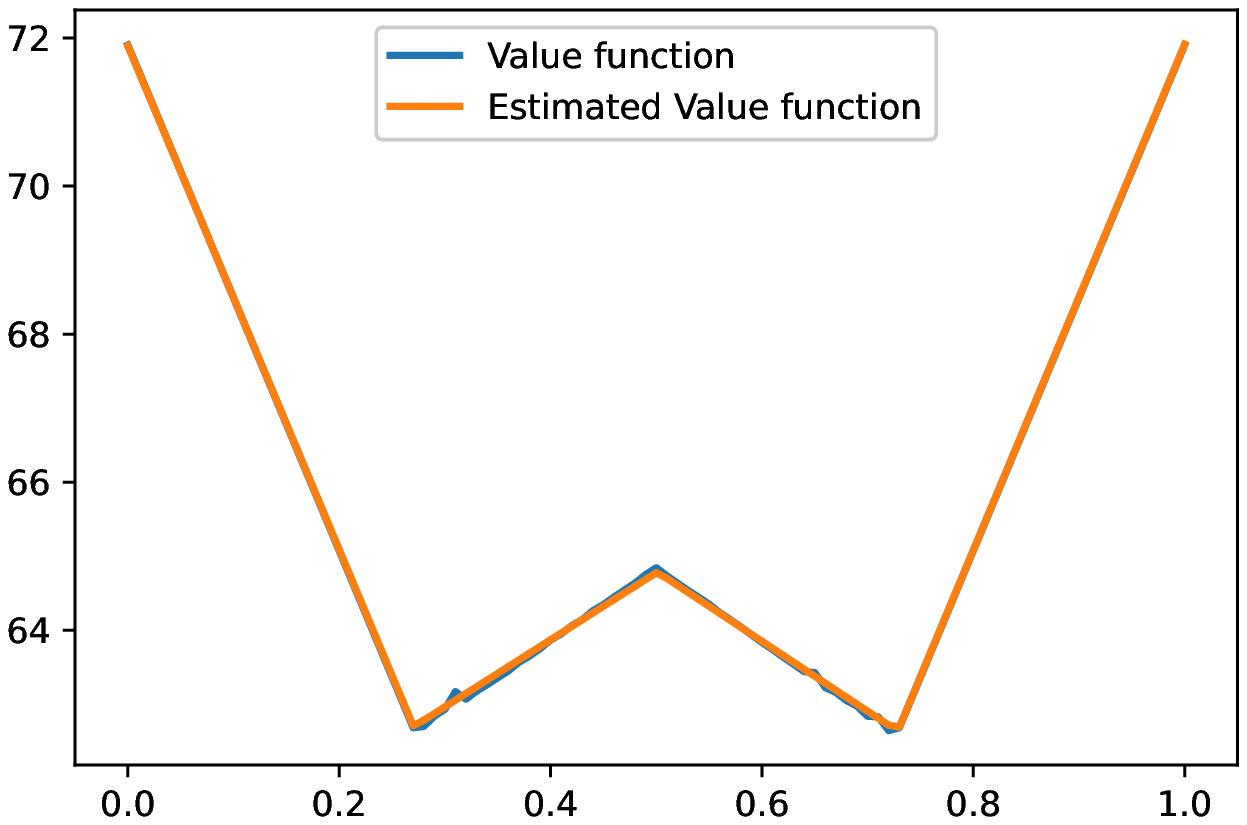}}
\caption{Estimated value function $\hat{V}_{1}$}
\label{icml-historical}
\end{center}
\vskip -0.2in
\end{figure}

\end{document}